\newcommand{\longversion}[1]{#1}
\newcommand{\shortversion}[1]{}
\newcommand{\springerversion}[1]{}
\newcommand{\arxivversion}[1]{#1}

\arxivversion{
 \documentclass[]{elsarticle}%
  \usepackage{lineno,hyperref}
\bibliographystyle{elsarticle-num}
}
\springerversion{\documentclass[runningheads,nvcountsame,a4paper,USenglish]{llncs}}

\newcommand{\FIX}[1]{#1} 

\usepackage{multirow}
\usepackage{amssymb,amsfonts,amsmath}%
\usepackage{xspace}
\usepackage[ruled,vlined,linesnumbered]{algorithm2e}
\usepackage[table,x11names]{xcolor}
\usepackage{graphicx}
\renewcommand*{\algorithmcfname}{Listing}
\SetKwInput{KwData}{In}
\SetKwInput{KwResult}{Out}
\SetAlFnt{\small}
\SetAlCapFnt{\small}
\SetAlCapNameFnt{\small}
\SetAlCapHSkip{0pt}
\IncMargin{-1em}
\SetEndCharOfAlgoLine{}
\SetInd{.2em}{1em}

\usepackage{microtype}
\usepackage{tikz}
\usetikzlibrary{shapes}
\usetikzlibrary{calc}
\usetikzlibrary{positioning}
\tikzstyle{tdnode} = [draw,rounded corners,top color=vertexTopColor,bottom color=vertexBottomColor,minimum size=1.5em]
\tikzstyle{stdnode} = [tdnode, font=\scriptsize]
\tikzstyle{stdnodecompact} = [stdnode, inner sep = 1.5pt, outer sep = 0.1pt]
\tikzstyle{stdnodetable} = [stdnode, inner sep = 1.5pt, outer sep = 0]
\tikzstyle{stdnodenum} = [minimum size=1.5em, font=\scriptsize]
\tikzstyle{tdedge} = [-,draw,thick]
\tikzstyle{tdlabel} = [draw=none, rectangle, fill=none, inner sep=0pt, font=\scriptsize]
\colorlet{vertexTopColor}{white}
\colorlet{vertexBottomColor}{black!10}

\usepackage{ctable}
\usepackage{ifthen}
\newif\iflong
\usepackage{todonotes}

\usepackage{float}

\newcommand{\PMC}{\textsc{PMC}\xspace}%
\newcommand{\AlgA}{\algo{A}}%
\newcommand{\PROJ}{\algo{PROJ}\xspace}
\newcommand{\var}{\text{\normalfont var}}
\newcommand{\bigO}[1]{\ensuremath{{\mathcal O}(#1)}}

\newcommand{\Prev}[0]{{\it{PP}}\hy Tabs\xspace}

\usepackage{array}
\newcolumntype{H}{>{\setbox0=\hbox\bgroup}c<{\egroup}@{}}

\DeclareMathOperator{\bucket}{=_P}%
\DeclareMathOperator{\buckets}{EqClasses}
\DeclareMathOperator{\subbuckets}{sub\hy buckets}
\DeclareMathOperator{\children}{children}
\newcommand{\childrenseq}{\children}
\DeclareMathOperator{\pcnt}{pmc}
\DeclareMathOperator{\sipmc}{s-ipmc}
\DeclareMathOperator{\orig}{Origins}
\DeclareMathOperator{\origs}{Origins}
\newcommand{\origse}[1]{\operatorname{Origins}}

\DeclareMathOperator{\local}{local}

\DeclareMathOperator{\Ext}{Exts}
\DeclareMathOperator{\Exts}{Exts}
\DeclareMathOperator{\PExt}{SatExt}

\DeclareMathOperator{\pmc}{pmc}
\DeclareMathOperator{\ipmc}{ipmc}
\DeclareMathOperator{\poly}{poly}
\DeclareMathOperator{\icnt}{ipmc}

\newcommand{\Tab}[1]{\ensuremath{\text{Child-Tabs}}}

\def\thyph{\text{-}\penalty0\hskip0pt\relax}
\newcommand{\ATabs}[2]{\ensuremath{#1\thyph\text{Comp[$#2$]}}}
\newcommand{\ATab}[1]{\ensuremath{#1\thyph\text{Comp}}}

\springerversion{\spnewtheorem{observation}{Observation}{\bfseries}{\itshape}}
\springerversion{\renewenvironment{example}{\begin{EXa}}{\hfill\ensuremath{\blacksquare}\end{EXa}}
\spnewtheorem{EXa}{Example}{\bfseries}{\normalfont}

\renewenvironment{proof}{\begin{pf}}{\qed\end{pf}}}

\arxivversion{\usepackage{amsthm}
\newtheorem{observation}{Observation}
\newtheorem{example}{Example}
\newtheorem{definition}{Definition}
\newtheorem{theorem}{Theorem}
\newtheorem{remark}{Remark}
\newtheorem{corollary}{Corollary}
\newtheorem{proposition}{Proposition}
\newtheorem{lemma}{Lemma}

}

\newenvironment{restateobservation}[1][\unskip]{%
  \begingroup
  \def\theobservation{\ref{#1}} 
}%
{%
  \addtocounter{observation}{-1}
  \endgroup
}%

\newenvironment{restatecorollary}[1][\unskip]{%
  \begingroup
  \def\thecorollary{\ref{#1}} 
}%
{%
  \addtocounter{corollary}{-1}
  \endgroup
}%

\newenvironment{restatetheorem}[1][\unskip]{%
  \begingroup
  \def\thetheorem{\ref{#1}} 
}%
{%
  \addtocounter{theorem}{-1}
  \endgroup
}%

\makeatletter
\newcommand{\algorithmfootnote}[2][\footnotesize]{
  \let\old@algocf@finish\@algocf@finish
  \def\@algocf@finish{\old@algocf@finish
    \leavevmode\rlap{\begin{minipage}{\linewidth}
    #1#2
    \end{minipage}}
  }
}
\makeatother

\graphicspath{{./figures/}}
\DeclareGraphicsExtensions{.pdf,.png,.jpg}
\springerversion{\bibliographystyle{splncs03}
\toctitle{Lecture Notes in Computer Science}
\tocauthor{Authors' Instructions}}

\title{Solving Projected Model Counting by Utilizing Treewidth and its
  Limits%
}
\usepackage{url}
\urldef{\mailsa}\path|{fichte, hecher,morak,woltran}@dbai.tuwien.ac.at|
\springerversion{\titlerunning{Exploiting Treewidth for Projected Model Counting \& Limits}}

\springerversion{
\author{Johannes K. Fichte %
\and Markus Hecher \and Michael Morak \and Stefan Woltran%
}}

\springerversion{\institute{Institute of Logic and Computation, TU Wien, Vienna, Austria\\
  \email{lastname@dbai.tuwien.ac.at}\\
}
\authorrunning{Fichte et al.}}

\newenvironment{indented}{\begin{changemargin}{1cm}{0cm}}{\end{changemargin}}

\let\phi\varphi
\let\epsilon\varepsilon

\renewcommand{\models}{\vDash}

\newcommand{\calT}{\mathcal{T}}

\newcommand{\ta}[1]{\ensuremath{2^{#1}}}
\newcommand{\card}[1]{\left|#1\right|}
\newcommand{\CCard}[1]{\|#1\|}

\newcommand{\Nat}{\mathbb{N}} 

\newcommand{\algo}[1]{\ensuremath{\mathbb{#1}}}

\newcommand{\NP}{\ensuremath{\textsc{NP}}\xspace}

\newcommand{\PSPACE}{\ensuremath{\textsc{PSpace}}}

\newcommand{\tw}[1]{\mathit{tw}(#1)}

\newcommand{\SB}{\{}%
\newcommand{\SM}{\mid}%
\newcommand{\SE}{\}}%
\def\hy{\hbox{-}\nobreak\hskip0pt}

\newcommand{\solverix}[1]{\mbox{\text{#1}}\xspace}

\newcommand{\sharpsat}{\solverix{SharpSAT}}
\newcommand{\dynasp}[1]{\ensuremath{\solverix{DynASP2}}}
\newcommand{\dynaspplus}[1]{\ensuremath{\solverix{DynASP2.5}}}
\newcommand{\prog}{\ensuremath{F}}

\DeclareMathOperator{\width}{width}
\newcommand{\algNS}{{\algo{N}\algS}\xspace}%
\newcommand{\algS}{\AlgS}%
\newcommand{\algEST}[1]{{{\algo{PMC}}}\xspace}%
\newcommand{\algNES}{\algo{N}\algEST{t}}%
\newcommand{\algHES}{\algo{H}\algEST{t}}%

\DeclareMathOperator{\depth}{depth}
\DeclareMathOperator{\cache}{cache}

\newcommand{\hdpa}{\ensuremath{\mathtt{HybDP}}\xspace}
\newcommand{\adpa}{\ensuremath{\mathtt{NestDP}}\xspace}
\DeclareMathOperator{\compat}{comp}
\newcommand{\primal}[1]{\ensuremath{G_{#1}}}
\newcommand{\nested}[2]{\ensuremath{G_{#1}^{#2}}}

\newcommand{\nesthdb}{\textsf{nestHDB}\xspace}
\newcommand{\solver}{\nesthdb}
\newcommand{\solversc}{\textsf{nestHDB(sc)}\xspace}
\newcommand{\countAntom}{\textsf{countAntom}\xspace}
\newcommand{\gpusat}{\textsf{gpusat2}\xspace}

\newcommand{\dpdb}{\textsf{dpdb}\xspace}
\newcommand{\clingo}{\textsf{clingo}\xspace}
\newcommand{\htd}{\textsf{htd}\xspace}

\newcommand{\minic}{\textsf{minic2d}\xspace}
\newcommand{\picosat}{\textsf{picosat}\xspace}
\newcommand{\projmc}{\textsf{projMC}\xspace}
\newcommand{\ganak}{\textsf{ganak}\xspace}

\DeclareMathOperator{\rootOf}{root}

\newcommand{\pmcp}{\textsf{pmc}\xspace}

\newcommand{\ASP}{\textsc{ASP}\xspace}
\newcommand{\cSAT}{\textsc{\#Sat}\xspace}
\newcommand{\sharpSAT}{\cSAT}
\newcommand{\cESAT}{\PMC}

\makeatletter
\DeclareRobustCommand{\rvdots}{%
  \vbox{
    \baselineskip4\p@\lineskiplimit\z@
    \kern-\p@
    \hbox{.}\hbox{.}\hbox{.}
  }}
\makeatother

\DeclareFontFamily{U}{matha}{\hyphenchar\font45}
\DeclareFontShape{U}{matha}{m}{n}{
      <5> <6> <7> <8> <9> <10> gen * matha
      <10.95> matha10 <12> <14.4> <17.28> <20.74> <24.88> matha12
      }{}
\DeclareSymbolFont{matha}{U}{matha}{m}{n}
\DeclareMathSymbol{\squplus}{2}{matha}{"5D}

\makeatletter

\newcommand{\raisemath}[1]{\mathpalette{\raisem@th{#1}}}
\newcommand{\raisem@th}[3]{\raisebox{#1}{$#2#3$}}

\newcommand{\pushright}[1]{\ifmeasuring@#1\else\omit\hfill\ensuremath{\displaystyle#1}\fi\ignorespaces}
\newcommand{\pushleft}[1]{\ifmeasuring@#1\else\omit$\displaystyle#1$\hfill\fi\ignorespaces}
\providecommand{\leftsquigarrow}{%
	\mathrel{\mathpalette\reflect@squig\relax}%
}
\newcommand{\reflect@squig}[2]{%
	\reflectbox{$\m@th#1\rightsquigarrow$}%
}

\makeatother

\DeclareMathOperator{\cntc}{\#\cdot}%

\DeclareMathOperator{\type}{type}
\newcommand{\intr}{\textit{int}}
\newcommand{\leaf}{\textit{leaf}}
\newcommand{\rem}{\textit{rem}}
\newcommand{\join}{\textit{join}}

\DeclareMathOperator{\post}{post-order}
\DeclareMathOperator{\dom}{dom}

\let\P\undefined
\DeclareMathOperator{\P}{P}

\newcommand{\BIGOP}[1]
{
\mathop{\mathchoice%
{\raise-0.22em\hbox{\huge $#1$}}%
{\raise-0.05em\hbox{\Large $#1$}}{\hbox{\large $#1$}}{#1}}}

\newcommand{\BIGboxplus}{\mathop{\mathchoice%
{\raise-0.35em\hbox{\huge $\boxplus$}}%
{\raise-0.15em\hbox{\Large $\boxplus$}}{\hbox{\large $\boxplus$}}{\boxplus}}}

\newcommand{\TTT}{\ensuremath{\mathcal{T}}}%
\newcommand{\WWW}{\ensuremath{\mathcal{W}}}%

\newcommand*\mcupinn[2]{\vcenter{\hbox{$\mathsurround=0pt
  \ifx\displaystyle#1\textstyle\else#1\fi\bigcup$}}}

\newcommand{\NAT}{\ensuremath{\mathbb{N}}}

\newcommand{\inputPredColor}{orange!55!red}
\newcommand{\outputPredColor}{blue!45!black}
\newcommand{\statePredColor}{green!62!black}
\newcommand{\specialPredColor}{red!62!black}

\newcommand{\tuplecolor}[1]{\textcolor{#1}}

\newcommand{\tabval}{\ensuremath{u}}
\newcommand{\tab}[1]{\ensuremath{\tau_{#1}}}

\newcommand{\progt}[1]{\ensuremath{\prog_{\hspace{-0.05em}\leq\hspace{-0.05em}#1}}}

\newcommand{\dpa}{\ensuremath{\mathtt{DP}}\xspace}

\newcommand{\mdpa}[1]{\ensuremath{\mathtt{PCNT}_{#1}}}

\newcommand{\eqdef}{\ensuremath{\,\mathrel{\mathop:}=}}

\newcommand{\Card}[1]{|#1|}
\renewcommand{\P}{\text{\normalfont P}\xspace}

\newcommand{\PRIM}{\AlgS} 

\newcommand{\SAT}{\textsc{Sat}\xspace}

\newcommand{\QBFSAT}{\textsc{QSat}\xspace}

\newcommand{\Q}{\ensuremath{Q}}

\newcommand{\INC}{\ensuremath{{\algo{INC}}}\xspace}

\newcommand{\problemFont}[1]{\textsc{#1}}

\newlength\problemlength
\newcommand\dproblem[3]{%
\begin{center}
\fbox{%
\begin{minipage}{.93	\linewidth}%
\begin{list}{}{\labelwidth\problemlength \labelsep.7em \rightmargin1.5em
\leftmargin\problemlength \advance\leftmargin by3em
\parsep0ex \itemsep.2ex plus.1ex}
\item[{\sl Problem:\hfill}] {\problemFont{#1}}
\item[{\sl Input:  \hfill}] #2
\item[{\sl Task: \hfill}] #3
\end{list}
\end{minipage}
}
\end{center}
}

\newcommand{\AlgS}{\algo{SAT}\xspace}%
\begin{document}
\arxivversion{\begin{frontmatter}


  \author[liu]{Johannes K. Fichte}%
  \ead{johannes.klaus.fichte@liu.se}
  \author[mit]{Markus Hecher\corref{cor1}}%
  \ead{hecher@mit.edu}
  \author[klag]{Michael Morak}%
  \ead{michael.morak@aau.at}
  %
  \author[vieA]{Patrick Thier}%
  \ead{thier@tuwien.ac.at}
  \author[vieA]{Stefan Woltran}%
  \ead{woltran@dbai.tuwien.ac.at}
  %
  %
  \address[liu]{AIICS, IDA, Link\"oping University,\\
    581 83 Link\"oping, Sweden\smallskip}%
  %
  \address[mit]{Computer Science and Artificial Intelligence Lab, Massachusetts Institute of Technology,\\
    32 Vassar St., Cambridge, MA, United States\smallskip}%
  \address[vieA]{Database and Artificial Intelligence Group, TU Wien,\\
    Favoritenstrasse 9-11, 1040 Vienna, Austria\smallskip}%
  \address[klag]{Department of Artificial Intelligence and Cybersecurity, University of Klagenfurt,\\
    Universit\"atsstra{\ss}e 65-67, 9020 Klagenfurt am W\"orthersee, Austria\smallskip}%
  %
  %
  %
  \cortext[cor1]{Corresponding author.}

  \date{\today} 

\begin{abstract}%
  In this paper, we introduce a novel algorithm to solve
  \emph{projected model counting} (\PMC). \PMC asks to count solutions
  of a Boolean formula with respect to a given set of \emph{projection
    variables}, where multiple solutions that are identical when
  restricted to the projection variables count as only one solution.
  Inspired by the observation that the so-called ``treewidth'' is one of the 
most prominent structural parameters,
  our algorithm utilizes small treewidth of the primal 
  graph of the input instance. More precisely, it runs in time~$\bigO{2^{2^{k+4}} n^2}$
  where $k$ is the treewidth and $n$ is the input size of the
  instance. In other words, we obtain that the problem~\PMC is
  fixed-parameter tractable when parameterized by treewidth.  Further,
  we take the exponential time hypothesis (ETH) into consideration and
  establish lower bounds of bounded treewidth algorithms for \PMC,
  yielding asymptotically tight runtime bounds of our algorithm.

While the algorithm above serves as a first theoretical upper bound
and although it might be quite appealing for small values of~$k$,
unsurprisingly a naive implementation adhering to this runtime bound
 suffers already from instances of relatively small width.
Therefore, 
	we turn our attention to several measures in order to
resolve this issue towards exploiting treewidth in practice:
We present a technique called nested dynamic programming,
where different levels of abstractions of the primal graph are used to (recursively) compute and refine
tree decompositions of a given instance.
Further, we 
integrate the concept of hybrid solving,
where subproblems hidden by the abstraction 
are solved by classical search-based solvers, 
which leads to an interleaving of
parameterized and classical solving. 
Finally, we provide a nested dynamic programming algorithm and an implementation that relies on database technology
for PMC and a prominent special case of PMC, namely model counting (\cSAT).
Experiments indicate that the advancements are promising, allowing us to solve instances of treewidth upper bounds beyond 200.
\end{abstract}

\end{frontmatter}}


\section{Introduction}\label{sec:introduction}
A problem that has been used to solve a large variety of real-world
questions is the \emph{model counting problem}
(\cSAT)~\cite{AbramsonBrownEdwards96a,ChoiBroeckDarwiche15a,DomshlakHoffmann07a,MeelEtAl17a,ManningRaghavanSchutze08a,PourretNaimBruce08a,SahamiDumaisHeckerman98a,SangBeameKautz05a,XueChoiDarwiche12a}.
It asks to compute the number of solutions of a Boolean
formula~\cite{GomesKautzSabharwalSelman08a} and is theoretically of
high worst-case complexity
($\cntc\P$-complete~\cite{Valiant79,Roth96a}). Lately, both \cSAT and
its approximate version have received renewed attention in theory and
practice~\cite{ChakrabortyMeelVardi16a,MeelEtAl17a,LagniezMarquis17a,SaetherTelleVatshelle15a}.
A concept that allows very natural abstractions of data and query
results is projection. Projection has wide applications in
databases~\cite{AbiteboulHullVianu95} and declarative problem
modeling.  
The problem \emph{projected model counting} (\PMC) asks to count
solutions of a Boolean formula with respect to a given set of
\emph{projection variables}, 
where multiple solutions that are identical when restricted to the
projection variables count as only one solution.
If all variables of the formula are projection variables, then \PMC is
the \cSAT problem and if there are no projection variables then it is
simply the \SAT problem.
Projected variables allow for solving problems where one needs to
introduce auxiliary variables, in particular, if these variables are
functionally independent of the variables of interest, in the problem
encoding,~e.g.,~\cite{GebserSchaubThieleVeber11,GinsbergParkesRoy98a}.
Projected model counting is a fundamental problem in artificial intelligence
and was also subject to a dedicated track in the first model counting competition~\cite{FichteHecherHamiti20}.
It turns out that there are plenty of use cases and applications for \PMC,
ranging from a variety of real-world
questions in modern society, %
artificial intelligence~\cite{LagniezMarquis19}, reliability estimation~\cite{MeelEtAl17a} and
combinatorics~\cite{AzizChuMuise15a}.
Variants of this problem are relevant to problems in probabilistic and quantitative reasoning, e.g.,~\cite{ChoiBroeckDarwiche15a,DomshlakHoffmann07a,XueChoiDarwiche12a} and Bayesian reasoning~\cite{SangBeameKautz05a}.
\FIX{This work also inspired follow-up work, as 
extensions of projected model counting as well as generalizations
for logic programming and
quantified Boolean formulas have been presented recently, e.g.,~\cite{CapelliMengel19,FichteHecher19,DudekPanVardi21}.}

%
%

When we consider the computational complexity of \PMC it turns out
that under standard assumptions the problem is even harder than \cSAT,
more precisely, complete for the class
$\cntc\NP$~\cite{DurandHermannKolaitis05}.
Even though there is a \PMC solver~\cite{AzizChuMuise15a} and an \ASP
solver that implements projected
enumeration~\cite{GebserKaufmannSchaub09a}, \PMC has received very
little attention in parameterized algorithmics so far.
Parameterized
algorithms~\cite{CyganEtAl15,DowneyFellows13,FlumGrohe06,Niedermeier06}
tackle computationally hard problems by directly exploiting certain
structural properties (parameter) of the input instance to solve the
problem faster, preferably in polynomial-time for a fixed parameter
value.
%
%
%
In this paper, we consider the treewidth of graphs associated with the
given input formula as parameter, namely the primal 
graph~\cite{SamerSzeider10b}.
Roughly speaking, small \emph{treewidth} of a graph measures its
tree-likeness and sparsity. Treewidth is defined in terms of
\emph{tree decompositions (TDs)}, which are arrangements of graphs
into trees.
%
When we take advantage of small treewidth, we usually take a TD and
evaluate the considered problem in parts, via \emph{dynamic
  programming~(DP)} on the TD. 
This dynamic programming technique utilizes tree decompositions,
where a tree decomposition is traversed in post-order, i.e., from the leaves towards the root, and thereby for each node of the TD tables are computed such that a problem is solved by cracking smaller (partial) problems.

In this work we apply tree decompositions for projected model counting \FIX{and study precise \emph{runtime dependency on treewidth}.
While there are also related works on properties for efficient counting algorithms, e.g.,~\cite{DurandMengel13,ChenMengel17,GrecoScarcello17}, even for treewidth, precise runtime dependency for projected model counting has been left open.} 
We design a novel algorithm that runs in \emph{double exponential time}\footnote{Runtimes that are double exponential in the treewidth indicates expressions of the form~$2^{2^{\mathcal{O}(k)}}\cdot\poly(n)$, where~$n$ indicates the number of variables of a given formula and~$k$ refers to the treewidth of its primal graph.} in the treewidth, but it is quadratic in the number of variables of a given formula.
Later, we also establish a conditional lower bound showing that under reasonable assumptions it is quite \emph{unlikely that one can significantly improve} this algorithm.

Naturally, it is expected that our proposed \PMC algorithm can be only competitive for instances where the treewidth is very low.
Still, despite our new theoretical result, it turns out that in 
practice there is a way to efficiently implement dynamic programming and tree decompositions for solving \PMC.
However, most of the existing systems based on dynamic programming guided along a tree decomposition are suffering from maintaining large tables, since the size of these tables (and thus the computational efforts required) are 
bounded by a function in the treewidth of the instance. 
Although dedicated competitions~\cite{Dell17a} for treewidth
advanced the state-of-the-art 
for efficiently computing 
treewidth and TDs~\cite{AbseherMusliuWoltran17a,Tamaki19
},
these systems and approaches 
reach their limits when instances have higher treewidth.
Indeed, such approaches based on dynamic programming 
reach their limits when instances have higher treewidth; a situation which 
can even occur in structured real-world instances~\cite{ManiuSenellartJog2019}.
Nevertheless in the area of Boolean satisfiability, this approach proved to be successful for counting problems, such as, e.g., (weighted) model counting~\cite{FichteEtAl20,FichteHecherZisser19,SamerSzeider10b}.
To further 
increase the practical applicability of dynamic programming for \PMC, novel techniques are required, where we 
rely on certain simplifications of a graph, which we call \emph{abstraction}\footnote{A formal account on these abstractions will be given in Definition~\ref{def:nestprimalgraph}.}.
Thereby, we (a) rely on different \emph{levels of abstraction} of the instance at hand;
(b) 
\emph{treat subproblems} orginating in the abstraction 
by standard solvers
whenever widths appear too high; 
and (c) use highly \emph{sophisticated data management} 
in order to store and process tables obtained by dynamic programming.

\medskip
\paragraph{Contributions}
In more details, we provide the following contributions.
%
\begin{enumerate}
\item We introduce a novel algorithm to \emph{solve projected model
    counting} in time~$\bigO{2^{2^{k+4}} n^2}$ where $k$ is the
  treewidth of the primal 
  graph of the instance and $n$ is the size of the input instance.
  Similar to recent DP algorithms for problems on the second level of
  the polynomial hierarchy~\cite{FichteEtAl17b}, our algorithm
  traverses the given tree decomposition multiple times (multi-pass).
  In the first traversal, we run a dynamic programming algorithm on
  tree decompositions to solve \SAT~\cite{SamerSzeider10b}. In a
  second traversal, we construct equivalence classes on top of the
  previous computation to obtain model counts with respect to the
  projection variables by exploiting combinatorial properties of
  intersections.
\item Then, we establish that our \emph{runtime bounds are asymptotically tight under the
    exponential time hypothesis (ETH)}~\cite{ImpagliazzoPaturiZane01}
  using a recent result by Lampis and Mitsou~\cite{LampisMitsou17},
  who established lower bounds for the
  problem~$\exists\forall$\hy\SAT assuming ETH.
  Intuitively, ETH states a complexity theoretical lower bound on how
  fast satisfiability problems can be solved. More precisely, one
  \emph{cannot} solve 3\hy\SAT in
  time~$2^{s\cdot n}\cdot n^{\bigO{1}}$ for some~$s>0$ and number~$n$
  of variables.
  \item Finally, we also provide an implementation for \PMC that efficiently utilizes treewidth and is highly competitive with state-of-the-art solvers.
  In more details, we treat above aspects (a), (b), and (c) as follows.
  \begin{enumerate} 
\item[(a)] To tame the beast of high treewidth, we propose 
	\emph{nested dynamic programming},
where only parts of some abstraction of a graph are decomposed.
Then, each TD node also needs to solve a \emph{subproblem}
residing in the graph, but may involve vertices outside the abstraction.
In turn, for solving such subproblems, the idea of nested DP is 
		to subsequently repeat decomposing and solving more fine-grained graph abstractions in a nested fashion.%
%
	\shortversion{This results not only in elegant DP algorithms, but also allows to
	deal with high treewidth.}%
  While candidates for obtaining such abstractions often naturally originate from the problem \PMC, nested DP may require computing those during nesting, 
for which we even present a generic solution.
%
	\item[(b)] To further improve the capability of handling high treewidth,
	we show how to apply nested DP in the context of \emph{hybrid solving},
	where established, standard solvers (e.g., \SAT solvers) and caching are incorporated in 
	nested DP such that the best of two worlds are combined.
Thereby, we solve counting problems like \PMC, where we apply DP to parts of the problem instance that are \emph{subject to counting}, while depending on the existence of a solution for certain subproblems. 
Those subproblems that are \emph{subject to searching} for the existence of a solution reside in the abstraction only and are solved via standard solvers.
%
	\item[(c)] We implemented a system based on a recently published tool~\cite{FichteEtAl20} for using database management systems (DBMS) to efficiently perform table manipulation operations needed during DP. 
	Our system is called \solver{}\footnote{\solver{} is open-source and available at \href{https://github.com/hmarkus/dp\_on\_dbs/tree/nesthdb}{github.com/hmarkus/dp\_on\_dbs/tree/nesthdb}.} and uses and significantly extends this tool in order to perform hybrid solving, thereby combining nested DP and standard solvers.
		As a result, we use DBMS for efficiently implementing the handling of tables needed by nested DP.
	Preliminary experiments indicate that nested DP with hybrid solving can be fruitful, where we are capable of solving instances, whose treewidth upper bounds are beyond 200.
\end{enumerate}
\end{enumerate}
%
%

%
%

This paper combines research of work that is published at the 21st International Conference on Satisfiability (SAT 2018)~\cite{FichteEtAl18} and research that was presented at the 23rd International Conference on Satisfiability (SAT 2020)~\cite{HecherThierWoltran20}.
In addition to these conference versions, we added detailed proofs, further examples, and significantly improved the presentation throughout the document.

\section{Preliminaries}\label{sec:preliminaries}

\FIX{We assume familiarity with basic notions from set theory and on sequences. We write a sequence consisting of~$\ell$ elements~$e_i$ for~$1\leq i\leq \ell$ in angular brackets,~i.e., $\langle e_1, e_2 \ldots, e_\ell \rangle$.}
%
For a set~$X$, let $\ta{X}$ be the \emph{power set of~$X$}
consisting of all subsets~$Y$ with $\emptyset \subseteq Y \subseteq X$.
%
%
Recall the well-known combinatorial inclusion-exclusion
principle~\cite{GrahamGrotschelLovasz95a}, which states that for two
finite sets~$A$ and $B$ it is true
that~$\Card{A \cup B} = \Card{A} + \Card{B} - \Card{A \cap B}$. Later,
we need a generalized version for arbitrary many sets. Given for some
integer~$n$ a family of finite sets~$X_1$, $X_2$, $\ldots$, $X_n$,
%
the number of elements in the union
over all sets is
$\Card{\bigcup^n_{j = 1} X_j} = \sum_{I \subseteq \{1, \ldots, n\}, I
  \neq \emptyset} (-1)^{\Card{I}-1} \Card{\bigcap_{i \in I}
  X_i}$. 

\paragraph{Satisfiability}
A literal is a (Boolean) variable~$x$ or its negation~$\neg x$. A
\emph{clause} is a finite set of literals, interpreted as the
disjunction of these literals.
%
%
A \emph{(CNF) formula} is a finite set of clauses, interpreted as the
conjunction of its clauses.  A 3\hy CNF has clauses of length at
most~3. 
Let $F$ be a formula.  A \emph{sub-formula~$S$} of~$F$ is a
subset~$S\subseteq F$ of~$F$.  For a clause~$c \in F$, we let
$\var(c)$ consist of all variables that occur in~$c$ and
$\var(F)\eqdef\bigcup_{c \in F} \var(c)$.  An \emph{assignment} is a mapping $\alpha: V \rightarrow \{0,1\}$ for a set~$V\subseteq\var(F)$ of variables.
For $x\in V,$ we define $\alpha(\neg x) \eqdef 1 - \alpha(x)$.
The formula~$F$ \emph{under an assignment~$\alpha$} 
is the formula~$F[\alpha]$ obtained from~$F$ by removing all
clauses~$c$ containing a literal set to~$1$ by $\alpha$ and removing
from the remaining clauses all literals set to~$0$ by $\alpha$. An
assignment~$\alpha$ is \emph{satisfying} if $F[\alpha]=\emptyset$, denoted by~$\alpha \models F$.
Then, $F$ is \emph{satisfiable} if there is such a satisfying
assignment~$\alpha$, otherwise we say~$F$ is \emph{unsatisfiable}. %
Let $V$ be a set of variables. An \emph{interpretation} is a
set~$J\subseteq V$ and its induced assignment~$\alpha_{J,V}$ of~$J$
with respect to $V$ is defined as
follows~$\alpha_{J,V} \eqdef \SB v \mapsto 1 \SM v \in J \cap V \SE
\cup \SB v \mapsto 0 \SM v \in V \setminus J \SE$.
We simply write $\alpha_{J}$ for $\alpha_{J,V}$ if $V=\var(F)$.
An interpretation~$J$ is a \emph{model} of~$F$ if its
induced assignment~$\alpha_J$ is satisfying, i.e., $\alpha_J\models F$.
%
%
%
%
Given a formula~$F$; the problem \SAT asks whether $F$ is satisfiable
and the problem \cSAT asks to output the number of models
of~$F$,~i.e., $\Card{S}$ where $S$ is the set of all models of~$F$.


\paragraph{Projected Model Counting} %
An instance of the projected model counting problem is a pair~$(F,P)$
where $F$ is a (CNF) formula and $P$ is a set of Boolean variables
such that $P \subseteq\var(F)$.  We call the set~$P$ \emph{projection
  variables} of the instance. The \emph{projected model count} of a
formula~$F$ with respect to~$P$ is the number of total
assignments~$\alpha$ to variables in~$P$ such that the
formula~$F[\alpha]$ under~$\alpha$ is satisfiable.
The \emph{projected model counting problem
  (\PMC)}~\cite{AzizChuMuise15a} asks to output the projected model
count of~$F$,~i.e., $\Card{ \SB M \cap P \SM M \in S \SE}$ where $S$
is the set of all models of~$F$.

\begin{example}\label{ex:running0}
  Consider
  formula~$F\eqdef \{\overbrace{\neg a \vee b \vee p_1}^{c_1},
  \overbrace{a\vee \neg b \vee \neg p_1}^{c_2}, \overbrace{a \vee
    p_2}^{c_3}, \overbrace{a \vee \neg p_2}^{c_4}\}$ and
  set~$P\eqdef\{p_1,p_2\}$ of projection variables.
  The models of formula~$F$ are $\{a,b\}$, $\{a,p_1\}$,
  $\{a,b,p_1\}$,$\{a,b,p_2\}$, $\{a,p_1,p_2\}$, and $\{a,b,p_1,p_2\}$.
  However, projected to the set~$P$, we only have models $\emptyset$,
  $\{p_1\}$, $\{p_2\}$, and $\{p_1,p_2\}$.
  Hence, the model count of~$F$ is 6 whereas the projected model count
  of instance~$(F,P)$ is 4.
\end{example}

\paragraph{Quantified Boolean Formulas (QBFs)}
A \emph{(prenex) quantified Boolean formula}~$\Q$ is of the
form
  $Q_1 V_1. Q_2 V_2.\ldots Q_m V_m. F$
  where $Q_i \in \{\forall, \exists\}$, $V_i$ are disjoint sets of
  Boolean variables, and $F$ is a Boolean formula that contains only
  the variables in $\bigcup^m_{i=1} V_i$.
  %
  %
  The truth (evaluation) of quantified Boolean formulas is defined in the standard way, where for~$\Q$ above if~$Q_1=\exists$, then~$\Q$ evaluates to true if and only if there
  exists an assignment~$\alpha: V_1\rightarrow \{0,1\}$ such
  that~$Q_2 V_2.\ldots Q_m V_m. F[\alpha]$ evaluates to true.  If~$Q_1=\forall$,
  then~$\Q$ evaluates to true if for any
  assignment~$\alpha: V_1 \rightarrow\{0,1\}$, we have that $Q_2 V_2.\ldots Q_m V_m. F[\alpha]$ evaluates to
  true.
  Given a quantified Boolean formula~$\Q$, the evaluation problem of
  quantified Boolean formulas~\QBFSAT asks whether $Q$ evaluates to
  true.
  The problem~\QBFSAT is \PSPACE-complete and is therefore believed to
  be computationally harder than
  \SAT~\cite{KleineBuningLettman99,Papadimitriou94,StockmeyerMeyer73}.
  A well known fragment of \QBFSAT is $\forall\exists$\hy \SAT where
  the input is restricted to quantified Boolean formulas of the
  form~$\forall V_1.\exists V_2.F$ where $F$ is a Boolean
  CNF formula. The complexity class consisting of all problems that are
  polynomial-time reducible to $\forall\exists$\hy \SAT is denoted by
  $\Pi_2^P$, and its complement is denoted by $\Sigma_2^P$. 
For more
detailed information on QBFs we refer to other sources,
e.g.,~\cite{BiereHeuleMaarenWalsh09,KleineBuningLettman99}.



%




\newcommand{\restrict}[2]{\ensuremath{#1\cap #2}}


\paragraph{Computational Complexity}
We assume familiarity with standard notions in computational
complexity~\cite{Papadimitriou94}
and use counting complexity classes as defined by Hemaspaandra and
Vollmer~\cite{HemaspaandraVollmer95a}.
%
%
%
For parameterized complexity, we refer to standard
texts~\cite{CyganEtAl15,DowneyFellows13,FlumGrohe06,Niedermeier06}.
%
%
Let $\Sigma$ and $\Sigma'$ be some finite alphabets.  We call
$I \in \Sigma^*$ an \emph{instance} and $\CCard{I}$ denotes the size
of~$I$.  
%
Let $L \subseteq \Sigma^* \times \Nat$ and
$L' \subseteq {\Sigma'}^*\times \Nat$ be two parameterized problems. An
\emph{fpt-reduction} $r$ from $L$ to $L'$ is a many-to-one reduction
from $\Sigma^*\times \Nat$ to ${\Sigma'}^*\times \Nat$ such that for all
$I \in \Sigma^*$ we have $(I,k) \in L$ if and only if
$r(I,k)=(I',k')\in L'$ such that $k' \leq g(k)$ for a fixed computable
function $g: \Nat \rightarrow \Nat$, and there is a computable function
$f$ and a constant $c$ such that $r$ is computable in time
$O(f(k)\CCard{I}^c)$~\cite{FlumGrohe06}.
A \emph{witness function} is a
function~$\mathcal{W}\colon \Sigma^* \rightarrow 2^{{\Sigma'}^*}$ that
maps an instance~$I \in \Sigma^*$ to a finite subset
of~${\Sigma'}^*$. We call the set~$\WWW(I)$ the \emph{witnesses}. A
\emph{parameterized counting
  problem}~$L: \Sigma^* \times \NAT_0 \rightarrow \Nat_0$ is a
function that maps a given instance~$I \in \Sigma^*$ and an
integer~$k \in \NAT$ to the cardinality of its
witnesses~$\card{\WWW(I)}$.
We call $k$ the \emph{parameter}.
%
%
%
The \emph{exponential time hypothesis} (ETH) states
that 
the (decision) problem~\SAT on 3\hy CNF formulas \emph{cannot} be
solved in time $2^{s\cdot n}\cdot n^{\bigO{1}}$ for some~$s>0$ where
$n$ is the number of variables~\cite{ImpagliazzoPaturiZane01}.

\paragraph{Graph Theory}
We recall some graph theoretical notations. For further basic
terminology on graphs and digraphs, we refer to standard
texts~\cite{Diestel12,BondyMurty08}.
An \emph{undirected graph} or simply a \emph{graph} is a
pair~$G=(V,E)$ where $V\neq \emptyset$ is a set of \emph{vertices} and
$E \subseteq \SB \{u,v\}\subseteq V \SM u \neq v \SE$ is a set of
\emph{edges}.
A graph~$G'=(V',E')$ is a
\emph{subgraph} of $G$ if $V'\subseteq V$ and $E'\subseteq E$ and an
\emph{induced subgraph} if additionally for any $u,v \in V'$ and
$\{u,v\} \in E$ also $\{u,v\} \in E'$. 
Let $G=(V,E)$ be a graph and~$A\subseteq V$ be a set of vertices.
We define the \emph{subgraph~$G-A$}, which is the graph obtained from $G$ by removing vertices~$A$, by~$G-A\eqdef (V\setminus A, \{e\mid e\in E, e\cap A=\emptyset\}$.
Graph $G$ is \emph{complete} if for any two
vertices~$u,v \in V$ there is an edge~$uv \in E$. $G$ contains a
\emph{clique} on $V'\subseteq V$ if the induced
subgraph~$(V',E')$ of $G$ is a complete graph.
A \emph{(connected) component}~$C\subseteq V$ of
$G$ is a $\subseteq$-largest set such that
for any two vertices~$u, v \in C$ there is a path from~$u$
to~$v$ in~$G$.

\paragraph{Tree Decompositions and Treewidth} %
For basic terminology on graphs
, we refer to standard
texts~\cite{Diestel12,BondyMurty08}.  For a (rooted) tree~$T=(N,A)$ with
root node~$\rootOf(T)$ and a node~$t \in N$, we let $\children(t)$ be the
sequence of all nodes~$t'$ in arbitrarily but fixed order, which have
an edge~$(t,t') \in A$.
Let $G=(V,E)$ be a graph.
A \emph{tree decomposition (TD)} of graph~$G$ is a pair
$\TTT=(T,\chi)$ where $T=(N,A)$ is a rooted tree 
and $\chi$ a mapping that assigns to each node $t\in N$ a set
$\chi(t)\subseteq V$, called a \emph{bag}, such that the following
conditions hold:
(i)~$V=\bigcup_{t\in N}\chi(t)$ and
$E \subseteq\bigcup_{t\in N}\SB uv \SM u,v\in \chi(t)\SE$; 
(ii)
for each $r, s, t\in N$ such that $s$ lies on the path from $r$ to
$t$, we have $\chi(r) \cap \chi(t) \subseteq \chi(s)$.
Then, $\width(\TTT) \eqdef \max_{t\in N}\Card{\chi(t)}-1$.  The
\emph{treewidth} $\tw{G}$ of $G$ is the minimum $\width({\TTT})$ over
all tree decompositions $\TTT$ of $G$.
For arbitrary but fixed $w \geq 1$, it is feasible in linear time to
decide if a graph has treewidth at most~$w$ and, if so, to compute a
tree decomposition of width $w$~\cite{Bodlaender96}.
In order to simplify case distinctions in the algorithms, we always
use so-called nice tree decompositions, which can be computed in
linear time without increasing the width~\cite{BodlaenderKoster08} and
are defined as follows.
For a node~$t \in N$, we say that $\type(t)$ is $\leaf$ if
$\children(t)=\langle \rangle$; $\join$ if
$\children(t) = \langle t',t''\rangle$ where
$\chi(t) = \chi(t') = \chi(t'') \neq \emptyset$; $\intr$
(``introduce'') if $\children(t) = \langle t'\rangle$,
$\chi(t') \subseteq \chi(t)$ and $|\chi(t)| = |\chi(t')| + 1$; $\rem$
(``removal'') if $\children(t) = \langle t'\rangle$,
$\chi(t') \supseteq \chi(t)$ and $|\chi(t')| = |\chi(t)| + 1$. If for
every node $t\in N$, $\type(t) \in \{ \leaf, \join, \intr, \rem\}$ and
bags of leaf nodes and the root are empty, then the TD is called
\emph{nice}.

\section{Dynamic Programming on TDs for SAT}\label{sec:sat}

\begin{figure}[t]%
\centering
\begin{tikzpicture}[node distance=7mm,every node/.style={fill,circle,inner sep=2pt}]
\node (a) [label={[text height=1.5ex,yshift=0.0cm,xshift=0.05cm]left:$p_2$}] {};
\node (b) [right of=a,label={[text height=1.5ex]right:$a$}] {};
\node (c) [below left of=b,label={[text height=1.5ex,yshift=0.09cm,xshift=0.05cm]left:$b$}] {};
\node (d) [below right of=b,label={[text height=1.5ex,yshift=0.09cm,xshift=-0.05cm]right:$p_1$}] {};
\draw (a) to (b);
\draw (b) to (c);
\draw (b) to (d);
\draw (c) to (d);
\end{tikzpicture}\hspace{1em}%
\begin{tikzpicture}[node distance=0.5mm]
\tikzset{every path/.style=thick}

\node (leaf1) [tdnode,label={[yshift=-0.25em,xshift=0.5em]above left:$t_1$}] {$\{a,b,p_1\}$};
\node (leaf2) [tdnode,label={[xshift=-1.0em, yshift=-0.15em]above right:$t_2$}, right = 0.1cm of leaf1]  {$\{a,p_2\}$};
\coordinate (middle) at ($ (leaf1.north east)!.5!(leaf2.north west) $);
\node (join) [tdnode,ultra thick,label={[]left:$t_3$}, above  = 1mm of middle] {$\{a\}$};

\coordinate (top) at ($ (join.north east)+(3.5em,0) $);
\coordinate (bot) at ($ (top)+(0,-4em) $);

\draw [->] (join) to (leaf1);
\draw [->] (join) to (leaf2);
\end{tikzpicture}%
\caption{Primal graph~$\primal{F}$ of~$F$ from Example~\ref{ex:running1}
  (left) with a TD~${\cal T}$ of graph~$\primal{F}$
  (right).}%
\label{fig:graph-td}%
\end{figure}

Before we introduce our algorithm, we need some notations for dynamic
programming on tree decompositions and recall how to solve the
decision problem~\SAT by exploiting small treewidth.
To this end, we present in Section~\ref{sec:dpforsat} basic notation and a simple algorithm for solving \SAT and \cSAT via utilizing treewidth.
The simple algorithm is inspired by related work~\cite{SamerSzeider10b}, which is extended by the capability of actually computing some (projected) models in Section~\ref{lab:computing}. The algorithm and the definitions of the whole section will then serve as a basis for solving projected model counting in Section~\ref{sec:projmodelcounting}.


%
%
%
%
%
%
%
%


\subsection{Dynamic Programming for \SAT} %
\label{sec:dpforsat}
%

\paragraph{Graph Representation of \SAT Formulas}
In order to use tree decompositions for satisfiability problems, we
need a dedicated graph representation of the given formula~$F$.
The \emph{primal graph}~$\primal{F}$ of~$F$ has as vertices the variables of~$F$
and two variables are joined by an edge if they occur together in a
clause of~$F$.
%
%
%
%
Further, we define some auxiliary notation.  For a given node~$t$ of a
tree decomposition~$(T,\chi)$ of the primal graph, we let the \emph{bag formula} $F_t \eqdef \SB c \SM c \in F, \var(c) \subseteq \chi(t)\SE$,~i.e.,
clauses entirely covered by~$\chi(t)$.  The set~$\progt{t}$  denotes
the union over~$F_{s}$ for all descendant nodes~$s$ of~$t$.
%
%
%
In the following, we sometimes simply write \emph{tree decomposition of
  formula~$F$} or \emph{treewidth of~$F$} and omit the actual graph
representation of~$F$.

\begin{example}\label{ex:running1}
  Consider formula~$F$ from Example~\ref{ex:running0}.
  The primal graph~$\primal{F}$ of formula~$F$ and a tree
  decomposition~$\TTT$ of~$\primal{F}$ are depicted in
  Figure~\ref{fig:graph-td}. Intuitively, ${\cal T}$ allows to
  evaluate formula~$F$ in parts. When evaluating $F_{\leq t_3}$, we
  split into $F_{\leq t_1}=\{c_1,c_2\}$ and
  $F_{\leq t_2}=\{c_3, c_4\}$, respectively.
\end{example}
%
%
%

%
%

%
%
  

%



\longversion{
\begin{figure}[t]
\centering
\includegraphics[scale=0.9]{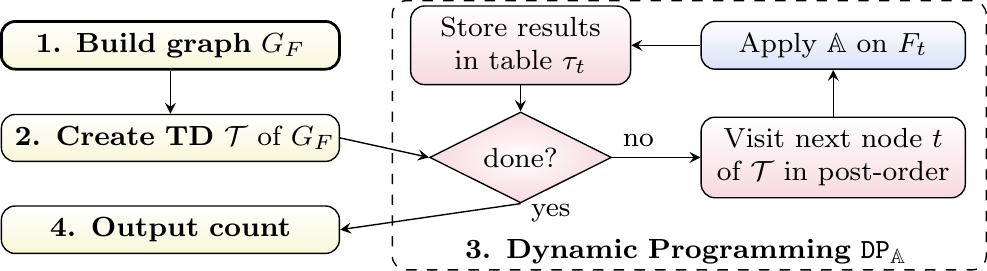}
\caption{The DP approach, where table algorithm~$\algo{A}$ modifies
  tables.~\cite{FichteEtAl17a}}
\label{fig:framework}
\end{figure}%
}

Algorithms that solve \SAT or \cSAT~\cite{SamerSzeider10b} in linear
time for input formulas of bounded treewidth proceed by dynamic
programming along the tree decomposition (in post-order) where at each
node~$t$ of the tree information is gathered~\cite{BodlaenderKloks96}
in a table~$\tau_t$.
A \emph{table}~$\tab{}$ is a set of rows, where a
\emph{row}~$\vec\tabval \in \tab{}$ is a sequence of fixed length\FIX{, which is denoted by angle brackets}. 
Tables are derived by an algorithm, which we therefore call
\emph{table algorithm}~$\AlgA$.  
The actual length, content, and meaning of the rows depend on the
algorithm~$\AlgA$ that derives tables.  Therefore, we often explicitly
state \emph{$\AlgA$-row} if rows of this \emph{type} are syntactically
used for table algorithm~$\AlgA$ and similar \emph{$\AlgA$-table} for
tables.
For sake of comprehension, we
specify the rows before presenting the actual table algorithm for
manipulating tables.
The rows used by a table algorithm~$\AlgS$ have in
common that the first position of these rows manipulated by~$\AlgS$
consists of an interpretation. The remaining positions of the row
depend on the considered table algorithm.
For each sequence~$\vec \tabval \in \tab{}$, we
write~$I(\vec \tabval)$ to address the interpretation (first) part of
the sequence~$\vec\tabval$. Further, for a given positive integer~$i$,
we denote by $\vec\tabval_{(i)}$ the $i$-th element of
row~$\vec\tabval$ and define~$\tab{(i)}$ as
$\tab{(i)}\eqdef\{\vec u_{(i)} \mid \vec u \in \tab{}\}$.

%
\begin{algorithm}[t]%
  \KwData{%
    Table algorithm $\AlgA$, and instance~$(F,P)$ of \PMC, a TD~$\TTT=(T,\chi)$ of the primal graph~$\primal{F}$ of~$F$, and tables~\Prev.\hspace{-5em}
  }%
  \KwResult{%
    Table mapping $\ATab{\AlgA}$, which maps each TD node~$t$ of~$T$ to some computed
    table~$\tau_t$.\hspace{-5em}
  } %
  $\ATab{\algo{A}} \leftarrow \{\}\qquad \tcc*[h]{empty mapping}$
  
  \For{\text{\normalfont iterate} $t$ in \text{\normalfont post-order}$(T)$}{
    \vspace{-0.05em}%

    $\Tab{} \leftarrow \langle \ATab{\AlgA}[t_1],\ldots,
    \ATab{\AlgA}[t_\ell] \rangle$ where
    $\children(t) = \langle t_1, \ldots, t_\ell
    \rangle\hspace{-5em}$

    $\ATab{\AlgA}[t] \leftarrow {\AlgA}(t,\chi(t),\prog_t,P \cap
    \chi(t), \Tab{},\Prev)$ %
    \vspace{-0.5em} %
  }%
  \Return{$\ATab{\AlgA}$} \vspace{-0.2em}%
 \caption{Algorithm ${\dpa}_{\AlgA}( (F,P), \TTT, \Prev)$ for DP on
   TD~${\cal T}$.} 
\label{fig:dpontd}
\end{algorithm}%
Then, the dynamic programming approach for Boolean
satisfiability \longversion{works as outlined in Figure~\ref{fig:framework} and}
performs the following steps:
\begin{enumerate}
\item Construct the primal graph 
$\primal{F}$ of~$F$. 
\item Compute a tree decomposition~$(T,\chi)$ of~$\primal{F}$, obtainable via heuristics. 
\item\label{step:dp} Run $\dpa_\AlgS$, \FIX{as presented in Listing~\ref{fig:dpontd}},
  which executes a table algorithm~$\AlgS$ for every node~$t$ in post-order of the nodes
  of~$T$, and returns $\ATab{\AlgS}$ mapping every node~$t$ to its table. $\AlgS$ takes as
  input\footnote{Actually, \AlgS takes in addition as input \Prev,
    which contains a mapping of nodes of the tree decomposition to
    tables,~i.e., tables of the previous pass. Later, we use this for
    a second traversal to pass results ($\ATab{\AlgS}$) from the first
    traversal to the table algorithm~$\PROJ$ for projected model
    counting in the second traversal.
  } %
  bag~$\chi(t)$, sub-formula~$F_t$, and tables \Tab{} previously
  computed at children of~$t$ and outputs a
  table~$\tab{t}$. 
\item Print a positive result whenever the table for node~$\rootOf(T)$ is not empty.
\end{enumerate}

\noindent\FIX{The basic steps of the approach are briefly summarized by Listing~\ref{fig:dpfordummies}.}

\begin{algorithm}[t]%
  \KwData{%
    A Boolean formula~$F$ in CNF.\hspace{-5em}
  }%
  \KwResult{%
    Satisfiability of~$F$.\hspace{-5em}
  } %
  $\mathcal{T}=(T,\chi) \leftarrow  \text{Decompose\_via\_Heuristics}(\primal{F})\qquad\tcc*[h]{Decompose}\hspace{-1em}$
  
  $\ATab{\AlgS} \leftarrow {\dpa}_{\AlgS}((F,P), \TTT, \emptyset)\qquad \tcc*[h]{DP via table algorithm $\AlgS$}$
 
  %
 %
  \Return{$\ATab{\PROJ}[\rootOf(T)] \neq \emptyset$}\quad\tcc*[h]{true iff root table is not empty\hspace{-1em}}
 \vspace{-0.2em}%
 \caption{Algorithm for solving $\SAT$ via dynamic programming.} 
\label{fig:dpfordummies}
\end{algorithm}%

\bigskip
\noindent%
Listing~\ref{fig:prim} presents table algorithm~\PRIM that uses
the primal graph representation. We provide only brief intuition, for
details we refer to the original source~\cite{SamerSzeider10b}.
The main idea is to store in table~$\tab{t}$ only interpretations restricted to
bag~$\chi(t)$ that can be extended to a model of sub-formula~$F_{\leq t}$.
Table algorithm~\PRIM transforms at node~$t$ certain row
combinations of the tables ($\Tab{}$) of child nodes of~$t$ into rows
of table~$\tab{t}$. The transformation depends on a case where
variable~$a$ is added or not added to an interpretation ($\intr$),
removed from an interpretation ($\rem$), or where coinciding
interpretations are required ($\join$).  In the end, an
interpretation~$I(\vec u)$ from a row~$\vec u$ of the table~$\tau_n$
at the root~$n$ proves that there is a superset~$J \supseteq I(\vec u)$
that is a model of~$F = \progt{n}$, and hence that the formula is
satisfiable.%
%
%
%
%

\renewcommand{\eqdef}{\leftarrow}
%
%
 \begin{algorithm}[t]
   \KwData{Node~$t$, bag $\chi_t$, clauses~$\prog_t$, and sequence~$\Tab{}=\langle \tau_1,\ldots,\tau_\ell\rangle$ of child \AlgS-tables of~$t$.}
   \KwResult{\AlgS-Table~$\tab{t}.$}
   \lIf(\hspace{-1em})
   {$\type(t) = \leaf$}{%
     $\tab{t} \eqdef \{ \langle
     \tuplecolor{\inputPredColor}{\emptyset}
     \rangle \}$%
     %
   }%
  \uElseIf{$\type(t) = \intr$ and $a\in\chi_t$ is introduced}{
   \vspace{-0.05em}
   \makebox[3.19cm][l]{$\tab{t} \eqdef \{ \langle \tuplecolor{\inputPredColor}{K} \rangle$}
     $|\;\langle \tuplecolor{\inputPredColor}{J} \rangle \in \tau_1,  {{\tuplecolor{black}{K \in \{J, J \cup \{a\}\}}, \tuplecolor{black}{K}}} \models \prog_t
      \} \hspace{-5em}$
     %
   \vspace{-0.05em}
     }\vspace{-0.05em}%
     \uElseIf{$\type(t) = \rem$ and $a \not\in \chi_t$ is removed}{%
       \makebox[3.3cm][l]{$\tab{t} \eqdef \{ \langle \tuplecolor{\inputPredColor}{J \setminus \{a\}}
       \rangle$}$|\;\langle \tuplecolor{\inputPredColor}{J}
       \rangle \in \tau_1 \}\hspace{-5em}$
       \vspace{-0.1em}
     } %
     \uElseIf{$\type(t) = \join$}{%
       \makebox[3.3cm][l]{$\tab{t} \eqdef \{ \langle \tuplecolor{\inputPredColor}{J}
         \rangle$}$|\;\langle \tuplecolor{\inputPredColor}{J} \rangle \in \tau_1\cap \tau_2
       \}\hspace{-5em}$
       \vspace{-0.1em}
     } 
     \Return $\tab{t}$
     \vspace{-0.25em}
     \caption{Table
       algorithm~$\algo{SAT}(t, \chi_t,\prog_t,\cdot, \Tab{},
       \cdot)$~\protect\cite{SamerSzeider10b}.}
 \label{fig:prim}\label{alg:prim}
\end{algorithm}%
\renewcommand{\eqdef}{{\ensuremath{\,\mathrel{\mathop:}=}}}
%
%
%
%
%
Example~\ref{ex:sat} lists selected tables when running
algorithm~$\dpa_{\PRIM}$ on a nice tree decomposition. 
\FIX{Note that illustration along the lines of a nice TD allows us
to visualize the basic cases separately. If one was to implement such an algorithm
on general TDs, one still obtains the same basic cases, but interleaved.
}
\begin{example}\label{ex:sat}
  Consider formula~$\prog$ from Example~\ref{ex:running1}.
  Figure~\ref{fig:running1} illustrates a nice TD~$\TTT'=(\cdot, \chi)$ of the primal graph of~$F$ and
  tables~$\tab{1}$, $\ldots$, $\tab{12}$ that are obtained during the
  execution of~$\dpa_{\PRIM}((F,\cdot),\TTT',\cdot)$.
  We assume that each row in a table $\tab{t}$ is identified by a
  number,~i.e., row $i$ corresponds to
  $\vec{u_{t.i}} = \langle J_{t.i} \rangle$.

  Table~$\tab{1}=\SB \langle\emptyset\rangle \SE$, due to
  $\type(t_1) = \leaf$.
  Since $\type(t_2) = \intr$, we construct table~$\tab{2}$
  from~$\tab{1}$ by taking~$J_{1.i}$ and $J_{1.i}\cup \{a\}$ for
  each~$\langle J_{1.i}\rangle \in \tab{1}$. Then,
  $t_3$ introduces $p_1$ and $t_4$ introduces $b$.
  $\prog_{t_1}=\prog_{t_2}=\prog_{t_3} = \emptyset$, but since
  $\chi(t_4) \subseteq \var(c_1)$ we have
  $\prog_{t_4} = \{c_1,c_2\}$ for $t_4$.
  In consequence, for each~$J_{4.i}$ of table~$\tab{4}$, we have
  $\alpha_{{J_{4.i}}} \models \{c_1,c_2\}$ since \PRIM enforces
  satisfiability of $\prog_t$ in node~$t$.  
  Since $\type(t_5) = \rem$, we remove variable~$p_1$ from all
  elements in $\tab{4}$ to construct $\tab{5}$. Note that we have
  already seen all rules where $p_1$ occurs and hence $p_1$ can no
  longer affect interpretations during the remaining traversal. We
  similarly create $\tab{6}=\{\langle \emptyset \rangle, \langle a \rangle\}$
  and~$\tab{{10}}=\{\langle a \rangle\}$.
  Since $\type(t_{11})=\join$, we build table~$\tab{11}$ by taking
  the intersection of $\tab{6}$ and $\tab{{10}}$. Intuitively, this
  combines interpretations agreeing on~$a$.
  %
  %
  By definition (primal graph and TDs), for every~$c \in \prog$,
  variables~$\var(c)$ occur together in at least one common bag.
  Hence, $\prog=\progt{t_{12}}$ and since
  $\tab{12} = \{\langle \emptyset \rangle \}$, we can reconstruct for example
  model~$\{a,b,p_2\} = J_{11.1} \cup J_{5.4} \cup J_{9.2}$ of~$F$ using highlighted (yellow) rows in Figure~\ref{fig:running1}.
  On the other hand, if~$F$ was unsatisfiable, $\tab{12}$ would be empty ($\emptyset$). 
  %
%
\end{example}%

\begin{figure}[t]
\centering
\begin{tikzpicture}[node distance=0.5mm]
\tikzset{every path/.style=thick}

\node (l1) [stdnode,label={[tdlabel, xshift=0em,yshift=+0em]right:${t_1}$}]{$\emptyset$};
\node (i1) [stdnode, above=of l1, label={[tdlabel, xshift=0em,yshift=+0em]right:${t_2}$}]{$\{a\}$};
\node (i12) [stdnode, above=of i1, label={[tdlabel, xshift=0em,yshift=+0em]right:${t_3}$}]{$\{a,p_1\}$};
\node (i13) [stdnode, above=of i12, label={[tdlabel, xshift=0em,yshift=+0em]right:${t_4}$}]{$\{a,b,p_1\}$};
\node (r1) [stdnode, above=of i13, label={[tdlabel, xshift=0em,yshift=+0em]right:${t_5}$}]{$\{a,b\}$};
\node (r12) [stdnode, above=of r1, label={[tdlabel, xshift=0em,yshift=+0em]right:${t_6}$}]{$\{a\}$};
\node (l2) [stdnode, right=2.5em of i12, label={[tdlabel, xshift=0em,yshift=+0em]left:${t_7}$}]{$\emptyset$};
\node (i2) [stdnode, above=of l2, label={[tdlabel, xshift=0em,yshift=+0em]left:${t_8}$}]{$\{p_2\}$};
\node (i22) [stdnode, above=of i2, label={[tdlabel, xshift=0em,yshift=+0em]left:${t_9}$}]{$\{a,p_2\}$};
\node (r2) [stdnode, above=of i22, label={[tdlabel, xshift=0em,yshift=+0em]left:${t_{10}}$}]{$\{a\}$};
\node (j) [stdnode, above left=of r2, yshift=-0.25em, label={[tdlabel, xshift=0em,yshift=+0.15em]right:${t_{11}}$}]{$\{a\}$};
\node (rt) [stdnode,ultra thick, above=of j, label={[tdlabel, xshift=0em,yshift=+0em]right:${t_{12}}$}]{$\emptyset$};
\node (label) [font=\scriptsize,left=of rt]{${\cal T}'$:};
\node (leaf1) [stdnode, left=1.25em of i1, yshift=-0.5em, label={[tdlabel, xshift=2.75em,yshift=+2em]above left:$\tab{4}$}]{%
	\begin{tabular}{l}%
		\multicolumn{1}{l}{$\langle \tuplecolor{\inputPredColor}{J_{4.i}} \rangle$}\\
		\hline\hline
		$\langle \tuplecolor{\inputPredColor}{\emptyset}\rangle$\\\hline
		$\langle \tuplecolor{\inputPredColor}{\{b\}}\rangle$\\\hline
		\rowcolor{yellow}$\langle\tuplecolor{\inputPredColor}{\{a,b\}}\rangle$\\\hline
		$\langle \tuplecolor{\inputPredColor}{\{p_1\}}\rangle$\\\hline
		$\langle\tuplecolor{\inputPredColor}{\{a,p_1\}}\rangle$\\\hline
		$\langle\tuplecolor{\inputPredColor}{\{a,b,p_1\}}\rangle$\\
	\end{tabular}%
};
\node (leaf1b) [stdnodenum,left=of leaf1,xshift=0.6em,yshift=0pt]{%
	\begin{tabular}{c}%
		\multirow{1}{*}{$i$}\\ 
		\hline\hline
		$1$ \\\hline
		$2$ \\\hline
		$3$ \\\hline
		$4$ \\\hline
		$5$ \\\hline
		$6$
	\end{tabular}%
};
\node (leaf0x) [stdnode, left=0.75em of leaf1b, yshift=1.5em, label={[tdlabel, xshift=2em,yshift=+1.5em]above left:$\tab{5}$}]{%
	\begin{tabular}{l}%
		\multicolumn{1}{l}{$\langle \tuplecolor{\inputPredColor}{J_{5.i}} \rangle$}\\
		\hline\hline
		$\langle \tuplecolor{\inputPredColor}{\emptyset}\rangle$\\\hline
		$\langle\tuplecolor{\inputPredColor}{\{a\}}\rangle$\\\hline
		$\langle\tuplecolor{\inputPredColor}{\{b\}}\rangle$\\\hline
		\rowcolor{yellow}$\langle\tuplecolor{\inputPredColor}{\{a,b\}}\rangle$\\
	\end{tabular}%
};
\node (leaf0b) [stdnodenum,left=of leaf0x,xshift=0.6em,yshift=0pt]{%
	\begin{tabular}{c}%
		\multirow{1}{*}{$i$}\\ 
		\hline\hline
		$1$ \\\hline
		$2$ \\\hline
		$3$ \\\hline
		$4$ 
	\end{tabular}%
};
\node (leaf2b) [stdnodenum,right=2.5em of j,xshift=-0.75em,yshift=+0.25em]  {%
	\begin{tabular}{c}%
		\multirow{1}{*}{$i$}\\ 
		\hline\hline
		$1$\\\hline
		$2$\\
	\end{tabular}%
};
\node (leaf2) [stdnode,right=-0.4em of leaf2b, label={[tdlabel, xshift=0em,yshift=-0.25em]below:$\tab{9}$}]  {%
	\begin{tabular}{l}%
		\multirow{1}{*}{$\langle \tuplecolor{\inputPredColor}{J_{9.i}} \rangle$}\\ 
		\hline\hline
		$\langle \tuplecolor{\inputPredColor}{\{a\}}\rangle$\\\hline
		\rowcolor{yellow}$\langle \tuplecolor{\inputPredColor}{\{a,p_2\}}\rangle$\\
	\end{tabular}%
};
\coordinate (middle) at ($ (leaf1.north east)!.5!(leaf2.north west) $);
\node (join) [stdnode,left=5.5em of r12, yshift=0.0em, label={[tdlabel, xshift=2em,yshift=+0.25em]above left:$\tab{{11}}$}] {%
	\begin{tabular}{l}%
		\multirow{1}{*}{$\langle \tuplecolor{\inputPredColor}{J_{11.i}} \rangle$}\\
		\hline\hline
		\rowcolor{yellow}$\langle \tuplecolor{\inputPredColor}{\{a\}} \rangle$\\
	\end{tabular}
};
\node (joinb) [stdnodenum,left=-0.45em of join] {%
	\begin{tabular}{c}
		\multirow{1}{*}{$i$}\\
		\hline\hline
		$1$\\
	\end{tabular}%
};
\node (rtx) [stdnode,left=0.0em of r12, yshift=2.75em, label={[tdlabel, xshift=0em,yshift=-0.8em]right:$\tab{{12}}$}] {%
	\begin{tabular}{l}%
		\multirow{1}{*}{$\langle \tuplecolor{\inputPredColor}{J_{12.i}} \rangle$}\\
		\hline\hline
		\rowcolor{yellow}$\langle \tuplecolor{\inputPredColor}{\emptyset} \rangle$\\
	\end{tabular}
};
\node (rtb) [stdnodenum,left=-0.45em of rtx] {%
	\begin{tabular}{c}
		\multirow{1}{*}{$i$}\\
		\hline\hline
		$1$\\
	\end{tabular}%
};
\node (leaf0n) [stdnodenum,yshift=0.5em, right=2.5em of l1] {%
	\begin{tabular}{c}%
		\multirow{1}{*}{$i$}\\ 
		\hline\hline
		$1$
	\end{tabular}%
};
\node (leaf0) [stdnode,right=-0.5em of leaf0n, label={[tdlabel, xshift=-1em,yshift=0.15em]above right:$\tab{1}$}] {%
	\begin{tabular}{l}%
		\multicolumn{1}{l}{$\langle \tuplecolor{\inputPredColor}{J_{1.i}} \rangle$}\\
		\hline\hline
		\rowcolor{yellow}$\langle \tuplecolor{\inputPredColor}{\emptyset}\rangle$
	\end{tabular}%
};
\coordinate (top) at ($ (leaf2.north east)+(0.6em,-0.5em) $);
\coordinate (bot) at ($ (top)+(0,-12.9em) $);

\draw [<-] (j) to (rt);
\draw [->] (j) to ($ (r12.north)$);
\draw [->] (j) to ($ (r2.north)$);
\draw [->](r2) to (i22);
\draw [<-](i2) to (i22);
\draw [<-](l2) to (i2);
\draw [<-](l1) to (i1);
\draw [->](i12) to (i1);
\draw [->](i13) to (i12);
\draw [->](r1) to (i13);
\draw [->](r12) to (r1);

\draw [dashed, bend left=0] (j) to (join);
\draw [dashed, bend right=15] (rtx) to (rt);
\draw [dashed, bend right=20] (i22) to (leaf2);
\draw [dashed, bend right=50] (i13) to (leaf1);
\draw [dashed, bend left=25] (leaf0) to (l1);
\draw [dashed, bend left=22] (leaf0x) to (r1);
\end{tikzpicture}
\caption{Selected tables obtained by algorithm~$\dpa_{\algo{SAT}}$ on tree decomposition~${\cal T}'$.}
\label{fig:running1}
\end{figure}

Interestingly, the above table algorithm \PRIM can be easily extended 
to also count models. Such a table algorithm for solving~\cSAT works similarly to \PRIM,
but additionally also maintains a counter~\cite{SamerSzeider10b}.
There, intuitively, rows of tables for leaf nodes set this counter to~$1$
and introduce nodes basically just copy the counter value of child rows.
Then, upon removing a certain variable, one has to add (sum up) counters accordingly, and for join nodes counters need to be multiplied.
\FIX{Finally, the counters of the table for the root node can be summed up to obtain the solution to the~\cSAT problem.}

\subsection{\FIX{(Re-)constructing Interpretations and Models}}\label{lab:computing}
\FIX{Even further, with the help of the obtained tables during dynamic programming,
one can actually construct (projected) models by combining
suitable predecessor rows.
The idea is to combine those obtained rows that contain parts of models that fit together.
To this end, we require the following definition, which we will also use later.} At a node~$t$
and for a row~$\vec\tabval$ of the \FIX{computed} table $\ATab{\AlgS}[t]$, it yields
the \FIX{\emph{originating rows}} in the tables of the children of~$t$ that were involved in
computing row~$\vec\tabval$ by algorithm~\AlgS.

\newcommand{\llangle}{\ensuremath{\langle\hspace{-2pt}\{\hspace{-0.2pt}}}
\newcommand{\rrangle}{\ensuremath{\}\hspace{-2pt}\rangle}}
\newcommand{\STab}{\ensuremath{\ATab{\AlgS}}}%

\begin{definition}[\FIX{Origins}, cf.,~\cite{FichteEtAl17b}]\label{def:origin}
  Let $F$ be a formula, $\TTT=(T, \chi)$ be a tree decomposition of~$F$,
  $t$ be a node of~$T$ with $\children(t)=\langle t_1, \ldots, t_{\ell}\rangle$, and 
  %
  %
  %
  $\tau_1 \in \ATabs{\AlgS}{t_1}, \ldots, \tau_{\ell}\in\ATabs{\AlgS}{t_\ell}$ be the tables computed by
  $\dpa_\AlgS$. 
  %
  %
  %

  For a given $\AlgS$-row~$\vec u$ in~$\ATabs{\AlgS}{t}$, we define its originating
  $\AlgS$-rows by
  %
    \FIX{$\orig(t,\vec \tabval) \eqdef \SB \vec s \SM \vec s \in 
    \tau_1 \times \cdots \times \tau_{\ell}, \tab{} =
    {\AlgS}(t,\chi(t),\prog_t, \cdot,\llangle \vec s\rrangle, 
    \cdot), \vec u \in \tab{} \SE.$\footnote{
  \FIX{Given a sequence~$\vec s=\langle s_1, \ldots, s_{\ell} \rangle$, we
  let
  $\llangle \vec s\rrangle \eqdef \langle \{s_1\}, \ldots,
  \{s_{\ell}\} \rangle$, for technical reasons.}} %
  %
  %
  We naturally extend this to a $\AlgS$-table~$\sigma$ by
  $\origs(t,\sigma) \eqdef$ $\bigcup_{\vec u \in \sigma}\orig(t,\vec
    u).$}
\end{definition}


\FIX{Example~\ref{ex:origins} illustrates Definition~\ref{def:origin} for
our running example, where we briefly show origins for some rows of selected tables.}

\begin{example}\label{ex:origins}
  Consider formula~$F$, tree decomposition~$\TTT'=(T,\chi)$, and
  tables $\tab{1}, \ldots, \tab{12}$ from Example~\ref{ex:sat}.  We
  focus
  on~$\vec{\tabval_{1.1}} = \langle J_{1.1} \rangle
  =\langle\emptyset\rangle$ of table~$\tab{1}$ of the leaf~$t_1$. The
  row~$\vec{\tabval_{1.1}}$ has no preceding row,
  since~$\type(t_1)=\leaf$. Hence, we have
  $\origse{\PRIM}(t_1,\vec{\tabval_{1.1}})=\{\langle \rangle\}$.
  The origins of row~$\vec{\tabval_{5.1}}$ of table~$\tab{5}$ are
  given by $\origse{\PRIM}(t_5,\vec{\tabval_{5.1}})$, which correspond
  to the preceding rows in table~$t_4$ that lead to
  row~$\vec{\tabval_{5.1}}$ of table~$\tab{5}$ when running
  algorithm~$\PRIM$,~i.e.,
  $\origse{\PRIM}(t_5,\vec{\tabval_{5.1}}) = \{\langle
  \vec{\tabval_{4.1}} \rangle, \langle\vec{\tabval_{4.4}}\rangle\}$.
  Observe that $\origse{\PRIM}(t_i,\vec\tabval)=\emptyset$ for any
  row~$\vec\tabval\not\in\tab{i}$.
  For node~$t_{11}$ of type~$\join$ and row~$\vec{\tabval_{11.1}}$, we
  obtain
  $\origse{\PRIM}(t_{11},\vec{\tabval_{11.1}}) =
  \{\langle\vec{\tabval_{6.2}},$ $\vec{\tabval_{10.1}} \rangle\}$
	(see Example~\ref{ex:sat}).
  %
  %
  More general, when using algorithm~\PRIM, at a node~$t$ of
  type~$\join$ with table~$\tau$ we have
  $\origse{\PRIM}(t, \vec u)=\{\langle \vec\tabval,
  \vec\tabval\rangle\}$ for 
	row~$\vec u \in \tau$. 
%
\end{example}

\FIX{Definition~\ref{def:origin} refers to the predecessors of rows. 
In order to reconstruct models, one needs to recursively combine
these origins from a node~$t$ down to the leafs. 
%
This idea of \emph{combining suitable rows} is formalized in the following definition,
which introduces the concept of \emph{extensions}. Thereby, rows are \emph{extended}
such that one can then reconstruct models from these extensions.
}


\begin{definition}[\FIX{Extensions}]\label{def:extensions}
  Let $F$ be a formula, $\TTT=(T, \chi)$ be a tree decomposition, $t$ be a node of~$T$, 
and $\vec u$ be a row of $\ATab{\AlgS}[t]$.

  An \emph{extension below~$t$} is a set of pairs where a pair consists
  of a node~$t'$ of~$T[t]$ and a row~$\vec v$ of $\ATab{\AlgS}[t']$
  and the cardinality of the set equals the number of nodes in the
  sub-tree~$T[t]$. We define the family of \emph{extensions below~$t$}
  recursively as follows.  If $t$ is of type~\leaf, then
  $\Ext_{\leq t}(\vec u) \eqdef \{\{\langle t,\vec u\rangle\}\}$;
  otherwise
  $\Ext_{\leq t}(\vec u) \eqdef \bigcup_{\vec v \in \origs(t,\vec u)}
  \big\SB\{\langle t,\vec u\rangle\}\cup X_1 \cup \ldots \cup X_\ell
  \SM X_i\in\Ext_{\leq t_i}({\vec v}_{(i)})\big\SE$ 
  for the~$\ell$ children~$t_1, \ldots, t_\ell$ of~$t$.
  %
  We lift this notation for a $\AlgS$-table~$\sigma$ by
  $\Ext_{\leq t}(\sigma)\eqdef \bigcup_{\vec u\in\sigma} \Ext_{\leq
    t}(\vec u)$.  Further, we
  let~$\Exts \eqdef \Ext_{\leq n}(\ATab{\AlgS}[n])$. 

\end{definition}

\FIX{Indeed, if we construct extensions below the root~$n$, it allows us
to also obtain all models of a formula~$F$.}  
\FIX{Finally, we define notation that gives us a way to
\emph{reconstruct interpretations} from such (families of) extensions.}


\begin{definition}[\FIX{Interpretations of Extensions}]\label{def:iextensions}
 Let~$(F, P)$ be an instance of \PMC, $\TTT=(T, \chi)$ be a tree decomposition
  of~$F$, $t$ be a node of~$T$. Further, let $E$ be a family of extensions below~$t$, and $P$ be a set of projection variables. We
  define the \emph{set~$I(E)$ of interpretations} of~$E$ by
  $I(E) \eqdef \big\SB \bigcup_{\langle \cdot, \vec u \rangle \in X} I(\vec u)
  \mid X \in E \big\SE$
  and the set~$I_P(E)$ of \emph{projected interpretations} by
  $I_P(E) \eqdef \big\SB \bigcup_{\langle \cdot, \vec u \rangle \in X} I(\vec
  u) \cap P \mid X \in E \big\SE$.

\end{definition}

\FIX{We briefly illustrate these concepts along the lines of our running example.}

\begin{example}
  Consider again formula~$F$ and tree decomposition~${\cal T}'$ with
  root~$n$ of~$F$ from Example~\ref{ex:sat}.
  Let~$X=\{\langle t_{12}, \langle\emptyset\rangle\rangle, \langle t_{11},
  \langle\{a\}\rangle\rangle,$
  $\langle t_6, \langle\{a\}\rangle \rangle, \langle t_5, \langle\{a,b\}\rangle\rangle,$ $\langle
  t_4,\hspace{-0.1em} \langle\{a,b\}\rangle\rangle,$
  $\langle t_3,\hspace{-0.1em} \langle\{a\}\rangle \rangle, \langle t_2,\hspace{-0.1em} \langle\{a\}\rangle\rangle, \langle t_1,\hspace{-0.1em}
  \langle\emptyset \rangle\rangle, \langle t_{10},\hspace{-0.1em} \langle\{a\}\rangle\rangle,\langle t_9,\hspace{-0.1em}
  $ $\langle\{a,p_2\}\rangle\rangle, \langle t_8,\hspace{-0.1em} \langle\{p_2\}\rangle\rangle,$ $\langle t_7, \langle\emptyset
  \rangle\rangle \}$ be an extension below~$n$.  Observe that~$X\in\Exts$ and
  that Figure~\ref{fig:running1} highlights those rows of 
  tables for nodes~$t_{12},t_{11},t_9,t_5,t_4$ and~$t_1$ that also
  occur in~$X$ (in yellow). Further, $I(\{X\})=\{a,b,p_2\}$ computes
  the corresponding model of~$X$, and $I_P(\{X\}) = \{p_2\}$ derives
  the projected model of~$X$.  $I(\Exts)$ refers to the set of
  models of~$F$, whereas~$I_P(\Exts)$ is the set of projected models of~$F$.
\end{example}

\FIX{In order to only construct extensions that correspond to (parts of) models of the formula,
we simply need to access only those extensions that contain rows
that lead to models of the formula. As already observed in the previous example,
these rows are precisely the ones contained in~$\Exts$. 
The resulting extensions for a node~$t$ are formalized in the following concept
of \emph{satisfiable extensions}, whereby we take only those extensions of~$\Ext_{\leq t}$ that are also contained in~$\Exts$.
}

\begin{definition}[\FIX{Satisfiable Extension}]\label{def:satext}
  Let $F$ be a formula, $\TTT=(T, \chi)$ be a tree decomposition of~$F$, $t$
  be a node of $T$, and $\sigma \subseteq \ATab{\AlgS}[t]$ be a set of rows.
  %
  Then, we define 
  the \emph{satisfiable
    extensions below~$t$} for~$\sigma$ by
  $\PExt_{\leq t}(\sigma)\eqdef \bigcup_{\vec u\in\sigma} \SB X \SM X
    \in \Ext_{\leq t}(\vec u), X \subseteq Y, Y \in \Exts\SE.$
\end{definition}


\section{Counting Projected Models by Dynamic Programming}\label{sec:projmodelcounting}

While the transition from deciding \SAT to solving \cSAT is quite simple
by adding an additional counter, it turns out that the problem \PMC
requires more effort. 
We solve this problem \PMC by providing an algorithm in Section~\ref{sec:algo} that utilizes treewidth and adheres to multiple passes (rounds) of computation that are guided along a tree decomposition.
Then, we give detailed formal arguments on correctness of this algorithm in Section~\ref{sec:correct}.
Later, in Section~\ref{sec:complexityresults}
we discuss complexity results in the form of 
matching upper and lower bounds, where it turns out
that our algorithm cannot be significantly improved.
%
\subsection{Solving \PMC by means of Dynamic Programming}\label{sec:algo}
Next, we introduce the dynamic programming
algorithm~\mdpa{\AlgS} to solve the projected model counting problem
(\PMC) for Boolean formulas.
\FIX{From a high-level perspective, our algorithm builds upon
the table algorithm~\AlgS from the previous section; we assume again a formula~$F$
and a tree decomposition~$\mathcal{T}=(T,\chi)$ of~$F$, and 
additionally a set~$P$ of projection variables. Thereby, 
the table for each tree decomposition node~$t$ consists of 
a set~$\sigma$ of assignments restricted
to bag variables~$\chi(t)$ (as computed by~\AlgS) that agree on
their assignment of variables in~$P\cap\chi(t)$, and a counter~$c$.
Intuitively, this counter~$c$ \emph{counts} those \emph{satisfying 
assignments} of~$F_{\leq t}$ restricted to~$P\cap\chi(t)$ that are   
among 
satisfiable extensions 
and extend any assignment in~$\sigma$.
Then, 
for the (empty) tree decomposition root~$n$, there is only one single counter 
which is the projected
model count of~$F$ with respect to~$P$.
The \emph{challenge} of our algorithm~$\mdpa{\AlgS}$ is to compute these counts~$c$
by only considering local information, i.e., previously computed tables of child nodes of~$t$.
To this end, we utilize mathematical combinatorics,
namely the principle of inclusion-exclusion
principle~\cite{GrahamGrotschelLovasz95a}, which we need to apply in an interleaved fashion.
}

Concretely, our algorithm~\mdpa{\AlgS} traverses the tree decomposition twice following a
multi-pass dynamic programming paradigm~\cite{FichteEtAl17b}.
\longversion{Figure~\ref{fig:multiarch} illustrates the steps of our
algorithm~\mdpa{\AlgS}\FIX{, which are also presented in the form of Listing~\ref{fig:mpontd}.}}
Similar to the previous section\longversion{ (cf., Figure~\ref{fig:framework})}, we
construct a graph representation and heuristically compute a tree
decomposition of this graph. Then, we run $\dpa_\AlgS$ (see
Listing~\ref{fig:dpontd}) in Step~3a as \emph{first pass}. Step~3a can
also be seen as a preprocessing step for projected model counting,
from which we immediately know \FIX{whether the formula has a
model. However, we keep the \AlgS-tables that have been computed
in Step~3a. These tables form the basis for the next step.}

\longversion{
\begin{figure}[t]
\centering
\includegraphics[scale=0.81]{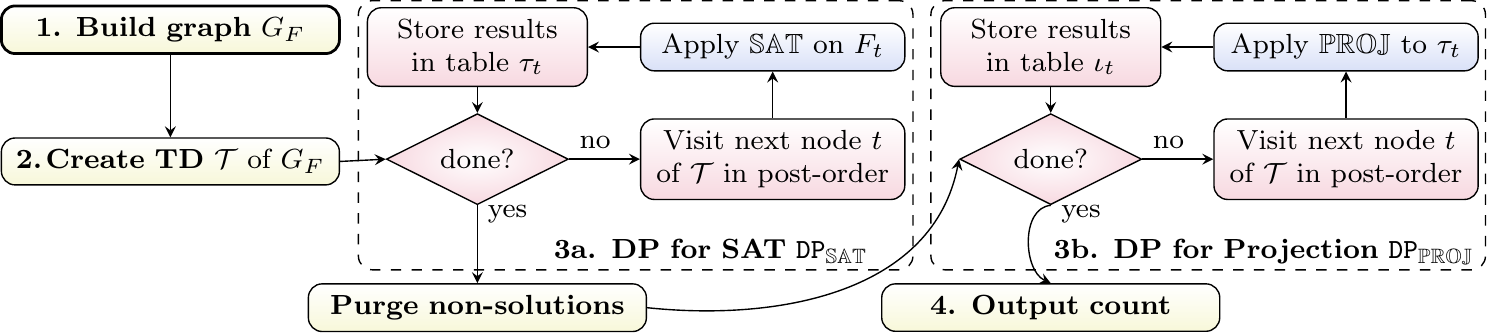}
\vspace{-2em}
\caption{Algorithm~$\mdpa{\AlgS}$ consists of~$\dpa_\AlgS$
  and~$\dpa_\PROJ$. 
}
\label{fig:multiarch}
\end{figure}%
}

\begin{algorithm}[t]%
  \KwData{%
    An instance~$(F,P)$ of \PMC.\hspace{-5em}
  }%
  \KwResult{%
    The projected model count of~$(F,P)$.\hspace{-5em}
  } %
  $\mathcal{T}=(T,\chi) \leftarrow  \text{Decompose\_via\_Heuristics}(\primal{F})\qquad\tcc*[h]{Decompose}\hspace{-1em}$
  
  $\ATab{\AlgS} \leftarrow {\dpa}_{\AlgS}((F,P), \TTT, \emptyset)\qquad \tcc*[h]{DP via table algorithm $\AlgS$}$
 
  \tcc*[h]{Purge non-solutions of $\ATab{\AlgS}$; ensured by using~$\PExt$ below.\hspace{-1em}}

  $\ATab{\PROJ} \leftarrow {\dpa}_{\PROJ}((F,P), \TTT, \ATab{\AlgS})\quad \tcc*[h]{DP via algorithm $\PROJ$}$
  
  %
 %
  \Return{$\sum_{\langle \varphi, c\rangle \in \ATab{\PROJ}[\rootOf(T)]}c$}\qquad\tcc*[h]{Return projected model count\hspace{-1em}}
 \vspace{-0.2em}%
 \caption{Algorithm $\mdpa{\AlgS}(F,P)$ for solving PMC via dynamic programming.} 
\label{fig:mpontd}
\end{algorithm}%

\FIX{There, we remove all rows from the obtained \AlgS-tables} which
cannot be extended to a model of the \SAT problem (\emph{``Purge
  non-solutions''}).
In other words, we keep only rows~$\vec u$ in table~$\ATab{\AlgS}[t]$
at node~$t$ if its interpretation~$I(\vec u)$ can be extended to a
model of~$F$.
%
%
%
%
%
Thereby, we avoid redundancies and can simplify the description and presentation of our
next step, since we then only consider rows that are (parts of) models.
\FIX{%
Intuitively, the rows involving non-models contributes only non-relevant information, as
also observed in related works~\cite{Tamaki19,BannachBerndt22}.
%
Formally, this is achieved by utilizing satisfiable extensions as defined in Definition~\ref{def:satext},
since these extensions precisely consider the rows that contribute to models.}

In Step~3b ($\dpa_\PROJ$), we perform the \emph{second pass}, where we traverse the tree decomposition a second
time to count projections of interpretations of rows in
$\AlgS$-tables.
%
%
\FIX{Observe that the tree traversal in $\dpa_\PROJ$ is the same as before. Therefore, in the following, we describe the ingredients that lead to table algorithm~$\PROJ$.}
For \PROJ, 
a row at a node~$t$ is a pair $\langle\sigma, c \rangle$ where $\sigma$ is a
$\AlgS$-table, in particular, a subset of $\ATab{\AlgS}[t]$
computed by $\dpa_\AlgS$, and $c$ is a non-negative integer.
\FIX{Below, we characterize~$\sigma$, which is based on 
grouping rows in equivalence classes.}

\bigskip
\noindent\textbf{\FIX{Equivalence Classes for $\AlgS$-Tables.}}
The following definitions provide central notions for grouping rows of
tables according to the given projection of variables\FIX{, which yields an equivalence relation.}

\newcommand{\RRR}{\ensuremath{\mathcal{R}}}


\begin{definition}
  Let $(F,P)$ be an instance of \PMC and $\sigma$ be a $\AlgS$-table.
  We define the relation~$\bucket \subseteq \sigma \times \sigma$ to
  consider equivalent rows with respect to the projection of its
  interpretations by 
  $\bucket \eqdef \SB (\vec u,\vec v) \SM \vec u, \vec v \in \sigma,
  \restrict{I(\vec u)}{P} = \restrict{I(\vec v)}{P}\SE.$
\end{definition}

\begin{observation}\label{obs:relation}
  The relation~$\bucket$ is an equivalence relation.
\end{observation}

\noindent \FIX{Based on this equivalence relation, we define corresponding equivalence classes.}

\begin{definition}[\FIX{Equivalence Classes}]
  Let~$\tau$ be a $\AlgS$-table and $\vec u$ be a row of $\tau$.  The
  relation~$\bucket$ induces equivalence classes~$[\vec u]_P$ on the
  $\AlgS$-table~$\tau$ in the usual way,~i.e.,
  $[\vec u]_P = \SB \vec v \SM \vec v \bucket \vec u,\vec v \in
  \tau\}$~\cite{Wilder12a}.
  We denote by~$\buckets_P(\tau)$ the set of equivalence classes
  of~$\tau$,~i.e.,
  $\buckets_P(\tau) \eqdef\, (\tau / \bucket) = \SB [\vec u]_P \SM
  \vec u \in \tau\SE$.
  %

\end{definition}

\noindent\FIX{These classes are briefly demonstrated on our running example.}

\begin{example}\label{ex:equiv}
  Consider again formula~$F$ and set~$P$ of projection variables from
  Example~\ref{ex:running0} and tree
  decomposition~$\mathcal{T}'=(T,\chi)$ and $\AlgS$-table~$\tab{4}$
  from Figure~\ref{fig:running1}.
  %
  %
  We have $\vec{ u_{4.1}} =_P \vec{ u_{4.2}}$ and
  $\vec{ u_{4.4}} =_P \vec{ u_{4.5}}$.  We obtain the
  set~$\tab{4}/\bucket$ of equivalence classes of $\tab{4}$
  by~$\buckets_P(\tab{4})=\{\{\vec{ u_{4.1}}, \vec{ u_{4.2}}, \vec{
    u_{4.3}}\}, \{\vec{ u_{4.4}},$
  $ \vec{ u_{4.5}}, \vec{ u_{4.6}}\}\}$.
\end{example}

\FIX{Indeed, the algorithm \PROJ, 
stores at a node~$t$ pairs $\langle\sigma, c \rangle$, where~$\sigma$
is actually a (non-empty) subset of the equivalence classes in~$\buckets_P(\ATab{\AlgS}[t])$. Next, we discuss how the integer~$c$ aids in projected counting for such a subset~$\sigma$.}

\bigskip
\noindent\textbf{\FIX{Counting for Equivalence Classes.}} \FIX{In fact, we store in integer~$c$ a count that expresses the number of
``intersection'' projected models ($\ipmc$) that indicates for~$\sigma$ the number of projected models up to node~$t$ that the rows in~$\sigma$ haves \emph{in common (intersection of models)}.
In the end, we aim for the
projected model count ($\pmc$), i.e., the \emph{combined} number of projected models (union of models), where~$\sigma$ is involved. However, it turns out that the process of computing these projected model counts will be heavily interleaved with the~$\ipmc$ counts.
In the following, we define both counts for a node~$t$ of a tree decomposition by means of the satisfying extensions below~$t$.}

\FIX{Notably, the effort of directly computing these counts when strictly following the definition below would not result in an algorithm that is fixed-parameter tractable. As a result, our approach is then subsequently developed thereafter, without explicitly involving \emph{every} descendant node below~$t$ in order to fulfill the desired runtime claims. %
%
%
}


%


\begin{definition}\label{def:pmc}
  Let $(F,P)$ be an instance of~$\PMC$, $\TTT=(T, \chi)$ be a tree decomposition of~$F$,
  $t$ be a node of~$T$, and $\sigma \subseteq \ATab{\AlgS}[t]$ be a set
  of~$\AlgS$-rows 
  for node~$t$. 
  Then, the \emph{intersection projected model count}
  $\ipmc_{\leq t}(\sigma)$ of $\sigma$ below~$t$ is the size of the
  intersection over projected interpretations of the satisfiable
  extensions of~$\sigma$ below~$t$,~i.e.,
  $\ipmc_{\leq t}(\sigma) \eqdef \Card{\bigcap_{\vec u\in\sigma}
    I_P(\PExt_{\leq t}(\{\vec u\}))}$.
    
   The \emph{projected model count} $\pmc_{\leq t}(\sigma)$ of
  $\sigma$ below~$t$ is the size of the union over projected
  interpretations of the satisfiable extensions of~$\sigma$ below~$t$,
  formally,
  $\pmc_{\leq t}(\sigma) \eqdef \Card{\bigcup_{\vec u\in\sigma}
    I_P(\PExt_{\leq t}(\{\vec u\}))}$.
\end{definition}

\FIX{Note that this definition relies on satisfiable extensions as given in Definition~\ref{def:satext}. Intuitively, the counts~$\ipmc_{\leq t}$ represent for a set~$\sigma$ of \AlgS-rows, the cardinality of those projected models of~$F_{\leq t}$ that can be extended to models of~$F$, where \emph{every} row in~$\sigma$ is involved.
Consequently, for the root~$n$ of a nice tree decomposition of~$F$ we have that~$\ipmc_{\leq n}(\{\langle \emptyset \rangle\}) = \pmc_{\leq n}(\{\langle \emptyset \rangle\})$ coincides with the \emph{projected model count} of~$F$.
This is the case since~$F_{\leq n}=F$, the bag of~$n$ is empty, and therefore the \AlgS-table for~$n$ contains one row if and only if~$F$ is satisfiable.
}

\FIX{Observe that when computing these counts for a node~$t$, we cannot directly count models since this would not yield a fixed-parameter tractability algorithm. 
Instead, in order to count, we may \emph{only utilize counters} for sets~$\sigma$ of rows in tables of~$t$ and direct child nodes of~$t$,
which is more involved than directly counting models. %
This is established for $\pmc$ next by relying on combinatorial counting principles like inclusion-exclusion~\cite{GrahamGrotschelLovasz95a}. %
}

\bigskip
\noindent\textbf{\FIX{Computing Projected Model Counts ($\pmc$).}}
\FIX{Since $\PROJ$ stores in $\PROJ$-tables an $\AlgS$-table together with a counter,
in the end we need to describe how these counters are maintained.
As the first step, we show how for a node~$t$, these counters ($\ipmc$ values) for child tables of~$t$
can be used to compute $\pmc$ values for~$t$.}
\shortversion{
%
%
%
Later, we use the definition in the context of looking up the already
computed projected counts for tables of \emph{children} of a given
node.
%


\begin{definition}\label{def:childpcnt}
  Given a $\PROJ$-table~$\iota$ and a $\AlgS$-table~$\sigma$ we define
  the \emph{stored $\ipmc$} for all rows of~$\sigma$ in~$\iota$ by
  $\sipmc(\iota, \sigma) \eqdef \sum_{\langle \sigma, c\rangle \in
    \iota} c.$
  %
  Later, we apply this to rows from several origins.
  Therefore, for a
  sequence~$s$ of $\PROJ$-tables of length $\ell$ and a set~$O$ of
  sequences of $\AlgS$-rows where each sequence is of length~$\ell$,
  we let
  $\sipmc(s, O)=\prod_{i \in \{1, \ldots,
    \ell\}, \langle O_{(i)}, c \rangle \in s_{(i)}} c$.
\end{definition}

When computing $\sipmc$ in Definition~\ref{def:childpcnt}, we select
the $i$-th position of the sequence together with sets of the $i$-th
position from the set of sequences. We need this somewhat technical
construction, since later at node~$t$ we apply this definition to
$\PROJ$-tables of children of~$t$ and origins of subsets of
$\AlgS$-tables. There, we may simply have several children if the node
is of type~$\join$ and hence we need to select from the right
children.%
}
%
%
Intuitively, when we are at a node~$t$ in the Algorithm~$\dpa_\PROJ$
we already computed all tables $\ATab{\AlgS}$ by $\dpa_\AlgS$
according to Step~3a, purged non-solutions, and computed
$\ATab{\PROJ}[t']$ for all nodes~$t'$ below~$t$ and in particular the
$\PROJ$-tables~$\Tab{}$ of the children of~$t$.
Then, we compute the projected model count of a subset~$\sigma$ of the
$\AlgS$-rows in~$\ATab{\AlgS}[t]$, which we formalize 
by applying the generalized inclusion-exclusion principle
to the stored intersection projected model counts of origins.
%

\FIX{%
The idea behind the following definition is that for every origin of~$\sigma$,
we lift the $\ipmc$ counts that are stored in the corresponding child tables.
However, if we sum up these counts, those models that two origins have in common are over-counted,
i.e., they need to be subtracted. But then, those models that three origins have in common
are under-counted, i.e., they need to be (re-)added again.
In turn, the inclusion-exclusion principle ensures that we obtain the correct $\pmc$ value for~$\sigma$.}

\begin{definition}\label{def:pcnt}\label{def:childpcnt}
  \FIX{Let $(F,P)$ be an instance of \PMC, $\TTT=(T, \chi)$ be a tree
  decomposition of~$F$, and~$t$ be a node of~$T$ with~$\ell$ children. }%
  \FIX{Further, let 
  $\Tab{} = \langle \ATab{\PROJ}[t_1], \ldots,
  \ATab{\PROJ}[t_{\ell}]\rangle$ be the sequence of $\PROJ$-tables
  computed by $\dpa_\PROJ((F,P),\TTT,\ATab{\AlgS})$, where
  $\children(t)=\langle t_1, \ldots, t_{\ell}\rangle$ and
  $\sigma \subseteq \ATab{\AlgS}[t]$ is a table.}
  We define the \emph{(inductive) projected model count} of $\sigma$: 
  \begin{align*}
    \pcnt(t,\sigma, \Tab{}) \eqdef & %
                                   \sum_{\emptyset \subsetneq O \subseteq {\origs(t,\sigma)}}\hspace{-1.5em} (-1)^{(\Card{O} - 1)} \cdot
                                 \sipmc(\Tab{}, O), \text{where}\\
                                 \sipmc(\Tab{}, O)\eqdef &\prod_{\substack{i \in \{1, \ldots,
    \ell\},\\\langle O_{(i)}, c \rangle \in \ATab{\PROJ}[t_i]}}\hspace{-2.5em}c\quad \text{is the \emph{stored ipmc} from child tables}.%
  \end{align*}

\end{definition}



Vaguely speaking, $\pcnt$ determines the origins of the
set~$\sigma$ of rows, goes over all subsets of these origins and looks up
the stored counts ($\sipmc$) in the $\PROJ$-tables of the children
of~$t$. 
\FIX{There, we may simply have several child nodes, i.e., nodes 
of type~$\join$, and hence in this case 
we need to multiply the corresponding children's (independent) $\ipmc$ values.}

Example~\ref{ex:pcnt} provides an idea on how to compute the
projected model count of tables of our running example using~$\pcnt$.

\begin{example}\label{ex:pcnt}
  The function defined in Definition~\ref{def:pcnt} allows us to
  compute the projected count for a given~\AlgS-table. Therefore,
  consider again formula~$F$ and tree decomposition~$\TTT'$ from
  Example~\ref{ex:running1} and Figure~\ref{fig:running1}. Say we want
  to compute the projected count $\pcnt(t_5,\{\vec{ u_{5.4}}\}, \Tab{})$
  where
  $\Tab{}\eqdef\allowdisplaybreaks[4] \big\SB \langle \{\vec{
    u_{4.3}}\}, 1\rangle,$
  $\langle \{\vec{ u_{4.6}}\},1\rangle\big\SE$ for
  row~$\vec{ u_{5.4}}$ of table $\tab{5}$. Note that~$t_5$ has
  $\ell=1$ child nodes~$\langle t_4 \rangle$ and therefore the product
  of Definition~\ref{def:childpcnt} consists of only one
  factor. Observe that
  $\origse{\PRIM}(t_5, \vec{ u_{5.4}}) = \{\langle\vec{ u_{4.3}}\rangle, \langle\vec{
    u_{4.6}}\rangle\}$. Since the rows~$\vec{ u_{4.3}}$ and~$\vec{ u_{4.6}}$
  do not occur in the same \PRIM-table of~\Tab{}, only the value of
  $\sipmc$ for the two singleton origin sets~$\{\langle\vec{ u_{4.3}}\rangle\}$ and
  $\{\langle\vec{ u_{4.6}}\rangle\}$ is non-zero; for the remaining set of origins we have zero. Hence, we obtain $\pcnt(t_5,\{\vec{ u_{5.4}}\}, \Tab{})=2$.
\end{example}

\bigskip
\noindent\textbf{\FIX{Computing Intersection Projected Model Counts ($\ipmc$).}}
\noindent Before we present algorithm~$\PROJ$
(Listing~\ref{fig:dpontd3}), we give the the definition allowing us at a
certain node~$t$ to obtain the 
$\ipmc$ value for a given~$\AlgS$-table~$\sigma$ by computing the $\pmc$ (using stored $\ipmc$ values from
$\PROJ$-tables for children of~$t$), and subtracting and adding~$\ipmc$ values for subsets~$\emptyset\subsetneq\rho\subsetneq\sigma$ accordingly.

\FIX{The intuition is that in order to obtain the number of those common projected models, where \emph{every single row} in~$\sigma$
participates, we take all involved projected models of~$\sigma$ and subtract every single row's projected model count ($\ipmc$ values).
There, we subtracted those models that two rows have in common more than once.
Again, these models need to be re-added. Then, the models that three rows have in common are subtracted and so forth. 
In turn, we end up with the intersection projected model count, i.e., those projected models,
where every row of~$\sigma$ is involved. %
}

\begin{definition}\label{def:ipmc}
  Let $\TTT=(T,\chi)$ be a tree decomposition, $t$ be a node of~$T$,
  $\sigma$ be a $\AlgS$-table, and $\Tab{}$ be a sequence of tables.
  Then, we define the \emph{(recursive) $\ipmc$} of~$\sigma$ as follows:
  \springerversion{\vspace{-1.5em}}
  \begin{align*}
    \icnt(t,\sigma,\Tab{})\eqdef
    \begin{cases} %
      1, \text{ if } \type(t) = \leaf,\\
      \big|\pcnt(t,\sigma, \Tab{})\;+ \\ \quad\sum_{\emptyset\subsetneq\rho\subsetneq\sigma}(-1)^{\Card{\rho}}
        \cdot \ipmc(t,\rho, \Tab{})\big|, \text{otherwise.}
    \end{cases}
  \end{align*}
\end{definition}

\noindent In other words, if a node is of type~$\leaf$ the $\ipmc$ is one, since
by definition of a tree decomposition the bags of nodes of
type~$\leaf$ contain only one projected interpretation (the empty set).
%
%
%
%
Otherwise, using Definition~\ref{def:pcnt}, we are able to compute the
$\ipmc$ for a given $\AlgS$-table~$\sigma$, which is by construction the same
as $\ipmc_{\leq t}(\sigma)$ (cf., proof of
Theorem~\ref{thm:correctness} later).
In more detail, we want to compute for a $\AlgS$-table~$\sigma$ its
$\ipmc$ that represents ``all-overlapping'' counts of~$\sigma$ with
respect to set~$P$ of projection variables, that is,
$\ipmc_{\leq t}(\sigma)$. Therefore, for $\ipmc$, we rearrange the
inclusion-exclusion principle.
To this end, we take $\pcnt$, which computes the ``non-overlapping''
count of~$\sigma$ with respect to~$P$, by once more exploiting the
inclusion-exclusion principle on origins of~$\sigma$ (as already discussed) such that
we count every projected model only once. Then we have to alternately subtract and add $\ipmc$ values for strict subsets~$\rho$ of~$\sigma$, accordingly.

\FIX{We provide an example on how this definition is carried out below.}

\bigskip
\noindent\textbf{\FIX{The Table Algorithm~$\PROJ$.}} Finally, Listing~\ref{fig:dpontd3} presents table algorithm~\PROJ,
which stores for given node~$t$ a \PROJ-table consisting of every
\FIX{non-empty subset of equivalence classes for the given table~$\ATabs{\AlgS}{t}$} together with its
$\ipmc$ (as presented above).




  
\begin{algorithm}[t]
  \KwData{ %
    Node~$t$, set~$P$ of projection variables, $\Tab{}$, and 
    $\ATab{\AlgS}$. 
    %
  }%
  \KwResult{Table~$\iota_{t}$ consisting of pairs~$\langle \sigma,
    c\rangle$, where $\sigma \subseteq \ATabs{\AlgS}{t}$ and $c \in
    \NAT$.\hspace{-5em}
  } %
  $\makebox[0em]{}\iota_{t} \leftarrow \big\SB\langle \sigma,
  \icnt(t,\sigma,\Tab{}) \rangle \big{|}\, C\in
  \buckets_P(\ATabs{\AlgS}{t}), \emptyset\subsetneq\sigma \subseteq C\big\SE\hspace{-5em}$ \;
  \Return{$\iota_{t}$}
  \vspace{-0.15em}
  \caption{Table algorithm $\PROJ(t, \cdot, \cdot, P, \Tab{},
    \ATab{\AlgS})$.}
  \label{fig:dpontd3}
\end{algorithm}


%
%


\begin{example}
  Recall instance~$(F,P)$ of \PMC, tree decomposition~$\TTT'$, and 
  tables~$\tab{1}$, $\ldots$, $\tab{12}$ from
  Example~\ref{ex:running0}, \ref{fig:running1}, and
  Figure~\ref{fig:running1}. Figure~\ref{fig:running2} depicts
  selected tables of~$\iota_1, \ldots \iota_{12}$ obtained after
  running~$\dpa_\PROJ$ for counting projected
  interpretations. 
  We assume numbered rows,~i.e., 
  row $i$ in table $\iota_t$ corresponds to
  $\vec{v_{t.i}} = \langle \sigma_{t.i}, c_{t.i} \rangle$.
  Note that for some nodes~$t$, there are rows among
  different~$\PRIM$-tables that occur in~$\Ext_{\leq t}$, but not
  in~$\PExt_{\leq t}$. These rows are removed
  during purging. In fact,
  rows~$\vec{\tabval_{4.1}}, \vec{\tabval_{4.2}}$,
  and~$\vec{\tabval_{4.4}}$ do not occur in table~$\iota_4$. Observe
  that purging is a crucial trick here that avoids to correct
  stored counters~$c$ by backtracking whenever a certain row of a
  table has no succeeding row in the parent table.
  
  Next, we discuss selected rows obtained
  by~$\dpa_\PROJ((F,P),\TTT',\ATab{\PRIM})$. Tables $\iota_1$,
  $\ldots$, $\iota_{12}$ that are computed at the respective nodes of
  the tree decomposition are shown in Figure~\ref{fig:running2}.
  Since~$\type(t_1)= \leaf$, we have
  $\iota_1=\langle\{\langle \emptyset \rangle \}, 1\rangle$.
  Intuitively, up to node~$t_1$ the
  \AlgS-row~$\langle\emptyset\rangle$ belongs to~$1$ equivalence class.
  Node~$t_2$ introduces variable~$a$, which results in
  table~$\iota_2\eqdef\big\SB\langle \{\langle \{a\} \rangle \},
  1\rangle\big\SE$. Note that the $\PRIM$-row
  $\langle\emptyset\rangle$ is subject to purging.
  Node~$t_3$ introduces~$p_1$ and node~$t_4$ introduces~$b$.
  Node~$t_5$ removes projection variable~$p_1$.  The
  row~$\vec{v_{5.2}}$ of $\PROJ$-table~$\iota_5$ has already been
  discussed in Example~\ref{ex:pcnt} and row~$\vec{v_{5.1}}$ works
  similar.
  For row~$\vec{v_{5.3}}$ we compute the
  count~$\ipmc(t_5,\{\vec{\tabval_{5.2}},\vec{\tabval_{5.4}}\},
  \langle \iota_4\rangle)$ by means of~$\pcnt$. Therefore, take
  for~$\rho$ the sets~$\{\vec{\tabval_{5.2}}\}$,
  $\{\vec{\tabval_{5.4}}\}$, and
  $\{\vec{\tabval_{5.2}},\vec{\tabval_{5.4}}\}$.
  For the singleton sets, we simply have
  $\pmc(t_5,\{\vec{\tabval_{5.2}}\}, \langle \iota_4\rangle) = \ipmc(t_5,\{\vec{\tabval_{5.2}}\}, \langle \iota_4\rangle) = c_{5.1}
  = 1$ and
  $\pmc(t_5,\{\vec{\tabval_{5.4}}\}, \langle \iota_4\rangle )= \ipmc(t_5,\{\vec{\tabval_{5.4}}\}, \langle \iota_4\rangle ) =
  c_{5.2}=2$.
  To compute
  $\pmc(t_5,\{\vec{\tabval_{5.2}},\vec{\tabval_{5.4}}\}, \langle
  \iota_4\rangle)$ following Definition~\ref{def:pcnt}, take for~$O$
  the sets~$\{\vec{u_{4.5}}\}$, $\{\vec{u_{4.3}}\}$, and
  $\{\vec{u_{4.6}}\}$ into account, since all other non-empty subsets
  of origins of~$\vec{\tabval_{5.2}}$ and~$\vec{\tabval_{5.4}}$
  in~$\iota_4$ do not occur in~$\iota_4$.
  Then, we take the sum over the values
  $\sipmc(\langle \iota_4\rangle, \{\langle\vec{\tabval_{4.5}}\rangle\})=1$,
  $\sipmc(\langle \iota_4\rangle, \{\langle\vec{\tabval_{4.3}}\rangle\})=1$, and
  $\sipmc(\langle \iota_4\rangle, \{\langle\vec{\tabval_{4.6}}\rangle \})$ $=1$; and
  subtract~$\sipmc(\langle \iota_4\rangle, $
  $\{\langle\vec{\tabval_{4.5}}\rangle, \langle\vec{\tabval_{4.6}}\rangle\})=1$. 
  Hence, 
  $\pmc(t_5,\{\vec{\tabval_{5.2}},\vec{\tabval_{5.4}}\},$ $\langle
  \iota_4\rangle)=2$.
  In order to
  compute~$\ipmc(t_5,\{\vec{\tabval_{5.2}},\vec{\tabval_{5.4}}\},
  \langle \iota_4\rangle) = | %
  \pmc(t_5,\{\vec{\tabval_{5.2}},\vec{\tabval_{5.4}}\},$ $\langle
  \iota_4\rangle)
  - \ipmc(t_5,\{\vec{\tabval_{5.2}}\}, \langle \iota_4\rangle) 
  - \ipmc(t_5,\{\vec{\tabval_{5.4}}\}, \langle \iota_4\rangle) 
  | = |2 -1 - 2| =|-1| = 1$.
  Hence, $c_{5.3} = 1$ represents the number of projected models, both
  rows~$\vec{u_{5.2}}$ and~$\vec{u_{5.4}}$ have in common. We then use it for table~$t_6$.

  For node~$t_{11}$ of type~$\join$ one simply in addition multiplies
  stored $\sipmc$ values for \AlgS-rows in the two children
  of~$t_{11}$ accordingly (see Definition~\ref{def:childpcnt}).
  In the end, the projected model count of~$F$ corresponds to~$\sum_{\langle \sigma, c \rangle\in\iota_{12}} c = c_{12.1} = 4$. 
\end{example}

\begin{figure}[t]
\centering
\begin{tikzpicture}[node distance=0.5mm]
\tikzset{every path/.style=thick}

\node (l1) [stdnode,label={[tdlabel, xshift=0em,yshift=+0em]right:${t_1}$}]{$\emptyset$};
\node (i1) [stdnode, above=of l1, label={[tdlabel, xshift=0em,yshift=+0em]right:${t_2}$}]{$\{a\}$};
\node (i12) [stdnode, above=of i1, label={[tdlabel, xshift=0em,yshift=+0em]right:${t_3}$}]{$\{a,p_1\}$};
\node (i13) [stdnode, above=of i12, label={[tdlabel, xshift=0em,yshift=+0em]right:${t_4}$}]{$\{a,b,p_1\}$};
\node (r1) [stdnode, above=of i13, label={[tdlabel, xshift=0em,yshift=+0em]right:${t_5}$}]{$\{a,b\}$};
\node (r12) [stdnode, above=of r1, label={[tdlabel, xshift=0em,yshift=+0em]right:${t_6}$}]{$\{a\}$};
\node (l2) [stdnode, right=2.5em of i12, label={[tdlabel, xshift=0em,yshift=+0em]left:${t_7}$}]{$\emptyset$};
\node (i2) [stdnode, above=of l2, label={[tdlabel, xshift=0em,yshift=+0em]left:${t_8}$}]{$\{p_2\}$};
\node (i22) [stdnode, above=of i2, label={[tdlabel, xshift=0em,yshift=+0em]left:${t_9}$}]{$\{a,p_2\}$};
\node (r2) [stdnode, above=of i22, label={[tdlabel, xshift=0em,yshift=+0em]left:${t_{10}}$}]{$\{a\}$};
\node (j) [stdnode, above left=of r2, yshift=-0.25em, label={[tdlabel, xshift=0em,yshift=+0.15em]right:${t_{11}}$}]{$\{a\}$};
\node (rt) [stdnode,ultra thick, above=of j, label={[tdlabel, xshift=0em,yshift=+0em]right:${t_{12}}$}]{$\emptyset$};
\node (label) [font=\scriptsize,left=of rt]{${\cal T}'$:};
\node (leaf1) [stdnode, left=1.6em of i1, yshift=-5.5em, label={[tdlabel, xshift=3.25em,yshift=1em]below right:$\iota_{4}$}]{%
	\begin{tabular}{l@{\hspace{0.0em}}r}%
		\multicolumn{1}{l}{$\langle \tuplecolor{\outputPredColor}{\sigma_{4.i}}, $}&\multicolumn{1}{r}{$\tuplecolor{\statePredColor}{c_{4.i}} \rangle$}\\
		\hline\hline
		$\langle\tuplecolor{\outputPredColor}{\{\langle}\tuplecolor{\inputPredColor}{\{a,b\}}\tuplecolor{\outputPredColor}{\rangle\}}, $&$\tuplecolor{\statePredColor}{1}\rangle$\\\specialrule{.1em}{.05em}{.05em}

		$\langle\tuplecolor{\outputPredColor}{\{\langle}\tuplecolor{\inputPredColor}{\{a,p_1\}}\tuplecolor{\outputPredColor}{\rangle\}}, $&$\tuplecolor{\statePredColor}{1}\rangle$\\\hline
		$\langle\tuplecolor{\outputPredColor}{\{\langle}\tuplecolor{\inputPredColor}{\{a,b,p_1\}}\tuplecolor{\outputPredColor}{\rangle\}}, $&$\tuplecolor{\statePredColor}{1}\rangle$\\\hline
		$\langle\tuplecolor{\outputPredColor}{\{\langle}\tuplecolor{\inputPredColor}{\{a,p_1\}}\tuplecolor{\outputPredColor}{\rangle, } \tuplecolor{\outputPredColor}{\langle} \tuplecolor{\inputPredColor}{\{a,b,p_1\}}\tuplecolor{\outputPredColor}{\rangle\}}, $&$\tuplecolor{\statePredColor}{1}\rangle$
	\end{tabular}%
};
\node (leaf1b) [stdnodenum,left=of leaf1,xshift=0.6em,yshift=0pt]{%
	\begin{tabular}{c}%
		\multirow{1}{*}{$i$}\\ 
		\hline\hline
		$1$ \\\specialrule{.1em}{.05em}{.05em}
		$2$ \\\hline
		$3$ \\\hline
		$4$
	\end{tabular}%
};
\node (leaf0x) [stdnode, left=-11.25em of leaf1b, yshift=7.5em, label={[tdlabel, xshift=0em,yshift=-1em]right:$\iota_{5}$}]{%
	\begin{tabular}{l@{\hspace{0.0em}}r}%
		\multicolumn{1}{l}{$\langle \tuplecolor{\outputPredColor}{\sigma_{5.i}}, $}&\multicolumn{1}{r}{$\tuplecolor{\statePredColor}{c_{5.i}} \rangle$}\\
		\hline\hline
		$\langle\tuplecolor{\outputPredColor}{\{\langle}\tuplecolor{\inputPredColor}{\{a\}}\tuplecolor{\outputPredColor}{\rangle\}}, $&$\tuplecolor{\statePredColor}{1}\rangle$\\\hline
		$\langle\tuplecolor{\outputPredColor}{\{\langle}\tuplecolor{\inputPredColor}{\{a,b\}}\tuplecolor{\outputPredColor}{\rangle\}}, $&$\tuplecolor{\statePredColor}{2}\rangle$\\\hline	
		$\langle\tuplecolor{\outputPredColor}{\{\langle}\tuplecolor{\inputPredColor}{\{a\}}\tuplecolor{\outputPredColor}{\rangle, } \tuplecolor{\outputPredColor}{\langle} \tuplecolor{\inputPredColor}{\{a,b\}}\tuplecolor{\outputPredColor}{\rangle\}}, $&$\tuplecolor{\statePredColor}{1}\rangle$\\
	\end{tabular}%
};
\node (leaf0b) [stdnodenum,left=of leaf0x,xshift=0.6em,yshift=0pt]{%
	\begin{tabular}{c}%
		\multirow{1}{*}{$i$}\\ 
		\hline\hline
		$1$ \\\hline
		$2$ \\\hline
		$3$ \\
	\end{tabular}%
};
\node (leaf2b) [stdnodenum,right=3em of j,xshift=-0.25em,yshift=-6.25em]  {%
	\begin{tabular}{c}%
		\multirow{1}{*}{$i$}\\ 
		\hline\hline
		$1$\\\specialrule{.1em}{.05em}{.05em}
		$2$\\
	\end{tabular}%
};
\node (leaf2) [stdnode,right=-0.4em of leaf2b, label={[tdlabel, xshift=0em,yshift=-0.25em]below:$\iota_{9}$}]  {%
	\begin{tabular}{l@{\hspace{0.0em}}r}%
		\multirow{1}{*}{$\langle \tuplecolor{\outputPredColor}{\sigma_{9.i}},$}& $\tuplecolor{\statePredColor}{c_{9.i}} \rangle$\\ 
		\hline\hline
		$\langle \tuplecolor{\outputPredColor}{\{\langle}\tuplecolor{\inputPredColor}{\{a\}}\tuplecolor{\outputPredColor}{\rangle\}},$ &$\tuplecolor{\statePredColor}{1}\rangle$\\\specialrule{.1em}{.05em}{.05em}
		$\langle \tuplecolor{\outputPredColor}{\{\langle}\tuplecolor{\inputPredColor}{\{a,p_2\}}\tuplecolor{\outputPredColor}{\rangle\}}, $& $\tuplecolor{\statePredColor}{1}\rangle$\\
	\end{tabular}%
};
\coordinate (middle) at ($ (leaf1.north east)!.5!(leaf2.north west) $);
\node (join) [stdnode,left=0.75em of r12, yshift=4.5em, label={[tdlabel, xshift=0.1em,yshift=+0.5em]right:$\iota_{{11}}$}] {%
	\begin{tabular}{l@{\hspace{0.0em}}r}%
		\multirow{1}{*}{$\langle \tuplecolor{\outputPredColor}{\sigma_{11.i}},$}& $\tuplecolor{\statePredColor}{c_{11.i}} \rangle$\\
		\hline\hline
		$\langle \tuplecolor{\outputPredColor}{\{\langle}\tuplecolor{\inputPredColor}{\{a\}}\tuplecolor{\outputPredColor}{\rangle\}},$ & $\tuplecolor{\statePredColor}{4}\rangle$\\
	\end{tabular}
};
\node (joinb) [stdnodenum,left=-0.45em of join] {%
	\begin{tabular}{c}
		\multirow{1}{*}{$i$}\\
		\hline\hline
		$1$ \\
	\end{tabular}%
};
\node (leaf0n) [stdnodenum,yshift=-1.5em, right=2.5em of l1] {%
	\begin{tabular}{c}%
		\multirow{1}{*}{$i$}\\ 
		\hline\hline
		$1$
	\end{tabular}%
};
\node (leaf0) [stdnode,right=-0.5em of leaf0n, label={[tdlabel, xshift=0.1em,yshift=0.15em]right:$\iota_{1}$}] {%
	\begin{tabular}{l@{\hspace{0.0em}}r}%
		\multicolumn{1}{l}{$\langle \tuplecolor{\outputPredColor}{\sigma_{1.i}}, $}&$\tuplecolor{\statePredColor}{c_{1.i}} \rangle$\\
		\hline\hline
		$\langle\tuplecolor{\outputPredColor}{\{\langle} \tuplecolor{\inputPredColor}{\emptyset}\tuplecolor{\outputPredColor}{\rangle\}}, $&$\tuplecolor{\statePredColor}{1}\rangle$
	\end{tabular}%
};
\node (joinrrt) [stdnode,right=6.5em of r12, yshift=4.5em, label={[tdlabel, xshift=0.1em,yshift=+0.25em]right:$\iota_{{12}}$}] {%
	\begin{tabular}{l@{\hspace{0.0em}}r}%
		\multirow{1}{*}{$\langle \tuplecolor{\outputPredColor}{\sigma_{12.i}},$} & $\tuplecolor{\statePredColor}{c_{12.i}} \rangle$\\
		\hline\hline
		$\langle \tuplecolor{\outputPredColor}{\{\langle}\tuplecolor{\inputPredColor}{\emptyset}\tuplecolor{\outputPredColor}{\rangle\}},$ &$\tuplecolor{\statePredColor}{4}\rangle$\\
	\end{tabular}
};
\node (joinrbrt) [stdnodenum,left=-0.45em of joinrrt] {%
	\begin{tabular}{c}
		\multirow{1}{*}{$i$}\\
		\hline\hline
		$1$ \\
	\end{tabular}%
};
\node (joinr) [stdnode,right=6.5em of r12, yshift=0.5em, label={[tdlabel, xshift=0.1em,yshift=+0.25em]right:$\iota_{{10}}$}] {%
	\begin{tabular}{l@{\hspace{0.0em}}r}%
		\multirow{1}{*}{$\langle \tuplecolor{\outputPredColor}{\sigma_{10.i}},$} & $\tuplecolor{\statePredColor}{c_{10.i}} \rangle$\\
		\hline\hline
		$\langle \tuplecolor{\outputPredColor}{\{\langle}\tuplecolor{\inputPredColor}{\{a\}}\tuplecolor{\outputPredColor}{\rangle\}},$ &$\tuplecolor{\statePredColor}{2}\rangle$\\
	\end{tabular}
};
\node (joinrb) [stdnodenum,left=-0.45em of joinr] {%
	\begin{tabular}{c}
		\multirow{1}{*}{$i$}\\
		\hline\hline
		$1$ \\
	\end{tabular}%
};
\node (joinl) [stdnode,left=1.6em of r12, yshift=0.5em, label={[tdlabel, xshift=5.9em,yshift=-1em]above left:$\iota_{{6}}$}] {%
	\begin{tabular}{l@{\hspace{0.0em}}r}%
		\multirow{1}{*}{$\langle \tuplecolor{\outputPredColor}{\sigma_{6.i}}, $}&\multirow{1}{*}{$\tuplecolor{\statePredColor}{c_{6.i}} \rangle$}\\
		\hline\hline
		$\langle \tuplecolor{\outputPredColor}{\{\langle}\tuplecolor{\inputPredColor}{\{a\}}\tuplecolor{\outputPredColor}{\rangle\}}, $&$\tuplecolor{\statePredColor}{2}\rangle$\\
	\end{tabular}
};
\node (joinlb) [stdnodenum,left=-0.45em of joinl] {%
	\begin{tabular}{c}
		\multirow{1}{*}{$i$}\\
		\hline\hline
		$1$ \\
	\end{tabular}%
};
\coordinate (top) at ($ (leaf2.north east)+(0.6em,-0.5em) $);
\coordinate (bot) at ($ (top)+(0,-12.9em) $);

\draw [<-] (j) to (rt);
\draw [->] (j) to ($ (r12.north)$);
\draw [->] (j) to ($ (r2.north)$);
\draw [->](r2) to (i22);
\draw [<-](i2) to (i22);
\draw [<-](l2) to (i2);
\draw [<-](l1) to (i1);
\draw [->](i12) to (i1);
\draw [->](i13) to (i12);
\draw [->](r1) to (i13);
\draw [->](r12) to (r1);

\draw [dashed, bend right=30] (joinrrt) to (rt);
\draw [dashed] (j) to (join);
\draw [dashed, bend left=20] (i22) to (leaf2);
\draw [dashed, bend right=5] (i13) to (leaf1);
\draw [dashed, bend left=35] (leaf0) to (l1);
\draw [dashed, bend right=14] (leaf0x) to (r1);
\draw [dashed, bend right=30] (joinr) to (r2);
\draw [dashed, bend right=15] (joinl) to (r12);
\end{tikzpicture}
\caption{Selected tables obtained by~$\dpa_{\algo{PROJ}}$ on
  TD~${\cal T}'$ using~$\dpa_{\PRIM}$ (cf.,
  Figure~\ref{fig:running1}).}
\label{fig:running2}
\end{figure}



\subsection{Correctness of the Algorithm}\label{sec:correct}

In the following, we state definitions required for the correctness
proofs of our algorithm \PROJ. In the end, we only store rows that
are restricted to the bag content to maintain runtime bounds. 
\longversion{In
related work~\cite{SamerSzeider10b}, it was shown that this suffices
for table algorithm~$\AlgS$, i.e., \PRIM 
is both sound and
complete.} Similar to related work~\cite{FichteEtAl17a,SamerSzeider10b}, we proceed in two steps. First, we define properties of
so-called \emph{$\PROJ$-solutions up to~$t$}, and then restrict
these to~\emph{$\PROJ$-row solutions} at~$t$.

\paragraph{Assumptions}

For the following statements, we assume that we have given an arbitrary instance~$(F,P)$ of \PMC and
a tree decomposition~$\TTT = (T,\chi)$ of
formula~$F$, where $T=(N, A)$, node $n=\rootOf(T)$ is the root and $\TTT$ is of width~$k$.
Moreover, for every~$t \in N$ of tree decomposition~$\TTT$, we let
$\ATabs{\AlgS}{t}$ be the tables that have been computed by running
algorithm~$\dpa_\AlgS$ for the dedicated input. Analogously, let
$\ATabs{\PROJ}{t}$ be the tables computed by running~$\dpa_\PROJ$.

%
%

\begin{definition}\label{def:globalsol}
  Let~$\emptyset \subsetneq \sigma \subseteq \ATab{\AlgS}[t]$ be a
  table with $\sigma\subseteq C$ for some $C \in \buckets_P(\ATab{\AlgS}[t])$.
  \hspace{-0.2em}We define a \emph{${\PROJ}$-solution up to~$t$} to be the sequence
  $\langle \hat \sigma\rangle \hspace{-0.15em}=\hspace{-0.15em} \langle\PExt_{\leq t}(\sigma)\rangle$.
\end{definition}

%
%

Next, we recall that we can reconstruct all models from the tables.

\begin{proposition}\label{prop:sat}
  $I(\PExt_{\leq n}(\ATab{\AlgS}[n])) \hspace{-0.1em}=\hspace{-0.1em} I(\Exts) \hspace{-0.1em}=\hspace{-0.1em} \{J \in
    \ta{\var(F)} | \alpha_J\models F\}.$
\end{proposition}
\begin{proof}[Proof (Sketch)]
\shortversion{We use a construction similar to Samer and Szeider~\cite{SamerSzeider10b} and Pichler, R\"ummele, and
  Woltran~\cite[Fig.~1]{PichlerRuemmeleWoltran10}, where we simply collect preceding rows.}
  \longversion{In fact, we can use the construction by Samer and
  Szeider~\cite{SamerSzeider10b} of the tables. Then, the extensions simply collect the corresponding,
  preceding rows. By taking the interpretation parts $I(\cdots)$ of
  these collected rows we obtain the set of all models of the formula.
  A similar construction is used by Pichler, R\"ummele, and
  Woltran~\cite[Fig.~1]{PichlerRuemmeleWoltran10}, which they use in
  a general algorithm to enumerate solutions by means of tables obtained during dynamic programming.}
\end{proof}

Before we present equivalence results between~$\ipmc_{\leq t}(\ldots)$
and the recursive version~$\ipmc(t, \ldots)$
(Definition~\ref{def:ipmc}) used during the computation of
$\dpa_\PROJ$, recall that~$\ipmc_{\leq t}$ and~$\pmc_{\leq t}$
(Definition~\ref{def:pmc}) are key to compute the projected model
count. The following corollary states that computing $\ipmc_{\leq n}$
at the root~$n$ actually suffices to compute the
projected model count~$\pmc_{\leq n}$ of the formula.

\begin{corollary}\label{cor:psat}
    $\ipmc_{\leq n}(\ATab{\AlgS}[n]) = \pmc_{\leq n}(\ATab{\AlgS}[n])
    =$\\ $\Card{I_P(\PExt_{\leq n}(\ATab{\AlgS}[n]))}$
    $=\hspace{-0.15em} \Card{I_P(\Exts)} = \Card{\{J \cap P
       \mid J \in \ta{\var(F)}, \alpha_J\models F\}}$
\end{corollary}
\begin{proof}
  The corollary immediately follows from Proposition~\ref{prop:sat}
  and the observation that $\Card{\ATab{\AlgS}[n]} \leq 1$ by properties of algorithm~$\AlgS$ and
  since $\chi(n) = \emptyset$.
\end{proof}

The following lemma establishes that the \PROJ-solutions up to
root~$n$ of a given tree decomposition solve the \PMC problem.

\begin{lemma}\label{lem:global}
  The
  value~$\sum_{\langle\hat\sigma\rangle\text{ is a \PROJ-solution
      up to } n}\Card{I_P(\hat \sigma)}$ corresponds to the
  projected model count~$c$ of~$F$ with respect to the set~$P$ of
  projection variables.
\end{lemma}
\begin{proof}
  (``$\Longrightarrow$''): 
  Assume
  that~$c = \sum_{\langle\hat\sigma\rangle\text{ is a \PROJ-solution
      up to } n}\Card{I_P(\hat \sigma)}$. Observe that there can be at
  most one projected solution up to~$n$, since~$\chi(n)=\emptyset$. %
  If~$c=0$, then $\ATab{\AlgS}[n]$ contains no rows. Hence, $F$ has no
  models,~cf., Proposition~\ref{prop:sat}, and obviously also no
  models projected to~$P$. Consequently, $c$ is the projected model
  count of~$F$.  
  If~$c>0$ we have by Corollary~\ref{cor:psat} that~$c$ is
  equivalent to the projected model count of~$F$ with respect to~$P$.
%
  %

  \FIX{(``$\Longleftarrow$''): 
  We proceed similar in the if direction. 
Assume that~$c$ is the projected model count of~$F$ and~$P$.
If~$c=0$, we have by Proposition~\ref{prop:sat} that $\PExt_{\leq n}(\ATab{\AlgS}[n])=\emptyset$
and therefore $\ATab{\AlgS}[n]=\emptyset$. As a result for~$c=0$, there does not exist any $\PROJ$-solution up to~$n$.
Otherwise, i.e., if~$c>0$, the result follows immediately by Corollary~\ref{cor:psat}.
}
\end{proof}

In the following, we provide for a given node~$t$ and a given \PROJ-solution up to~$t$, the definition of a \PROJ-row solution at~$t$.

\begin{definition}\label{def:loctab}
Let~$t, t'\in N$ be nodes of a given tree decomposition~${\cal T}$, and $\hat\sigma$ be a~\PROJ-solution up to~$t$. Then, we define \emph{the local table for}~$t'$ as
  $\local(t',\hat\sigma)\eqdef \{ \langle \vec{\tabval}\rangle |$ $
  \langle t', \vec{\tabval}\rangle \in \hat\sigma\}$, and
 if~$t=t'$, the \emph{$\PROJ$-row solution at $t$} by
  $\langle \local(t,\hat\sigma), \Card{I_P(\hat\sigma)}\rangle$.
  %
  %

\end{definition}






\begin{observation}\label{obs:unique}
  Let $\langle \hat \sigma\rangle$ be a \PROJ-solution up to a
  node~$t\in N$.  There is exactly one corresponding \PROJ-row
  solution
  $\langle \local(t,\hat\sigma), \Card{I_P(\hat\sigma)}\rangle$ at~$t$.

  Vice versa, let $\langle \sigma, c\rangle$ be a \PROJ-row
  solution at~$t$ for some integer~$c$. Then, there is exactly one
  corresponding \PROJ-solution~$\langle\PExt_{\leq t}(\sigma)\rangle$
  up to~$t$.
\end{observation}

We need to ensure that storing~$\PROJ$-row solutions at a
node suffices to solve the~\PMC problem, which is necessary
to obtain runtime bounds (cf., Corollary~\ref{cor:runtime}).

\begin{lemma}\label{lem:local}
  Let $t\in N$ be a node of the tree decomposition~$\TTT$.  There is a
  \PROJ-row solution at root~$n$ if and only if the projected
  model count of~$F$ with respect to the set~$P$ of projection variables is larger than~$0$.
\end{lemma}
\begin{proof}%

  (``$\Longrightarrow$''): Let $\langle \sigma, c\rangle$ be a
  \PROJ-row solution at root~$n$ where $\sigma$ is a $\AlgS$-table and
  $c$ is a positive integer. Then, by Definition~\ref{def:loctab},
  there also exists a
  corresponding~$\PROJ$-solution~$\langle \hat\sigma \rangle$ up
  to~$n$ such that $\sigma = \local(n,\hat\sigma)$ and
  $c=\Card{I_P(\hat\sigma)}$.
  Moreover, since~$\chi(n)=\emptyset$, we
  have~$\Card{\ATab{\AlgS}[n]}=1$.  
  Then, by Definition~\ref{def:globalsol},
  $\hat \sigma = \ATab{\AlgS}[n]$. By Corollary~\ref{cor:psat}, we
  have $c=\Card{I_P(\ATab{\AlgS}[n])}$.
  Finally, the claim follows.
  %
%

  \FIX{(``$\Longleftarrow$''): Assume that the projected model count of~$F$ with respect to~$P$ is larger than zero.
Then, by Lemma~\ref{lem:global}, there is at least one \PROJ-solution~$\hat\sigma$ up to the root~$n$.
As a result, by Definition~\ref{def:loctab}, there is also a \PROJ-row solution at~$t$, which is precisely
$\langle \local(n,\hat\sigma), \Card{I_P(\hat\sigma)}\rangle$.
}
\end{proof}

\begin{observation}\label{obs:main_incl_excl}
  Let $X_1$, $\ldots$, $X_n$ be  finite sets. 
  The number~$\Card{\bigcap_{i \in X} X_i}$ is given by
  $\Card{\bigcap_{i \in X} X_i} = \big|\Card{\bigcup^n_{j = 1} X_j} + \sum_{\emptyset \subsetneq I \subsetneq X} (-1)^{\Card{I}} 
                                              \Card{\bigcap_{i \in I} X_i}\big|.$
  \longversion{\begin{align*}
    \Card{\bigcap_{i \in X} X_i} 
    =& %
       \Bigg|\Card{\bigcup^n_{j = 1} X_j} &&- %
                                       \sum_{\emptyset \subsetneq I \subsetneq X, \Card{I}=1}
                                       \Card{\bigcap_{i \in I} X_i} + %
                                       \sum_{\emptyset \subsetneq I \subsetneq X, \Card{I}=2}
                                       \Card{\bigcap_{i \in I} X_i} - \ldots \\
                                         & &&+ \sum_{\emptyset \subsetneq I \subsetneq X, \Card{I}=n-1} (-1)^{\Card{I}} 
                                              \Card{\bigcap_{i \in I} X_i}\Bigg|.
  \end{align*}}
\end{observation}

\begin{lemma}\label{lem:main_incl_excl}
  Let $t\in N$ be a node of the tree decomposition~$\TTT$
  with~$\children(t) = \langle t_1, \ldots, t_\ell \rangle$ and let
  $\langle\sigma,\cdot\rangle$ be a~\PROJ-row solution at~$t$.
  Then, \vspace{-0.75em}
  \begin{enumerate}
  \item %
    $\ipmc(t,\sigma,\langle\ATab{\PROJ}[t_1],\makebox[1em][c]{.\hfil.\hfil.},
    \ATab{\PROJ}[t_{\ell}]\rangle) \hspace{-0.15em}=\hspace{-0.15em} \ipmc_{\leq t}(\sigma)$
  \item %
    If $\type(t) \neq \leaf$: $\pmc(t,\sigma,\langle\ATab{\PROJ}[t_1],\makebox[1em][c]{.\hfil.\hfil.},
    \ATab{\PROJ}[t_{\ell}]\rangle) \hspace{-0.15em}=\hspace{-0.15em} \pmc_{\leq t}(\sigma)$.
  \end{enumerate}
\end{lemma}
\begin{proof}[Proof]
  We prove the statement by simultaneous induction.

  (``Induction Hypothesis''): Lemma~\ref{lem:main_incl_excl} holds for the nodes in~$\children(t)$ and also for node~$t$, but on strict subsets~$\rho\subsetneq\sigma$.

  (``Base Cases''): \FIX{Towards showing the base case of the first claim, let $\type(t) = \leaf$.}
  By definition,
  $\ipmc(t,\emptyset, \langle \rangle) = \ipmc_{\leq t}(\emptyset) =
  1$.
  \FIX{Next, we establish the base case for the second claim. Since~$\type(t)\neq\leaf$, let~$t$
  be a node that has a node~$t'\in N$ with~$\type(t') = \leaf$ as child node. Observe that by definition of~$\mathcal{T}$, $t$ has 
  exactly one child.}
  Then, we have
  $\pmc(t,\sigma,\langle\ATab{\PROJ}[t']\rangle) = \sum_{\emptyset
    \subsetneq O \subseteq {\origs(t,\sigma)}} (-1)^{(\Card{O} - 1)}
  \cdot \sipmc(\langle \ATab{\AlgS}[t']\rangle, O) =
  \Card{\bigcup_{\vec u\in\sigma} I_P(\PExt_{\leq t}($ $\{\vec u\}))} =
  \pmc_{\leq t}(\sigma) = 1$ for \PROJ-row
  solution~$\langle\sigma,\cdot\rangle$ at~$t$.

  (``Induction Step''): 
  %
  %
%
  %
  %
  %
We distinguish two cases.

\FIX{\underline{Case~(i)}: Assume that $\ell=1$.
  Let~$\langle \sigma, c \rangle$ be a \PROJ-row solution at~$t$ for
  some integer~$c$, and~$t'=t_1$.%
%
%
}

\FIX{%
First, we show the second claim on $\pmc$ values.
  By Definition~\ref{def:pcnt}, we have
  $\pmc(t,\sigma,\langle\ATab{\PROJ}[t']\rangle) = \sum_{\emptyset \subsetneq O \subseteq {\origs(t,\sigma)}} (-1)^{(\Card{O} - 1)} \cdot
                                 \sipmc(\ATab{\PROJ}[t'], O)$, which by definition of~$\sipmc$ results in 
$\sum_{\emptyset \subsetneq O \subseteq {\origs(t,\sigma)}} (-1)^{(\Card{O} - 1)} \cdot
                                 \ipmc(t',O,\ATab{\PROJ}[t'])$.
By the induction hypothesis, this evaluates to $\sum_{\emptyset \subsetneq O \subseteq {\origs(t,\sigma)}} (-1)^{(\Card{O} - 1)} \cdot
                                 \ipmc_{\leq t'}(O)$. Then, by the construction based on the 
inclusion-exclusion principle (cf., Observation~\ref{obs:main_incl_excl}), this expression further simplifies to $\pmc_{\leq t'}(\orig(t,\sigma))$.
  %
%
By Definition~\ref{def:pmc}, $\pmc_{\leq t'}(\orig(t,\sigma))=\Card{\bigcup_{\vec u\in\orig(t,\sigma)}
    I_P(\PExt_{\leq t'}(\{\vec u\}))}$. 
However, since by construction of \PROJ, $\Card{\buckets_P(\sigma)} = 1$, i.e., $\sigma$ is contained in one equivalence class, we have $\Card{\bigcup_{\vec u\in\orig(t,\sigma)}
    I_P(\PExt_{\leq t'}(\{\vec u\}))}$ = $\Card{\bigcup_{\vec u\in\sigma}
    I_P(\PExt_{\leq t}(\{\vec u\}))}$.
This corresponds to $\pmc_{\leq t}(\sigma)$ and, 
consequently, $\pmc_{\leq t'}(\orig(t,\sigma))$=$\pmc_{\leq t}(\sigma)$.
%
This
  concludes the proof for the second claim on $\pmc$ values.} 

\FIX{%
The induction step for $\ipmc$ works
  similar.
By Definition~\ref{def:ipmc}, we have
$\ipmc(t,\sigma,\ATab{\PROJ}[t'])=\Card{\pcnt(t,\sigma, \ATab{\PROJ}[t'])} + \sum_{\emptyset\subsetneq\rho\subsetneq\sigma}(-1)^{\Card{\rho}}
        \cdot \ipmc($ $t,\rho, \ATab{\PROJ}[t'])$. 
By the proof on the second claim above, 
$\Card{\pcnt(t,\sigma, \ATab{\PROJ}[t'])}=\pmc_{\leq t}(\sigma)$.
Then, by the induction hypothesis on~$\rho$, we have 
$\ipmc(t,\sigma,\ATab{\PROJ}[t'])=\pmc_{\leq t}(\sigma) + \sum_{\emptyset\subsetneq\rho\subsetneq\sigma}(-1)^{\Card{\rho}}
        \cdot \ipmc_{\leq t}(\rho)$. 
Further, we follow by Definition~\ref{def:pmc} that $\ipmc(t,\sigma,\ATab{\PROJ}[t'])$ corresponds to the expression
$\Card{\bigcup_{\vec u\in\sigma}
    I_P(\PExt_{\leq t}(\{\vec u\}))} + \sum_{\emptyset\subsetneq\rho\subsetneq\sigma}(-1)^{\Card{\rho}}
        \cdot \Card{\bigcap_{\vec u\in\rho}
    I_P(\PExt_{\leq t}(\{\vec u\}))}$.
Finally, by Observation~\ref{obs:main_incl_excl}, this yields $\Card{\bigcap_{\vec u\in\sigma}
    I_P(\PExt_{\leq t}(\{\vec u\}))}$, which simplifies to~$\ipmc_{\leq t}(\sigma)$.
This concludes the proof for the first claim on $\ipmc$ values.
}

\FIX{\underline{Case~(ii)}: Assume that $\ell=2$.
}

\FIX{%
First, we show the induction step on the second claim over~$\pmc$.
By Definition~\ref{def:pcnt}, we have
  $\pmc(t,\sigma,\langle\ATab{\PROJ}[t_1], \ATab{\PROJ}[t_2]\rangle) = \sum_{\emptyset \subsetneq O \subseteq {\origs(t,\sigma)}}$ $(-1)^{(\Card{O} - 1)} \cdot
                                 \sipmc(\langle\ATab{\PROJ}[t_1], \ATab{\PROJ}[t_2]\rangle, O)$. This 
then results in 
$\sum_{\emptyset \subsetneq O \subseteq {\origs(t,\sigma)}} (-1)^{(\Card{O} - 1)} \cdot
                                 \ipmc(t_1,$ $O_{(1)},\ATab{\PROJ}[t_1]) \cdot \ipmc(t_2,O_{(2)},\ATab{\PROJ}[t_2])$.
By the induction hypothesis, this then evaluates to $\sum_{\emptyset \subsetneq O \subseteq {\origs(t,\sigma)}}$ $(-1)^{(\Card{O} - 1)} \cdot
                                 \ipmc_{\leq t_1}(O_{(1)}) \cdot \ipmc_{\leq t_2}(O_{(2)})$.
By expansion via Definition~\ref{def:pmc} and applying Observation~\ref{obs:main_incl_excl}, i.e., the inclusion-exclusion principle, this corresponds to $\Card{\bigcup_{\langle \vec u_1, \vec u_2 \rangle\in\orig(t,\sigma)}
    I_P(\PExt_{\leq t_1}(\{\vec u_1\})) \cdot I_P(\PExt_{\leq t_2}(\{\vec u_2\}))}$.
Since we have that $\Card{\buckets_P(\sigma)} = 1$, i.e., $\sigma$ is contained in one equivalence class and by Definition~\ref{def:satext} of~$\PExt$, this expression simplifies to $\Card{\bigcup_{\vec u\in\sigma}
    I_P(\PExt_{\leq t}(\{\vec u\}))}$.
This corresponds to $\pmc_{\leq t}(\sigma)$, which concludes the proof for~$\pmc$ of Case~(ii).
}%

\FIX{%
The induction step for~$\ipmc$ also works analogously to the proof for~$\ipmc$ of Case~(i), since it does not need to directly consider origins in multiple child nodes. This concludes the proof.%
}
%
%
\end{proof}


\begin{lemma}[Soundness]\label{lem:correct}
  Let $t\in N$ be a node of the tree decomposition~$\TTT$ with
  $\children(t) = \langle t_1, \ldots, t_\ell \rangle$.
  Then, each row~$\langle \tab{}, c \rangle$ at node~$t$ obtained by~$\PROJ$ 
  is a~\PROJ-row solution for~$t$.
\end{lemma}
\begin{proof}[Proof]
  Observe that Listing~\ref{fig:dpontd3} computes a row for each
  subset $\sigma$ with~$\emptyset\subsetneq\sigma\subseteq C$ for some~$C \in \buckets_P(\ATab{\AlgS}[t])$. The
  resulting row~$\langle\sigma, c \rangle$ obtained by~$\ipmc$ is
  indeed a \PROJ-row solution for~$t$ according to
  Lemma~\ref{lem:main_incl_excl}.
\end{proof}


\begin{lemma}[Completeness]\label{lem:complete}
  Let~$t\in N$ be a node of tree decomposition~$\TTT$ where~$\children(t) = \langle t_1, \ldots, t_\ell \rangle$
  and~$\type(t) \neq \leaf$. Given a 
	  \PROJ-row solution~$\langle \sigma, c \rangle$ at~$t$.
  Then, there is $\langle C_1, \ldots, C_\ell\rangle$ where each $C_i$ is a set
  of \PROJ-row solutions at~$t_i$
  with
  $\sigma = \PROJ(t, \cdot, \cdot, P, \langle C_1, \ldots,
  C_\ell\rangle, \ATab{\AlgS})$.
\end{lemma}
\begin{proof}[Proof] 
Since~$\langle\sigma,c \rangle$ is a~\PROJ-row solution for~$t$, there is by Definition~\ref{def:loctab} a corresponding ~\PROJ-solution~$\langle\hat\sigma\rangle$ up to~$t$ such that~$\local(t,\hat\sigma) = \sigma$. 
Then we define~$\hat{\sigma'}\eqdef \{(t'',\hat\rho) \mid (t'', \hat\rho)\in \sigma, t'' \neq t\}$ and proceed again by case distinction. 

\underline{Case (i)}: Assume that~$\ell=1$ and~$t'=t_1$. For each subset~$\emptyset\subsetneq\rho\subseteq\local(t',\hat{\sigma'})$, we define~$\langle \rho, \Card{I_P(\PExt_{\leq t}(\rho))}\rangle$ in accordance with Definition~\ref{def:loctab}. By Observation~\ref{obs:unique}, we have that~$\langle \rho, \Card{I_P(\PExt_{\leq t}(\rho))}\rangle$ is a \AlgS-row solution at~$t'$. 
Since we defined~\PROJ-row solutions for~$t'$ for all respective \PROJ-solutions up to~$t$, we encountered every~\PROJ-row solution for~$t'$ 
required for deriving~$\langle \sigma, c\rangle$ via~\PROJ (cf., Definitions~\ref{def:pcnt} and~\ref{def:ipmc}).
%
%

\underline{Case (ii)}: \FIX{Assume that~$\ell=2$, i.e., $t$ is a join node.
Similarly to above, we define~$\PROJ$-row solutions at~$t_1$ and~$t_2$. 
%
Analogously, we define for each subset~$\emptyset\subsetneq\rho\subseteq\local(t_1,\hat{\sigma'})$, a $\PROJ$-row solution~$\langle \rho, \Card{I_P(\PExt_{\leq t_1}(\rho))}\rangle$ at~$t_1$. Additionally, for each subset~$\emptyset\subsetneq\rho\subseteq\local(t_2,\hat{\sigma'})$, we construct a~$\PROJ$-row solution~$\langle \rho, \Card{I_P(\PExt_{\leq t_2}(\rho))}\rangle$ at~$t_2$ in accordance with Definition~\ref{def:loctab}. By Observation~\ref{obs:unique}, we have that these constructed rows are indeed a \AlgS-row solution at~$t_1$ and a \AlgS-row solution at~$t_2$, respectively. 
Since also for this case we defined~\PROJ-row solutions for~$t_1$ and~$t_2$ for all respective \PROJ-solutions up to~$t$, we encountered every~\PROJ-row solution for~$t_1$ and~$t_2$ 
required for deriving~$\langle \sigma, c\rangle$ via~\PROJ.
This concludes the proof.
}
\end{proof}


\begin{theorem}\label{thm:correctness}
  The algorithm~$\dpa_\PROJ$ is correct. 
  More precisely, 
  %
  $\dpa_\PROJ((F,P),\TTT,$ $\ATab{\AlgS})$ returns
  tables~$\ATab{\PROJ}$ such that $p=\sum_{\langle \sigma, c\rangle \in \ATab{\AlgS}[n]} c$ 
  is the projected model count of~$F$ with respect to the set~$P$ of
  projection variables.
\end{theorem}
\begin{proof}
  %
  By Lemma~\ref{lem:correct} we have soundness for every
  node~$t \in N$ and hence only valid rows as output of table
  algorithm~$\PROJ$ when traversing the tree decomposition in
  post-order up to the root~$n$.
  By Lemma~\ref{lem:local} we know that the projected model count~$p$
  of~$F$ is larger than zero if and only if there exists a
  certain~\PROJ-row solution for~$n$.
  This~\PROJ-row solution at node~$n$ is of the
  form~$\langle \{\langle\emptyset, \ldots\rangle\} ,p\rangle$. If
  there is no \PROJ-row solution at node~$n$,
  then~$\ATab{\AlgS}[n]=\emptyset$ since the table algorithm~$\AlgS$
  is correct (cf., Proposition~\ref{prop:sat}). Consequently, we have
  $p=0$. Therefore, $p=\sum_{\langle \sigma, c\rangle \in \ATab{\AlgS}[n]} c$ 
is the
  pmc of~$F$ w.r.t.~$P$ in both cases.
  %
  %
  
  %
  %

  %
  %

  Next, we establish completeness by induction starting from 
  root~$n$. Let therefore, $\langle \hat\sigma \rangle$ be the~\PROJ-solution up to~$n$, where for each row
  in~$\vec u\in \hat\sigma$, $I(\vec u)$ corresponds to a model of~$F$.  By
  Definition~\ref{def:loctab}, we know that for~$n$ we
  can construct a \PROJ-row solution at~$n$ of the
  form~$\langle \{\langle\emptyset, \ldots\rangle\} , p\rangle$
  for~$\hat\sigma$.  We already established the induction step in
  Lemma~\ref{lem:complete}.
  \longversion{Hence, we obtain some row for every
  node~$t$.}
  Finally, we stop at the leaves.
  %
%
%
\end{proof}

\begin{corollary}\label{cor:correctness}
  The algorithm $\mdpa{\AlgS}$ is correct, i.e., $\mdpa{\AlgS}$ solves~\PMC. 
\end{corollary}
\begin{proof}
  The result follows, since~$\mdpa{\AlgS}$ consists of
  pass~$\dpa_\AlgS$, a purging step and~$\dpa_\PROJ$. For
  correctness of~$\dpa_\AlgS$ we refer to other
  sources~\cite{FichteEtAl17a,SamerSzeider10b}. By Proposition~\ref{prop:sat}, ``purging'' neither destroys soundness nor completeness
  of~$\dpa_\PRIM$.
\end{proof}

\subsection{Runtime Analysis (Upper and Lower Bounds)}\label{sec:complexityresults}
In this section, we first present asymptotic upper bounds on the
runtime of our Algorithm~$\dpa_{\PROJ}$.  For the analysis, we
assume~$\gamma(i)$ to be the costs for multiplying two $i\hy$bit
integers, which can be achieved in time
$i\cdot \log(i) \cdot \log(\log(i))$~\cite{Knuth1998,Harvey2016}. \longversion{Recently, an even faster
  algorithm was published~\cite{Harvey2016}.}

Then, we present a
lower bound that establishes that there cannot be an algorithm that
solves \PMC in time that is only single exponential in the treewidth
and polynomial in the size of the formula unless the exponential time
hypothesis (ETH) fails.
This result establishes that there \emph{cannot} be an algorithm
exploiting treewidth that is \emph{asymptotically better} than our
presented algorithm, although one can likely improve on the analysis
and give a better algorithm.
\longversion{One could for example cache~$\pcnt$ values, which, however, overcomplicates worst-case analysis.}
%



\begin{theorem}
  \label{thm:runtime}
  Given a \PMC instance~$(F,P)$ and a tree
  decomposition~${\cal T} = (T,\chi)$ of~$F$ of width~$k$ with $g$
  nodes. Algorithm~$\dpa_{\PROJ}$ runs in time
  $\mathcal{O}(2^{2^{k+4}}\cdot \gamma(\CCard{F}) \cdot g)$.
\end{theorem}
\begin{proof}
  Let~$d = k+1$ be maximum bag size of~$\TTT$. For each node~$t$ of $T$, we consider
  table $\tau=\ATab{\AlgS}[t]$ which has been computed by
  $\dpa_\AlgS$~\cite{SamerSzeider10b}. The table~$\tau$ has at most
  $2^{d}$ rows.
  In the worst case we store in~$\iota = \ATab{\PROJ}[t]$ each
  subset~$\sigma \subseteq \tau$ together with exactly one
  counter. Hence, we have $2^{2^{d}}$ many rows in $\iota$.
  In order to compute $\ipmc$ for~$\sigma$, we consider every
  subset~$\rho \subseteq \sigma$ and compute~$\pcnt$. Since
  $\Card{\sigma}\leq 2^d$, we have at most~$2^{2^{d}}$ many subsets
  $\rho$ of $\sigma$. 
For computing $\pcnt$, there could be each subset of the origins of~$\rho$ for each child
  table, which are less than~$2^{2^{d+1}}\cdot 2^{2^{d+1}}$ 
  (join and remove case).
  %
  %
  %
  In total, we obtain a runtime bound of~
  $\bigO{2^{2^{d}} \cdot 2^{2^{d}} \cdot 2^{2^{d+1}}\cdot 2^{2^{d+1}}
    \cdot  \gamma(\CCard{F}) } \subseteq \bigO{2^{2^{d+3}} \cdot \gamma(\CCard{F}) }$
  since we also need multiplication of counters.
  Then, we apply this to every node~$t$ of the tree decomposition,
  which results in running
  time~$\bigO{2^{2^{d+3}} \cdot \gamma(\CCard{F})  \cdot g}$.
  %
\end{proof}

\begin{corollary}\label{cor:runtime}
  Given an instance $(F,P)$ of \PMC where $F$ has
  treewidth~$k$. Algorithm~$\mdpa{\AlgS}$ runs in
  time~$\mathcal{O}(2^{2^{k+4}}\cdot \gamma(\CCard{F})  \cdot\CCard{F})$.
\end{corollary}
\begin{proof}
  We compute in time~$2^{\mathcal{O}(k^3)}\cdot\Card{V}$ a tree
  decomposition~${\cal T'}$ of width at most~$k$~\cite{Bodlaender96}
  of primal graph~$\primal{F}$. Then, we run a decision version of the
  algorithm~$\dpa_{\AlgS}$ by Samer and Szeider~\cite{SamerSzeider10b}
  in time~$\mathcal{O}(2^k \cdot \gamma(\CCard{F}) \cdot
  \CCard{F})$. 
  Then, we again traverse the decomposition, thereby keeping
  rows that have a satisfying extension (``purging''), in
  time~$\mathcal{O}(2^k \cdot \CCard{F})$. 
  Finally, we run $\dpa_{\PROJ}$ and obtain the claim by
  Theorem~\ref{thm:runtime} and since~${\cal T'}$ has linearly many nodes~\cite{Bodlaender96}. 
\end{proof}

The next results also establish the lower bounds for our worst-cases.

\begin{theorem}
  Unless ETH fails, $\PMC$ cannot be solved in
  time~$2^{2^{o(k)}}\cdot \CCard{F}^{o(k)}$ for a given instance
  $(F,P)$ where~$k$ is the treewidth of the primal graph of~$F$.
\end{theorem}
\begin{proof}
  Assume for a proof by contradiction that there is such an algorithm.
  %
  %
  \FIX{We show that this contradicts a recent
  result~\cite[Theorem~13]{FichteHecherPfandler20}, which states that one cannot decide
  the validity of a
  \shortversion{QBF~\cite{BiereHeuleMaarenWalsh09,KleineBuningLettman99}}\longversion{quantified
    Boolean formula}
~$Q=\forall
  V_1. \exists V_2. E$ in time~$2^{2^{o(k)}}\cdot \CCard{E}^{o(k)}$
  under ETH. A version of this result for formulas in disjunctive normal form appeared earlier~\cite{LampisMitsou17}.} Given an instance~$(Q,k)$ of~$\forall\exists$-\SAT when
  parameterized by the treewidth~$k$ of~$E$, 
  %
  we
  provide a reduction to an instance~$((F,P,n),k)$ 
  of decision version~$\PMC$-exactly-$n$
  of~$\PMC$ 
  such that $F=E$, $P=V_1$, and the number $n$ of
  solutions is exactly~$2^{\Card{V_1}}$.
  The reduction is \longversion{in fact } an fpt-reduction, since
  the treewidth of $F$ is exactly~$k$.
\FIX{It is easy to see that the reduction
  gives a yes instance~$((F,P,n),k)$ of~$\PMC$-exactly-$n$ 
  if and only if $(\forall{V_1}.\exists V_2.E,k)$ is a yes instance of~$\forall\exists$-\SAT. 
Assume towards a contradiction that~$((F,P,n),k)$ is a yes-instance of~$\PMC$-exactly-$n$,
but $\forall{V_1}.\exists V_2.E$ evaluates to false.
Then, there is an assignment~$\alpha: V_1\rightarrow \{0,1\}$ such that~$E[\alpha]$ evaluates to false,
which contradicts that the projected model count of~$F$ with respect to~$P$ is~$2^{\Card{V_1}}$. 
In the other direction, assume that~$\forall{V_1}.\exists V_2.E$ evaluates to true, but the projected model count of~$F$ and~$P$ is~$<n$.
This, however, contradicts that~$\forall{V_1}.\exists V_2.E$ evalutes to true, which concludes the proof.
}
  %
  %
\end{proof}

\begin{corollary}
  Given any instance $(F,P)$ of \PMC where $F$ has treewidth~$k$. Then, under ETH,
  $\PMC$ 
  requires runtime~$2^{2^{\Theta(k)}} \cdot \poly(\CCard{F})$.
\end{corollary}

\section{Towards Efficiently Utilizing Treewidth for \PMC}\label{sec:nesteddp}

Although the tables obtained via table algorithms might be exponential in size,
the size is bounded by the width of the given TD of 
the primal graph~$\primal{F}$ of a formula~$F$.
Still, practical results of such algorithms show competitive behaviour~\cite{BannachBerndt19,FichteHecherZisser19} up to a certain width. 
As a result, instances with high (tree)width seem out of reach.
Even further, as we have shown above, lifting the table algorithm~$\PRIM$ in order to solve problem~$\cESAT$ results in an algorithm that is double exponential in the treewidth.

To mitigate these issues and to enable practical implementations, 
we present a novel approach to deal with high treewidth,
by nesting of DP on grpah simplifications (abstractions) of~$\primal{F}$.
These abstractions are discussed in Section~\ref{lab:abstractions}
and the basis for nested DP is presented in Section~\ref{lab:nested}.
As we will see, nested dynamic programming not only works for
~$\cSAT$, but 
also for
$\cESAT$ 
with adaptions.
Finally, Section~\ref{sec:hybriddp} concerns about hybrid dynamic programming, which is a further extension of nested DP.
More concretely, hybrid DP tries to combine the best of the two worlds \emph{(i) dynamic programming} and \emph{(ii) applying standard, search-based solvers}, where DP provides the basic structure guidance and delegates hard subproblems that occur during solving to these standard solvers.
%

\subsection{Abstractions are key}\label{lab:abstractions}

In the following, we discuss certain graph simplifications (called abstractions) of the primal graph in the context of the Boolean
satisfiability problem, namely for the problem~\sharpSAT. Afterwards we generalize the usage of
these abstraction to nested dynamic programming 
for \PMC.

To this end, let~$F$ be a Boolean formula.
Now, assume the situation that a set~$U$ of variables of~$F$, called \emph{nesting variables}, appears \emph{uniquely} in the bag of exactly one TD node~$t$ of a tree decomposition of~$\primal{F}$. 
Then, observe that one could do dynamic programming on the tree decomposition as explained in Section~\ref{sec:dpforsat}, 
but no truth value for any variable in $U$ requires to be stored.
\FIX{Instead, clauses of~$F$ over variables~$U$ could be evaluated within node $t$, 
since variables $U$ appear uniquely in the node~$t$.
Indeed, for dynamic programming on the non-nesting variables,
only the result of this evaluation is essential, as variables~$U$ appear uniquely within~$\chi(t)$.}
%

%
Before we can apply nested DP, we require a formal account of abstractions with room for choosing nesting variables between the empty set and the set of all the variables.
Let~$F$ be a Boolean formula and recall the primal graph~$\primal{F}=(\var(F),E)$ of~$F$.
Inspired by related work~\cite{DellRothWellnitz19,EibenEtAl19,stacs:GanianRS17,HecherMorakWoltran20}, we define the \emph{nested primal graph}~$\nested{F}{A}$ for a given formula~$F$
and a given set~$A\subseteq \var(F)$ of variables, referred to by \emph{abstraction variables}.
To this end, we say a path~$P$ in primal graph~$\primal{F}$ is a \emph{nesting path (between~$u$ and~$v$)} using~$A$, if $P=u, v_1, \ldots, v_\ell, v$ ($\ell\geq 0$), and every vertex~$v_i$ is a \emph{nesting variable}, i.e., $v_i\notin A$ for $1\leq i\leq \ell$. 
Note that any path in~$\primal{F}$ is nesting using~$A$ if $A=\emptyset$.
Then, the vertices of nested primal graph~$\nested{F}{A}$ correspond to~$A$ and there is an edge between two distinct vertices~$u,v\in A$ if there is a nesting path between~$u$ and~$v$.

\begin{definition}\label{def:nestprimalgraph}
Let $F$ be a Boolean formula and~$A\subseteq \var(F)$ be a set of variables.
Then, the \emph{nested primal graph~$\nested{F}{A}$} is defined by~$\nested{F}{A}\eqdef (A, \{\{u,v\}\mid u,v\in A, u\neq v,\text{ there is a nesting path in~$\primal{F}$ between~$u$ and~$v$}\})$. 
\end{definition}
Observe that the nested primal graph only consists of abstraction variables and, intuitively, ``hides'' nesting variables of nesting paths of primal graph~$\primal{F}$.
Even further, the connected components of~$\primal{F} - A$ are hidden in the nested primal graph~$\nested{F}{A}$ by means of cliques among~$A$.

\begin{figure}[t]%
\centering
\includegraphics[scale=1.0]{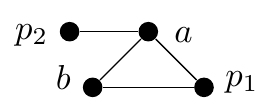}
\hspace{2em}\qquad%
\includegraphics[scale=1.0]{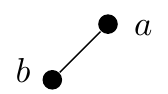}
\hspace{2em}\qquad%
\includegraphics[scale=1.0]{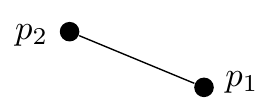}
\caption{Primal graph~$\primal{F}$
of~$F$ from Example~\ref{ex:running0} 
  (left), nested primal graph~$\nested{F}{\{a,b\}}$
  (middle), as well as nested primal graph~$\nested{F}{\{p_1,p_2\}}$ (right).}%
\label{fig:graph-td1}%
\end{figure}


\begin{example}\label{ex:D}
Recall formula~$F\eqdef \{\overbrace{\{\neg a, b, p_1\}}^{c_1},
  \overbrace{\{a, \neg b, \neg p_1\}}^{c_2}, \overbrace{\{a,
    p_2\}}^{c_3}, \overbrace{\{a, \neg p_2\}}^{c_4}\}$
  and primal graph~$\primal{F}$ of Example~\ref{ex:running0}, which is visualized in Figure~\ref{fig:graph-td1} (left).
Given abstraction variables~$A{=}\{a,b\}$, nesting paths of~$\primal{F}$ are, e.g., $P_1{=}a$, $P_2{=}a,p_2$, $P_3{=}p_2,a$, $P_4{=}a,b$, $P_5{=}a,p_1,b$. 
However, neither path~$P_6{=}b,a,p_2$, nor path~$P_7{=}\allowbreak p_2,a,b,p_1$ is nesting using~$A$.
Nested primal graph $\nested{F}{A}$ is shown in Figure~\ref{fig:graph-td1} (middle) and it contains an edge~$\{a,b\}$ over the vertices in~$A$ due to, e.g., paths~$P_4, P_5$.
Assume a different set~$A'=\{p_1,p_2\}$. Observe that~$\nested{F}{A'}$ as depicted in Figure~\ref{fig:graph-td1} (right) consists of the vertices~$A'$ and there is an edge between~$p_1$ and~$p_2$ due to, e.g., nesting path~$P'=p_1,a,p_2$ using~$A'$.
\end{example}

The nested primal graph provides abstractions of needed flexibility for nested DP.
Indeed, if we set abstraction variables to~$A{=}\var(F)$, we end up with full dynamic programming and zero nesting, whereas setting~$A{=}\emptyset$ results in full nesting, i.e., nesting of all variables. 
%
Intuitively, the nested primal graph ensures that
clauses subject to nesting (containing nesting variables) can be safely evaluated in
exactly one node of a tree decomposition of the nested primal graph.
%
%
%
%
%

To formalize this, we assume a tree decomposition~$\TTT=(T,\chi)$ of~$\nested{F}{A}$ and say a set~$U\subseteq \var(F)$ of variables is \emph{compatible} with a node~$t$ of~$T$, and vice versa,
if 
\begin{enumerate}
	\item[(I)] $U$ is a connected component of the graph~$\primal{F}-A$, which is obtained from primal graph~$\primal{F}$ by removing~$A$ and 
	\item[(II)] all neighbor vertices of~$U$ that are in~$A$ are contained in~$\chi(t)$, i.e.,
 $\{a\mid a\in A, u\in U,\text{ there is a nesting path from }a\text{ to }u\text{ using }A\}\subseteq\chi(t)$.
\end{enumerate}

%

If such a set~$U\subseteq \var(F)$ of variables is compatible with a node of~$T$, we say that~$U$ is a \emph{compatible set}.
By construction of the nested primal graph, any nesting variable is in at least one compatible set.
However, 
a compatible set could be compatible with several nodes of~$T$.
Hence, to enable nested evaluation in general, we need to ensure that each nesting variable is evaluated only in one unique node~$t$.

As a result, we formalize for every compatible set~$U$, a \emph{unique} node~$t$ of~$T$ that is compatible with~$U$, denoted by $\compat_{F,A,\mathcal{T}}(U)\eqdef t$. 
We denote the union of all compatible sets~$U$ with $\compat_{F,A,\mathcal{T}}(U)=t$, by \emph{nested bag variables}~$\chi_t^A\eqdef \bigcup_{U: \compat_{F,A,\mathcal{T}}(U)=t} U$.
Then, the \emph{nested bag formula}~$F_t^A$ for a node $t$ of
$T$ equals $F_t^A\eqdef \{c\mid c\in F, \var(c)\subseteq \chi(t) \cup \chi_t^A\}{\,\setminus\,}F_t$, where the bag formula~$F_t$ is defined as in the beginning of Section~\ref{sec:sat}. 
Observe that the definition of nested bag formulas ensures that any connected component~$U$ of $\primal{F} - A$ ``appears''
among nested bag variables of some unique node of~$T$.
Consequently, each variable $a \in \var(F) \setminus A$
appears \emph{only} in one nested bag formula~$F_t^A$ of a node $t$ of~$T$ that is unique for~$a$.

\begin{figure}[t]%
\centering
\includegraphics[scale=1.0]{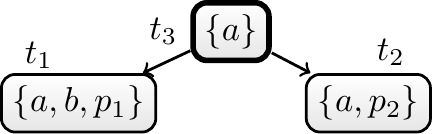}
\hspace{2em}\qquad%
\includegraphics[scale=1.0]{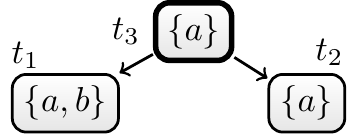}
\caption{TD~$\mathcal{T}$ (left) of the primal graph~$\primal{F}$ of Figure~\ref{fig:graph-td1}, and a TD~$\mathcal{T'}$ (right) of nested primal graph~$\nested{F}{\{a,b\}}$.} 
\label{fig:graph-td2}%
\end{figure}

\begin{example}\label{ex:abstracttd}
Recall formula~$F$, the tree decomposition~$\TTT=(T,\chi)$ of~$\primal{F}$, as depicted in Figure~\ref{fig:graph-td2} (left), and abstraction variables~$A=\{a,b\}$ of Example~\ref{ex:D}.
Consider TD~$\TTT'\eqdef (T,\chi')$, where~$\chi'(t)\eqdef \chi(t)\cap\{a,b\}$
for each node~$t$ of~$T$, which is given in Figure~\ref{fig:graph-td2} (right). Observe that~$\TTT'$ is~$\TTT$, but restricted to~$A$
and that $\TTT'$ is a TD of~$\nested{F}{A}$ of width~$1$.
There are two compatible sets, namely~$\{p_1\}$ and~$\{p_2\}$.
Observe that only for compatible set~$U=\{p_2\}$ we have two nodes compatible with~$U$, namely $t_2$ and $t_3$. 
We assume that~$\compat_{F,A,\mathcal{T}'}(U)=t_2$, i.e., we decide that~$t_2$ shall be the unique node for~$U$. 
Consequently, nested bag formulas are~$F_{t_1}^A=\{c_1,c_2\}$, $F_{t_2}^A=\{c_3,c_4\}$, and~$F_{t_3}^A{=}\emptyset$.
\end{example}

\subsection{Nested Dynamic Programming on Abstractions}\label{lab:nested}

\begin{algorithm}[t]%
  \KwData{%
    Nested table algorithm $\algo{N}$, nesting $\depth\geq 0$, instance~$(F,P)$ of~$\PMC$, abstraction variables~$A\subseteq \var(F)$, 
and a TD~$\TTT=(T,\chi)$ of the nested primal graph~$\nested{F}{A}$. 
  }%
  \KwResult{%
    Table mapping $\ATab{\algo{N}}$, which maps each TD node~$t$ of~$T$ to some computed
    table~$\tau_t$. 
  } %

  $\ATab{\algo{N}} \leftarrow \{\}\qquad \tcc*[h]{empty mapping}$

  \For{\text{\normalfont iterate} $t$ in $\post(T)$}{
    %
    %
    
    $\Tab{} \leftarrow \langle\ATab{\algo{N}}[t_1],\ldots, \ATab{\algo{N}}[t_\ell]\rangle\text{ where }\childrenseq(t)=\langle t_1,\ldots,t_\ell \rangle\hspace{-5em}$
  %
  
  $\ATab{\algo{N}}[t]  \leftarrow {\algo{N}}(\depth,t,\chi(t), F_t, F_t^A, P\cap\var(F_t^A), \Tab{})$
%
  }%
  \Return{$\ATab{\algo{N}}$} 
 \caption{Algorithm ${\adpa}_{\algo{N}}(\depth, (F,P), A,\TTT)$ for computing solutions of~$(F,P)$ via nested DP on
   TD~${\calT}$.}
\label{fig:adpontd}
\end{algorithm}%

Now, we have established required notation in order to discuss \emph{nested dynamic programming (nested DP)}.
Listing~\ref{fig:adpontd} presents algorithm~\adpa for solving a given problem by means of nested dynamic programming.
Observe that Listing~\ref{fig:adpontd} 
is almost identical to algorithm~\dpa as presented in Listing~\ref{fig:dpontd}.
The reason for this is that nested dynamic programming can be seen as a refinement of dynamic programming, cf.\ algorithm~$\dpa$
of Listing~\ref{fig:dpontd}.
Indeed, the difference of~\adpa compared to~\dpa is that~\adpa uses labeled tree decompositions of the nested primal graph and that it gets as additional parameter a set~$A$ of abstraction variables.
Further, instead of a table algorithm~$\AlgA$, algorithm~\adpa relies on a \emph{nested table algorithm}~$\algo{N}$ during dynamic programming, which is similar 
to a table algorithm that gets as additional parameter an integer $\depth\geq 0$ that will be used later and a nested bag instance that needs to be evaluated. 
For simplicity and generality, also the formula is passed as a parameter, which is, however, used only for passing problem-specific information of the instance. 
Indeed, most nested table algorithm do not require this parameter, which should not be used for direct problem solving instead of utilizing the bag instance.
Consequently, nested dynamic programming still follows the basic concept of dynamic programming as presented in Figure~\ref{fig:framework}.

Similar to above, for the ease of presentation
	our presented nested table algorithms use nice tree decompositions only.
However, this is not a hard restriction.
Indeed, it is easy to see that for arbitrary TDs the clear case distinctions 
of nice decompositions
are still valid, but are in general just overlapping.
Further, without loss of generality we also assume that each compatible set~$U$ gets assigned a unique node~$t=\compat_{F,A,\mathcal{T}}(U)$ that is an \emph{introduce node}, i.e., $\type(t)=\intr$.

\paragraph{Nested Dynamic Programming for \sharpSAT} 
In order to design a nested table algorithm for \sharpSAT, assume a Boolean formula~$F$ as well as a given labeled tree decomposition~$\TTT=(T,\chi)$ of~$\nested{F}{A}$ using any set~$A$ of abstraction variables.
Recall from the discussions above, that each variable $a \in \var(F) \setminus A$
appears \emph{only} in one nested bag formula~$F_t^A$ of a node $t$ of~$T$ that is unique for~$a$.
These unique variable appearances allow us to actually nest
the evaluation of nested bag formula $F_t^A$.
This evaluation is performed by a nested table algorithm~$\algNS$ in the context of nested dynamic programming.
Listing~\ref{alg:primhyb} shows this simple nested table algorithm~$\algNS$ for solving problem~\sharpSAT by means of algorithm~$\adpa_\algNS$.
For comparison, recall table algorithm~$\algS$ for solving problem~$\sharpSAT$ by means of dynamic programming,
as given by Listing~\ref{alg:prim}.
Observe that in contrast to Listing~\ref{alg:prim}, we store here assignments (and not interpretations), which simplifies the presentation of nesting.
However, the main difference of~$\algNS$ compared to~$\algS$ is that the nested table algorithm~$\algNS$ maintains a counter~$c$ and that it gets called on a nested primal graph, i.e., the algorithm gets additional parameters like the nested bag formula.
%
%
Then, the nested table algorithm evaluates this nested bag formula 
in Line~\ref{line:intr:hybr} via any procedure~$\sharpSAT$ for solving
problem~$\sharpSAT(F_t^A[J])$ on the nested bag formula~$F_t^A$ simplified by the current assignment~$J$ to variables in the bag~$\chi(t)$.
Note that this subproblem~$\sharpSAT(F_t^A[J])$ itself can be solved by again using nested dynamic programming with the help of algorithm~$\adpa_\algNS$.

\begin{algorithm}[t]
  \KwData{%
Node~$t$, 
bag $\chi_t$, bag formula~$F_t$, 
  nested bag formula~$F_t^A$,
    and sequence $\Tab{}=\langle \tab{1},\ldots \tab{\ell}\rangle$ of child tables  of~$t$.
}
\KwResult{Table~$\tab{t}.$} \lIf(\hspace{-1em})
  {$\type(t) = \leaf$}{%
    $\tab{t} \leftarrow \{ \langle
    \tuplecolor{\specialPredColor}{\emptyset},
    \tuplecolor{\statePredColor}{1} \rangle \}$%
  }%
  \uElseIf{$\type(t) = \intr$ and
    $a\hspace{-0.1em}\in\hspace{-0.1em}\chi(t)$ is introduced}{ %
    \makebox[1.7cm][l]{$\tab{t} \leftarrow \{ \langle
      \tuplecolor{\specialPredColor}{J},
      \tuplecolor{\statePredColor}{c'\cdot c} \rangle$
    }%
    \hspace{2em}$|\;
    \langle \tuplecolor{\specialPredColor}{I},
    \tuplecolor{\statePredColor}{c} \rangle \in \tab{1}, 
    %
    J\in I\cup\{a\mapsto v\}, v\in\{0,1\}, J \models F_t,c'{>}0,c'{=}$\newline
    \makebox[26em]{}$\sharpSAT(F_t^A[J])  \}\hspace{-4em}$  
    %
     \label{line:intr:hybr}%
  }\vspace{-0.05em}%
  \uElseIf{$\type(t) = \rem$ and $a \not\in \chi(t)$ is removed}{%
    \tcc{$\text{C}(I)$ is the set that contains the rows in~$\tab{1}$ for assignments $J$ that are equal to $I$ after removing~$a$}
    $\text{C}(I) \leftarrow \{ \langle J, c\rangle \SM \langle {J}, {c}\rangle \in \tab{1}, J\setminus\{a \mapsto 0, a \mapsto 1\} = I\setminus\{a \mapsto 0, a \mapsto 1\} \SE$\\
\makebox[5cm][l]{$\tab{t} \leftarrow \{ \langle
      \tuplecolor{\specialPredColor}{I\setminus\{a \mapsto 0, a \mapsto 1\}}$,
      %
    %
      \tuplecolor{\statePredColor}{%
        $\sum_{\langle J, c\rangle \in C(I)} c \rangle\}$
      }%
    }
    \hspace{3.825em}$\mid \langle \tuplecolor{\specialPredColor}{I},
    \tuplecolor{\statePredColor}{\cdot} \rangle \in \tab{1}     \}\hspace{-5em}$\;
  } %
  \uElseIf{$\type(t) = \join$}{%
    \makebox[3.3cm][l]{$\tab{t} \leftarrow \{ \langle
      \tuplecolor{\specialPredColor}{I},
      \tuplecolor{\statePredColor}{c_1 \cdot c_2}
      \rangle$}\hspace{9em}$|\;\langle \tuplecolor{\specialPredColor}{I},
    \tuplecolor{\statePredColor}{c_1} \rangle \in \tab{1}, \langle
    \tuplecolor{\specialPredColor}{I},
    \tuplecolor{\statePredColor}{c_2} \rangle \in \tab{2}
    \}\hspace{-5em}$
    \vspace{-0.25em}
  } %
  \Return $\tab{t}$ \vspace{-0.25em}
  \caption{Nested table algorithm~$\algNS(\cdot,t,\chi_t,F_t,F_t^A,\cdot,\Tab{})$ for solving \sharpSAT.}
  \label{alg:primhyb}
\end{algorithm}%

In the following, we briefly show the evaluation of nested dynamic programming for \sharpSAT on an example.

\begin{figure}[t]
\centering
\begin{tikzpicture}[node distance=3mm]
\tikzset{every path/.style=thick}

%
\node (r12) [stdnode, label={[tdlabel, xshift=0em,yshift=+0em]right:${t_1}$}]{$\{a,b\}$};
\node (r2) [stdnode, right=2.5em of r12, label={[tdlabel, xshift=0em,yshift=+0em]left:${t_{2}}$}]{$\{a\}$};
\node (j) [stdnode, ultra thick, above left=of r2, xshift=-.2em, yshift=-0.25em, label={[tdlabel, xshift=0em,yshift=+0.15em]right:${t_{3}}$}]{$\{a\}$};
\node (label) [font=\scriptsize,left=of j]{${\mathcal{T}}'$:};
\node (leaf0x) [stdnode, left=1em of r12, yshift=-1.35em, label={[tdlabel, xshift=0em,yshift=+.5em]below right:$\tab{1}$}]{%
	\begin{tabular}{l}%
		\multicolumn{1}{l}{$\langle \tuplecolor{\inputPredColor}{a, b}, \tuplecolor{\statePredColor}{\text{cnt}} \rangle$}\\
		\hline\hline
		$\langle \tuplecolor{\inputPredColor}{0, 0}, \tuplecolor{\statePredColor}{2}\rangle$\\\hline
		\rowcolor{yellow}$\langle\tuplecolor{\inputPredColor}{1, 0}, \tuplecolor{\statePredColor}{1}\rangle$\\\hline
		$\langle\tuplecolor{\inputPredColor}{0, 1}, \tuplecolor{\statePredColor}{1}\rangle$\\\hline
		\rowcolor{yellow}$\langle\tuplecolor{\inputPredColor}{1, 1}, \tuplecolor{\statePredColor}{2}\rangle$\\
	\end{tabular}%
};
\node (leaf0b) [stdnodenum,left=of leaf0x,xshift=0.9em,yshift=0pt]{%
	\begin{tabular}{c}%
		\multirow{1}{*}{$i$}\\ %
		\hline\hline
		$1$ \\\hline
		$2$ \\\hline
		$3$ \\\hline
		$4$ %
	\end{tabular}%
};
\node (leaf2b) [stdnodenum,below=2em of j,xshift=-2.25em,yshift=-.5em]  {%
	\begin{tabular}{c}%
		\multirow{1}{*}{$i$}\\ %
		\hline\hline
		$1$\\
	\end{tabular}%
};
\node (leaf2) [stdnode,right=-0.4em of leaf2b,yshift=-0em, label={[tdlabel, xshift=0em,yshift=0.25em]right:$\tab{3}$}]  {%
	\begin{tabular}{l}%
		\multirow{1}{*}{$\langle \tuplecolor{\inputPredColor}{a}, \tuplecolor{\statePredColor}{\text{cnt}} \rangle$}\\ %
		\hline\hline
		\rowcolor{yellow}$\langle \tuplecolor{\inputPredColor}{1}, \tuplecolor{\statePredColor}{6}\rangle$\\
	\end{tabular}%
};
\node (leaf3b) [stdnodenum,below=1.1em of r2,xshift=2.25em,yshift=.5em]  {%
	\begin{tabular}{c}%
		\multirow{1}{*}{$i$}\\ %
		\hline\hline
		$1$\\
	\end{tabular}%
};
\node (leaf3) [stdnode,right=-0.4em of leaf3b,yshift=0em, label={[tdlabel, xshift=0em,yshift=0.25em]right:$\tab{2}$}]  {%
	\begin{tabular}{l}%
		\multirow{1}{*}{$\langle \tuplecolor{\inputPredColor}{a}, \tuplecolor{\statePredColor}{\text{cnt}} \rangle$}\\ %
		\hline\hline
		\rowcolor{yellow}$\langle \tuplecolor{\inputPredColor}{1}, \tuplecolor{\statePredColor}{2}\rangle$\\
	\end{tabular}%
};
%
\coordinate (top) at ($ (leaf2.north east)+(0.6em,-0.5em) $);
\coordinate (bot) at ($ (top)+(0,-12.9em) $);

\draw [->] (j) to ($ (r12.north)$);
\draw [->] (j) to ($ (r2.north)$);

\draw [dashed, bend right=0] (leaf2) to (j);
\draw [dashed, bend right=0] (leaf3) to (r2);
%
\draw [dashed, bend left=20] (leaf0x) to (r12);
\end{tikzpicture}\qquad
\begin{tikzpicture}[node distance=1mm]
\tikzset{every path/.style=thick}

%
\node (i1) [stdnode, ultra thick, label={[tdlabel, xshift=0em,yshift=+0em]right:${t''}$}]{$\{a\}$};
\node (label) [font=\scriptsize,left=of i1]{${\mathcal{T}}''$:};
%
\node (join) [stdnode,below=0.75em of i1, yshift=-.75em, label={[tdlabel, xshift=.0em,yshift=+0.25em]right:$\tab{{t''}}$}] {%
	\begin{tabular}{l}%
		\multirow{1}{*}{$\langle \tuplecolor{\inputPredColor}{a}, \tuplecolor{\statePredColor}{\text{cnt}} \rangle$}\\
		\hline\hline
		\rowcolor{yellow}$\langle \tuplecolor{\inputPredColor}{1}, \tuplecolor{\statePredColor}{6} \rangle$\\
	\end{tabular}%
};
\node (joinb) [stdnodenum,left=-0.45em of join] {%
	\begin{tabular}{c}
		\multirow{1}{*}{$i$}\\
		\hline\hline
		$1$\\
	\end{tabular}%
};
\coordinate (top) at ($ (leaf2.north east)+(0.6em,-0.5em) $);
\coordinate (bot) at ($ (top)+(0,-12.9em) $);

\draw [dashed, bend right=25] (i1) to (join);
%
%
%
\end{tikzpicture}
\caption{Selected tables obtained by nested DP on TD~${\mathcal{T}}'$ of~$\nested{F}{\{a,b\}}$ (left) and on TD $\TTT''$ of~$\nested{F}{\{a\}}$ (right) for~$F$ of Example~\ref{ex:abstracttd} via~$\adpa_{\algNS}$.}
\label{fig:nested}
\end{figure}

\begin{example}\label{ex:nested}
Recall formula~$F$, set~$A$ of abstraction variables, and TD~$\TTT'$ of nested primal graph~$\nested{F}{A}$ given in Example~\ref{ex:abstracttd}.
As already mentioned, Formula $F$ has six satisfying assignments, namely $\{a\mapsto 1, b\mapsto 0, p_1 \mapsto 1, p_2\mapsto 0\}$, $\{a\mapsto 1, b\mapsto 0, p_1 \mapsto 1, p_2\mapsto 1\}$, $\{a\mapsto 1, b\mapsto 1, p_1 \mapsto 0, p_2\mapsto 0\}$, $\{a\mapsto 1, b\mapsto 1, p_1 \mapsto 0, p_2\mapsto 1\}$, $\{a\mapsto 1, b\mapsto 1, p_1 \mapsto 1, p_2\mapsto 0\}$, and $\{a\mapsto 1, b\mapsto 1, p_1 \mapsto 1, p_2\mapsto 1\}$.

Figure~\ref{fig:nested} (left) shows TD~$\TTT'$ of~$\nested{F}{A}$ and 
tables obtained by~$\adpa_{\algNS}(0,(F,\cdot),\allowbreak A, \TTT')$ for model counting (\sharpSAT) on~$F$. 
%
%
%
%
We briefly discuss executing~$\algNS$ 
on~$\TTT'$, resulting in tables~\tab{1}, \tab{2}, and \tab{3} as shown in Figure~\ref{fig:nested} (left).
%
%
Intuitively, table~$\tab{1}$ is the result of introducing variables~$a$ and~$b$. 
Recall from Example~\ref{ex:abstracttd} that~$F_{t_1}^A=\{c_1,c_2\}$ with
$c_1=\{\neg a, b, p_1\}$ and $c_2=\{a, \neg b, \neg p_1\}$.
Then, in Line~\ref{line:intr:hybr} of algorithm~$\algNS$, 
for each assignment~$I$ to~$\{a,b\}$ of each row~$r$ of~\tab{1}, 
we compute $\sharpSAT(F_{t_1}^A[I])$.
Consequently, for assignment~$I_1=\{a\mapsto 0, b\mapsto 0\}$,
we have that there are two satisfying assignments of~$F_{t_1}^A[I_1]$, namely
$\{p_1\mapsto 0\}$ and~$\{p_1\mapsto 1\}$.
Indeed, this count of~$2$ is obtained for the first row of table~$\tab{1}$ by Line~\ref{line:intr:hybr}.
Analogously, one can derive the remaining tables of~$\tab{1}$ and one obtains table~$\tab{2}$ similarly,
by using formula~$F_{t_2}^A$.
Then, table~$\tab{3}$ is the result of removing~$b$ in node~$t_1$ and combining agreeing 
assignments of rows accordingly.
Consequently, we obtain that there are six satisfying assignments of~$F$,
which are all required to set~$a$ to~$1$ due to formula~$F_{t_2}^A$ that is evaluated in node~$t_2$.
%
%

Figure~\ref{fig:nested} (right) shows TD~$\TTT''$ of~$\nested{F}{\{a\}}$ and tables obtained by $\adpa_{\algNS}(0,\allowbreak(F,\cdot),\allowbreak \{a\}, \TTT'')$ using TD~$\TTT''$.
Since~$F_{t''}^{\{a\}}=F$ and~$F[\{a\mapsto 0\}]$ is unsatisfiable, table~$\tab{t''}$ does not contain an entry corresponding to assignment~$\{a\mapsto 0\}$, cf.\ Condition~``$c'{>}0$'' in Line~\ref{line:intr:hybr} of Listing~\ref{alg:primhyb}.
%
%
Thus, there are six satisfying assignments of~$F_{t''}^{\{a\}}[\{a\mapsto 1\}]$ obtained by computing~$\sharpSAT(F_{t''}^{\{a\}}[\{a\mapsto 1]\})$.
\end{example}

While the overall concept of nested dynamic programming as given by algorithm~\adpa of Listing~\ref{fig:adpontd} is quite general, sometimes in practice it is sufficient to further restrict the set of choices for abstraction vertices~$A$ when constructing the nested primal graph.

\paragraph{Nested Table Algorithm for \PMC} 
To this end, we show the approach of nested dynamic programming for the problem~\PMC.
\begin{example}\label{ex:pmc}
Recall formula~$F$ as well as set~$A=\{a,b\}$ of abstraction variables from Example~\ref{ex:abstracttd}.
Then, we have that~$(F,A)$ is an instance of the projected model counting problem~$\PMC$.
Restricted to projection set~$A$, the Boolean formula $F$ has two satisfying assignments, 
namely $\{a\mapsto 1, b\mapsto 0\}$ and $\{a\mapsto 1, b\mapsto 1\}$.
Consequently, the solution to~$\cESAT$ on~$(F,A)$, i.e., $\cESAT(F,A)$, is~$2$.
\end{example}

Indeed, for solving projected model counting we mainly focus on the case, where for a given instance~$(F,P)$ with Boolean formula~$F$ of problem~\PMC, the abstraction variables~$A$ that are used for constructing the \emph{nested primal graph}~$\nested{F}{A}$ are \emph{among the projection variables}, i.e., $A\subseteq P$.
The approach of nested DP can then be applied for solving projected model counting such that the nested table algorithm naturally extends algorithm~\algNS of Listing~\ref{alg:primhyb}.
%
%


\begin{algorithm}[t]
  \KwData{%
Node~$t$, 
bag $\chi_t$, bag formula~$F_t$, 
  nested bag formula~$F_t^A$,
projection variables $P\subseteq\var(F_t^A)$,
    and sequence $\Tab{}=\langle \tab{1},\ldots \tab{\ell}\rangle$ of child tables  of~$t$.
}
\KwResult{Table~$\tab{t}.$} \lIf(\hspace{-1em})
  {$\type(t) = \leaf$}{%
    $\tab{t} \leftarrow \{ \langle
    \tuplecolor{\specialPredColor}{\emptyset},
    \tuplecolor{\statePredColor}{1} \rangle \}$%
  }%
  \uElseIf{$\type(t) = \intr$ and
    $a\hspace{-0.1em}\in\hspace{-0.1em}\chi(t)$ is introduced}{ %
    \makebox[3.3cm][l]{$\tab{t} \leftarrow \{ \langle
      \tuplecolor{\specialPredColor}{J},
      \tuplecolor{\statePredColor}{c'\cdot c} \rangle$
    }%
    \hspace{-0.5em}$|\;
    \langle \tuplecolor{\specialPredColor}{I},
    \tuplecolor{\statePredColor}{c} \rangle \in \tab{1}, 
    %
    = I\cup\{a\mapsto v\}, v\in\{0,1\}, J \models F_t, c'{>}0,c'{=}$\newline  
    \makebox[26em]{}$\PMC(F_t^A[J], P) \}\hspace{-1em}$
    %
    %
     \label{line:intr:hybrpmc}%
  }\vspace{-0.05em}%
  \uElseIf{$\type(t) = \rem$ and $a \not\in \chi(t)$ is removed}{%
    %
    %
    $\text{C}(I) \leftarrow \{ \langle J, c\rangle \SM \langle {J}, {c}\rangle \in \tab{1}, J\setminus\{a \mapsto 0, a \mapsto 1\} = I\setminus\{a \mapsto 0, a \mapsto 1\} \SE$\\
\makebox[5cm][l]{$\tab{t} \leftarrow \{ \langle
      \tuplecolor{\specialPredColor}{I\setminus\{a \mapsto 0, a \mapsto 1\}}$,
      %
    %
      \tuplecolor{\statePredColor}{%
        $\sum_{\langle J, c\rangle \in C(I)} c $}$\rangle$
    }
 \hspace{3.5em}$\mid \langle \tuplecolor{\specialPredColor}{I},
    \tuplecolor{\statePredColor}{\cdot} \rangle \in \tab{1}    \}\hspace{-15em}$\;
  } %
  \uElseIf{$\type(t) = \join$}{%
    \makebox[3.3cm][l]{$\tab{t} \leftarrow \{ \langle
      \tuplecolor{\specialPredColor}{I},
      \tuplecolor{\statePredColor}{c_1 \cdot c_2}
      \rangle$}\hspace{9em}$|\;\langle \tuplecolor{\specialPredColor}{I},
    \tuplecolor{\statePredColor}{c_1} \rangle \in \tab{1}, \langle
    \tuplecolor{\specialPredColor}{I},
    \tuplecolor{\statePredColor}{c_2} \rangle \in \tab{2}
    \}\hspace{-5em}$
    \vspace{-0.25em}
  } %
  \Return $\tab{t}$ \vspace{-0.25em}
  \caption{Nested table algorithm~$\algNES(\cdot, t,\chi_t, F_t, F_t^A, P, \Tab{})$ for solving \PMC.}
  \label{alg:primhybpmc}
\end{algorithm}%

The nested table algorithm~\algNES for solving projected model counting via nested dynamic programming is presented in Listing~\ref{alg:primhybpmc}.
Observe that nested table algorithm~\algNES does not significantly differ
from algorithm~\algNS due to~$A\subseteq P$.
Indeed, the main difference is only in Line~\ref{line:intr:hybrpmc} of Listing~\ref{alg:primhyb}, where instead of a procedure for model counting, a procedure~\PMC for solving a projected model counting question is called.
%

\subsection{Hybrid Dynamic Programming based on nested DP}\label{sec:hybriddp}\label{sec:hybrid}

Now, we have definitions at hand to further refine and discuss nested dynamic programming in the context of \emph{hybrid dynamic programming (hybrid DP)}, which combines using both standard solvers and parameterized solvers exploiting treewidth in the form of nested dynamic programming. 
We illustrate these ideas for 
the problem~\cESAT next. 
Afterwards we discuss how to implement the resulting algorithms in order to efficiently solve \PMC and \sharpSAT by 
means of database management systems. 

\begin{algorithm}[t]%
  \KwData{%
    Nesting~$\depth \geq 0$ and an instance~$(F,P)$ of \PMC.} 
  \KwResult{Number $\cESAT(F,P)$ of assignments.\hspace{-1em}
  }%
  %
    %

  \smallskip
  

  $(F',P')=\text{Preprocessing}(F,P)$ \label{line:bcp}


  $A \leftarrow P'$\label{line:updateD}

  \vspace{-0.35em} 
  \lIf{$F'\in\dom(\cache)\quad\tcc*[h]{\hspace{-0.5em}Cache Hit\hspace{-0.5em}}\qquad$}{\Return{$\cache(F')\cdot 2^{\Card{P\setminus P'}}$}\hspace{-1em}}\label{line:cache}

  \smallskip
  \lIf{$P'=\emptyset$}{\Return{$\SAT(F') \cdot 2^{\Card{P}}$}}\label{line:sat}
  \vspace{-.15em}
  $\mathcal{T}=(T,\chi) \leftarrow  \text{Decompose\_via\_Heuristics}(\nested{F'}{A})\qquad\tcc*[h]{Decompose}\hspace{-1em}$
  
  
  \If{$\width(\mathcal{T}) \geq\text{threshold}_{\text{hybrid}} \text{ or } \depth \geq \text{threshold}_{\text{depth}}$\,\tcc*[h]{\hspace{-.25em}Standard Solver\hspace{-1.5em}}}{\label{line:depth}
  	\lIf{$\var(F')=P'$}{$\cache \leftarrow \cache \cup \{ (F', \cSAT(F'))\}\hspace{-50em}$}\label{line:csat}
  	\lElse{\qquad\qquad\qquad\,\,\,$\cache \leftarrow \cache \cup \{ (F', \cESAT(F',P'))\}\hspace{-50em}$}\label{line:cesat}\vspace{-.3em}
  	\Return{$\cache(F')\cdot 2^{\Card{P\setminus P'}}$}\label{line:usecache}\vspace{-.2em}
  }

  \smallskip\smallskip
  
  \If{$\width(\mathcal{T}) \geq \text{threshold}_{\text{abstr}}\, \tcc*[h]{\hspace{-.25em}Abstract via Heuristics \& Decompose\hspace{-.5em}}\,$}{\label{line:nonesting}
  	$A\quad\,\,\,\,\,\leftarrow \text{Choose\_Subset\_via\_Heuristics}(A,F')$\label{line:abstraction}
  	
	\vspace{-0.2em}$\mathcal{T}=(T,\chi) \leftarrow \text{Decompose\_via\_Heuristics}(\nested{F'}{A})$\label{line:decomposenest}\vspace{-.2em}
  }
  
  \smallskip\smallskip

  $\ATab{\algo{N}} \leftarrow \adpa_{\algHES}(\depth, (F',P'), A, \mathcal{T})  \,\tcc*[h]{Nested Dynamic Programming}\hspace{-2em}$
\label{line:root}\label{line:nested}
  
%
%
%

  $\cache \leftarrow \cache \cup \{ (F', c) \mid \langle \emptyset, c\rangle \in\ATab{\algo{N}}[\rootOf(T)]\}\hspace{-1em}$\label{line:caching}
  
  \vspace{-.2em}
  \Return{$\cache(F')\cdot 2^{\Card{P\setminus P'}}$}\label{line:return}\vspace{-.25em}%
%
%
\caption{%
    Algorithm $\hdpa_{\algHES}(\depth, F,P)$ 
    for hybrid DP of \cESAT 
    based on nested DP. 
  %
}%
\label{fig:hdpontd}%
\end{algorithm}%

Listing~\ref{fig:hdpontd} depicts our algorithm~$\hdpa_{\algHES}$ for solving projeceted model counting, i.e., problem~$\cESAT$. 
This algorithm~$\hdpa_{\algHES}$ takes an instance~$(F,P)$ of \PMC consisting of Boolean formula~$F$ and projection variables~$P$. 
The algorithm maintains a global, but simple $\cache$ mapping a formula to an integer, 
and consists of the following four subsequent blocks of code, which are separated by empty lines: 
(1) Preprocessing \& Cache Consolidation, 
(2) Standard Solving, 
(3) Abstraction \& Decomposition, and (4) Nested Dynamic Programming,
which causes an indirect recursion through nested table algorithm~\algHES, as discussed later. 

Block (1) spans Lines~\ref{line:bcp}
-\ref{line:cache} and performs simple preprocessing techniques~\cite{LagniezMarquis14} like \emph{unit propagation}, thereby obtaining a simplified instance $(F',P')$, where simplified formula~$F'$ of~$F$ and projection variables~$P'\subseteq P$ are obtained. 
Any preprocessing simplifications are fine, as long as the solution of the resulting \PMC instance~$(F',P')$ is the same as solving \PMC on~$(F,P)$.
Then, in Line~\ref{line:updateD}, we set the set~$A$ of abstraction variables to $P'$, 
and consolidate $\cache$ with the updated formula~$F'$. 
Note that the operations in Line~\ref{line:bcp} are required to return a simplified instance 
that preserves satisfying assignments of the original formula when restricted to $P$.
If $F'$ is not cached, in Block~(2), we do standard solving if the width is out-of-reach for nested DP, which  spans over Lines~\ref{line:sat}-\ref{line:usecache}.
More precisely, if the updated formula $F'$ does not contain projection variables, 
in Line~\ref{line:sat} we 
employ a \SAT solver returning integer $1$ or $0$. 
If~$F'$ contains projection variables and either the width obtained by heuristically decomposing~$\primal{F'}$ is above~$\text{threshold}_{\text{hybrid}}$, or the nesting depth exceeds~$\text{threshold}_{\text{depth}}$, we use a standard~\cSAT or~\cESAT solver depending on~$P'$. 

Block (3) spans Lines~\ref{line:nonesting}-\ref{line:decomposenest} and is reached if no cache entry was found in Block (1) and standard solving was skipped in Block (2). If the width of the computed decomposition is above $\text{threshold}_{\text{abstr}}$, we need to use an abstraction in form of the nested primal graph. This is achieved by choosing suitable subsets~$E \subseteq A$ of abstraction variables and decomposing~${F}_t^{E}$ heuristically.
Finally, Block (4) concerns nested DP, cf.\ Lines~\ref{line:root}-\ref{line:return}. This block relies on nested table algorithm~$\algHES$, which is given in Listing~\ref{alg:primhyb2pmc} that is almost identical to nested table algorithm~$\algNES$ as already discussed above and given in Listing~\ref{alg:primhybpmc}.
The only difference of~$\algHES$ compared to~$\algNES$ is that in Line~\ref{line:intr:hybr2pmc}
the nested table algorithm~$\algHES$ uses the parameter $\depth$ and recursively executes
algorithm~$\hdpa_{\algHES}$ on the increased nesting depth of $\depth + 1$, 
and the same formula as the one used in the generic 
$\PMC$ oracle call in Line~\ref{line:intr:hybrpmc} of Listing~\ref{alg:primhybpmc}.

As a result, our approch 
deals with high treewidth by recursively finding and decomposing 
abstractions of the graph.
If the treewidth is too high for some parts, 
tree decompositions of abstractions are used to guide standard solvers.
Towards defining an actual implementation for practical solving, 
one still needs to find values for the threshold constants
$\text{threshold}_{\text{hybrid}}$, $\text{threshold}_{\text{depth}}$, and $\text{threshold}_{\text{abstr}}$.
The actual values of these constants will be made more precisely in the next section when discussing our implementation and experiments.

\begin{figure}[t]
\centering
\begin{tikzpicture}[node distance=3mm]
\tikzset{every path/.style=thick}

%
\node (r12) [stdnode, label={[tdlabel, xshift=0em,yshift=+0em]right:${t_1}$}]{$\{a,b\}$};
\node (r2) [stdnode, right=2.5em of r12, label={[tdlabel, xshift=0em,yshift=+0em]left:${t_{2}}$}]{$\{a\}$};
\node (j) [stdnode, ultra thick, above left=of r2, xshift=-.2em, yshift=-0.25em, label={[tdlabel, xshift=0em,yshift=+0.15em]right:${t_{3}}$}]{$\{a\}$};
\node (label) [font=\scriptsize,left=of j]{${\mathcal{T}}'$:};
\node (leaf0x) [stdnode, left=1em of r12, yshift=-1.35em, label={[tdlabel, xshift=0em,yshift=+.5em]below right:$\tab{1}$}]{%
	\begin{tabular}{l}%
		\multicolumn{1}{l}{$\langle \tuplecolor{\inputPredColor}{a, b}, \tuplecolor{\statePredColor}{\text{cnt}} \rangle$}\\
		\hline\hline
		$\langle \tuplecolor{\inputPredColor}{0, 0}, \tuplecolor{\statePredColor}{1}\rangle$\\\hline
		\rowcolor{yellow}$\langle\tuplecolor{\inputPredColor}{1, 0}, \tuplecolor{\statePredColor}{1}\rangle$\\\hline
		$\langle\tuplecolor{\inputPredColor}{0, 1}, \tuplecolor{\statePredColor}{1}\rangle$\\\hline
		\rowcolor{yellow}$\langle\tuplecolor{\inputPredColor}{1, 1}, \tuplecolor{\statePredColor}{1}\rangle$\\
	\end{tabular}%
};
\node (leaf0b) [stdnodenum,left=of leaf0x,xshift=1.2em,yshift=0pt]{%
	\begin{tabular}{c}%
		\multirow{1}{*}{$i$}\\ %
		\hline\hline
		$1$ \\\hline
		$2$ \\\hline
		$3$ \\\hline
		$4$ %
	\end{tabular}%
};
\node (leaf2b) [stdnodenum,below=2em of j,xshift=-2.25em,yshift=-.5em]  {%
	\begin{tabular}{c}%
		\multirow{1}{*}{$i$}\\ %
		\hline\hline
		$1$\\
	\end{tabular}%
};
\node (leaf2) [stdnode,right=-0.4em of leaf2b,yshift=0em, label={[tdlabel, xshift=0em,yshift=0.25em]right:$\tab{3}$}]  {%
	\begin{tabular}{l}%
		\multirow{1}{*}{$\langle \tuplecolor{\inputPredColor}{a}, \tuplecolor{\statePredColor}{\text{cnt}} \rangle$}\\ %
		\hline\hline
		\rowcolor{yellow}$\langle \tuplecolor{\inputPredColor}{1}, \tuplecolor{\statePredColor}{2}\rangle$\\
	\end{tabular}%
};
\node (leaf3b) [stdnodenum,below=1.1em of r2,xshift=2.25em,yshift=.5em]  {%
	\begin{tabular}{c}%
		\multirow{1}{*}{$i$}\\ %
		\hline\hline
		$1$\\
	\end{tabular}%
};
\node (leaf3) [stdnode,right=-0.4em of leaf3b,yshift=0em, label={[tdlabel, xshift=0em,yshift=0.25em]right:$\tab{2}$}]  {%
	\begin{tabular}{l}%
		\multirow{1}{*}{$\langle \tuplecolor{\inputPredColor}{a}, \tuplecolor{\statePredColor}{\text{cnt}} \rangle$}\\ %
		\hline\hline
		\rowcolor{yellow}$\langle \tuplecolor{\inputPredColor}{1}, \tuplecolor{\statePredColor}{1}\rangle$\\
	\end{tabular}%
};
%

\coordinate (top) at ($ (leaf2.north east)+(0.6em,-0.5em) $);
\coordinate (bot) at ($ (top)+(0,-12.9em) $);

\draw [->] (j) to ($ (r12.north)$);
\draw [->] (j) to ($ (r2.north)$);

\draw [dashed, bend right=0] (leaf2) to (j);
%
\draw [dashed, bend left=20] (leaf0x) to (r12);
\draw [dashed, bend right=0] (leaf3) to (r2);
\end{tikzpicture}\qquad%
\begin{tikzpicture}[node distance=3mm]
\tikzset{every path/.style=thick}

%
\node (i1) [stdnode, ultra thick, label={[tdlabel, xshift=0em,yshift=+0em]right:${t''}$}]{$\{a\}$};
\node (label) [font=\scriptsize,left=of i1]{${\mathcal{T}}''$:};
%
\node (join) [stdnode,below=0.75em of i1, yshift=-.75em, label={[tdlabel, xshift=.0em,yshift=+0.25em]right:$\tab{{t''}}$}] {%
	\begin{tabular}{l}%
		\multirow{1}{*}{$\langle \tuplecolor{\inputPredColor}{a}, \tuplecolor{\statePredColor}{\text{cnt}} \rangle$}\\
		\hline\hline
		\rowcolor{yellow}$\langle \tuplecolor{\inputPredColor}{1}, \tuplecolor{\statePredColor}{2} \rangle$\\
	\end{tabular}%
};
\node (joinb) [stdnodenum,left=-0.45em of join] {%
	\begin{tabular}{c}
		\multirow{1}{*}{$i$}\\
		\hline\hline
		$1$\\
	\end{tabular}%
};
\coordinate (top) at ($ (leaf2.north east)+(0.6em,-0.5em) $);
\coordinate (bot) at ($ (top)+(0,-12.9em) $);

\draw [dashed, bend right=25] (i1) to (join);
%
%
%
\end{tikzpicture}
\caption{Selected tables obtained by nested DP using~$\adpa_{\algHES}$ on  TD~${\mathcal{T}}'$ of~$\nested{F}{\{a,b\}}$ (left) and on TD $\TTT''$ of~$\nested{F}{\{a\}}$ (right) for instance~$(F,\{a,b\})$ of Example~\ref{ex:pmc}
.}
\label{fig:hybrid}
\end{figure}

\begin{algorithm}[t]
  \KwData{%
Nesting $\depth\geq 0$, 
node~$t$, 
bag $\chi_t$, bag formula~$F_t$, 
  nested bag formula~$F_t^A$, 
  projection variables~$P\subseteq\var(F_t^A)$, 
    and sequence $\Tab{}=\langle \tab{1},\ldots \tab{\ell}\rangle$ of child tables  of~$t$.
}
\KwResult{Table~$\tab{t}.$} \lIf(\hspace{-1em})
  {$\type(t) = \leaf$}{%
    $\tab{t} \leftarrow \{ \langle
    \tuplecolor{\specialPredColor}{\emptyset},
    \tuplecolor{\statePredColor}{1} \rangle \}$%
  }%
  \uElseIf{$\type(t) = \intr$ and
    $a\hspace{-0.1em}\in\hspace{-0.1em}\chi(t)$ is introduced}{ %
    \makebox[3.3cm][l]{$\hspace{-0em}\tab{t} \leftarrow \{ \langle
      \tuplecolor{\specialPredColor}{J},
      \tuplecolor{\statePredColor}{c'{\cdot}c} \rangle$
    }%
    \hspace{-1em}$|\;
    \langle \tuplecolor{\specialPredColor}{I},
    \tuplecolor{\statePredColor}{c} \rangle \in \tab{1}, 
    %
    J= I\cup\{a\mapsto v\}, v\in\{0,1\}, J \models F_t, c'>0, c'{=}\hspace{-4em}$ \newline 
	\makebox[19.5em]{}$\hdpa_{\algHES}(\depth{+}1,F_t^A[J], P) \} \hspace{-3em}$
    %
    %
     \label{line:intr:hybr2pmc}%
  }\vspace{-0.05em}%
  \uElseIf{$\type(t) = \rem$ and $a \not\in \chi(t)$ is removed}{%
    %
    %
    $\text{C}(I) \leftarrow \{ \langle J, c\rangle \SM \langle {J}, {c}\rangle \in \tab{1}, J\setminus\{a \mapsto 0, a \mapsto 1\} = I\setminus\{a \mapsto 0, a \mapsto 1\} \SE$\\
\makebox[5cm][l]{$\tab{t} \leftarrow \{ \langle
      \tuplecolor{\specialPredColor}{I\setminus\{a \mapsto 0, a \mapsto 1\}}$,
      %
    %
      \tuplecolor{\statePredColor}{%
        $\sum_{\langle J, c\rangle \in C(I)} c $}$\rangle$
    }
 \hspace{3.5em}$\mid \langle \tuplecolor{\specialPredColor}{I},
    \tuplecolor{\statePredColor}{\cdot} \rangle \in \tab{1}    \}\hspace{-15em}$\;
  } %
  \uElseIf{$\type(t) = \join$}{%
    \makebox[3.3cm][l]{$\tab{t} \leftarrow \{ \langle
      \tuplecolor{\specialPredColor}{I},
      \tuplecolor{\statePredColor}{c_1 \cdot c_2}
      \rangle$}\hspace{9em}$|\;\langle \tuplecolor{\specialPredColor}{I},
    \tuplecolor{\statePredColor}{c_1} \rangle \in \tab{1}, \langle
    \tuplecolor{\specialPredColor}{I},
    \tuplecolor{\statePredColor}{c_2} \rangle \in \tab{2}
    \}\hspace{-5em}$
    \vspace{-0.25em}
  } %
  \Return $\tab{t}$ \vspace{-0.25em}
  \caption{Nested table algorithm~$\algHES(\depth, t,\chi_t, F_t, F_t^A, P, \Tab{})$ for solving \PMC.}
  \label{alg:primhyb2pmc}
\end{algorithm}%

\begin{example}\label{ex:hybrid}
Recall instance~$(F,A)$ of Example~\ref{ex:pmc}, and set~$A$ of abstraction variables as well as TD~$\TTT'$ of nested primal graph~$\nested{F}{A}$ as given in Example~\ref{ex:abstracttd}.
Further, recall that restricted to projection set~$A$, formula $F$ has two satisfying assignments. 
%
Figure~\ref{fig:hybrid} (left) shows TD~$\TTT'$ of~$\nested{F}{A}$ and 
tables obtained by~$\adpa_{\algHES}(0,(F,A),A,\TTT')$ for solving projected model counting on~$(F,A)$. 
%

Note that nested table algorithm~$\algHES$ of Listing~\ref{alg:primhyb2pmc} works similar to the nested table algorithm~\algNES of Listing~\ref{alg:primhybpmc}, 
but it calls $\hdpa_{\algHES}$ recursively. 
We briefly discuss executing~$\algHES$ in the context of Line~\ref{line:nested} of algorithm~$\hdpa_{\algHES}$ on node $t_1$, resulting in table~\tab{1} as shown in Figure~\ref{fig:hybrid} (left).
Recall that 
$F_{t_1}^A=\{\{\neg a, b, p_1\},\{a,\neg b, \neg p_1\}\}$.
Then, in Line~\ref{line:intr:hybr2pmc} of algorithm~$\algHES$, 
for each assignment~$J$ to~$\{a,b\}$ of each row of~\tab{1}, 
we compute $\hdpa_{\algHES}(\depth+1, F_{t_1}^A[J], \emptyset)$. 
Each of these recursive calls, however, is already solved by unit propagation (preprocessing), e.g., $F_{t_1}^A[\{a\mapsto 1, b \mapsto 0\}]$ of Row 2 simplifies to $\{\{p_1\}\}$. 

Figure~\ref{fig:hybrid} (right) shows TD~$\TTT''$ of~$\nested{F}{E}$ with~$E\eqdef\{a\}$, and tables obtained by algorithm $\adpa_{\algHES}(0,(F,A),E,\TTT'')$.
Still, $F_{t''}^E[J]$ for a given assignment~$J$ to~$\{a\}$ of any row~$r\in\tab{t''}$ can be simplified. 
Concretely, $F_{t''}^E[\{a\mapsto 0\}]$ evaluates to $\emptyset$ and $F_{t''}^E[\{a\mapsto 1\}]$ evaluates to clause~$\{b, c\}$. 
Thus, restricted to $\{b\}=A\setminus\{a\}$, there are 2 satisfying assignments~$\{b\mapsto 0\}$, $\{b\mapsto 1\}$ of~$F_{t''}^E[\{a\mapsto 1\}]$.
\end{example}

\section{Hybrid Dynamic Programming in Practice}\label{sec:implreal}

Below, in Section~\ref{sec:hdp} we present an implementation of hybrid dynamic programming in order to solve the problems \cSAT as well as \PMC. This is then followed by an experimental evaluation and discussion of the results in Section~\ref{sec:experiments}\FIX{, where we also briefly elaborate on existing techniques of state-of-the-art solvers.}

\subsection{Implementing Hybrid Dynamic Programming}\label{sec:hdp}

We implemented a solver \solver{}\footnote{\solver{} is open-source and available at~\href{https://github.com/hmarkus/dp\_on\_dbs/tree/nesthdb}{github.com/hmarkus/dp\_on\_dbs/tree/nesthdb}. Instances and detailed results are available online at: \href{https://tinyurl.com/nesthdb}{tinyurl.com/nesthdb}.\label{foot:source}} based on hybrid dynamic programming
%
in Python3 and using table manipulation techniques by means of \emph{structured query language (SQL)} and the \emph{database management system (DBMS)} PostgreSQL. 
\FIX{Our solver builds upon the recently published prototype \dpdb~\cite{FichteEtAl20}, which applied a DBMS for the efficient implementation of plain dynamic programming algorithms. This \dpdb prototype provides a basic framework for implementing plain dynamic programming algorithms, which can be specified in the form of a plain table algorithm, e.g., the one of Listing~\ref{fig:prim}.
%
However, this system does not have support for neither hybrid nor nested dynamic programming.
In order to compare plain \dpdb and our solver \solver{} in a fair way, for both systems we used
the most-recent 
version 12 of PostgreSQL and we let it operate on a \emph{tmpfs-ramdisk} instead of disk space (HDD/SDD), i.e., within the main memory (RAM) of a machine.
In both \dpdb as well as our solver \solver{}, the DBMS serves the purpose of extremely efficient in-memory table manipulations and query optimization required by nested DP, and therefore~\solver benefits from database technology.
Those benefits are already available in the form of different and efficient join manipulations that are selected based
on several heuristics that are invoked during SQL query optimizing.
Note that especially efficient join operations have been already designed, implemented, combined, and tuned for decades~\cite{Ullman89,Garcia-MolinaUllmanWidom09,ElmasriNavathe16}.
Therefore it seems more than natural to rely on this technological advancement that database theory readily provides.%
}
%
%
We are certain that one can easily replace PostgreSQL by any other state-of-the-art relational database that uses standard SQL in order to express queries.
In the following, we briefly discuss implementation specifics that are crucial for a performant system that is competitive with state-of-the-art solvers. 
%

\medskip
\noindent\textbf{Nested DP \& Choice of Standard Solvers.}
We implemented dedicated nested DP algorithms for solving~\cSAT and \cESAT, where we do (nested) DP up to $\text{threshold}_{\text{depth}}=2$.
\FIX{Note that incrementing nesting depth results in getting again exponentially many (in the largest bag size) 
rows for each row of tables of the previous depth, i.e., a low nesting limit is highly expected. Currently,
we do not see a way to efficiently solve instances of higher nesting depth, which might change in case of
further advances allowing to decrease table sizes obtained during dynamic programming.
} 
Further, we set $\text{threshold}_{\text{hybrid}}=1000$ and therefore we do not ``fall back'' to standard solvers based on the width (cf., Line~\ref{line:depth} of Listing~\ref{fig:hdpontd}), but based on the nesting depth.

Also, the evaluation of the nested bag formula is ``shifted'' to the database if it uses \emph{at most}~$40$ abstraction variables, 
since PostgreSQL efficiently handles these small-sized Boolean formulas.
Thereby, further nesting is saved by executing optimized SQL statements within the TD nodes.
A value of 40 seems to be a nice balance between the overhead caused by standard solvers for small formulas and 
exponential growth counteracting the advantages of the DBMS.
%
%
For hybrid solving,
we use \cSAT solver \sharpsat~\cite{Thurley06a} 
and for~\cESAT we employ the recently published~\cESAT 
solver \projmc~\cite{LagniezMarquis19}, solver \sharpsat and \SAT solver \picosat~\cite{Biere08}.
Observe that our solver immediately benefits from better standard solvers
and further improvements of the solvers above.

\medskip
\noindent\textbf{Choosing Non-Nesting Variables \& Compatible Nodes.}
TDs are computed by means of heuristics via decomposition library 
\htd~\cite{AbseherMusliuWoltran17a}.
\FIX{%
  We implement a heuristic for finding practically sufficient
  abstractions,~i.e., abstraction variables for the nested primal
  graph, in reasonable using an external solver.
Therefore, we encode our heuristic into two logic programs (ASP) for
solver \clingo~\cite{GebserEtAl19}, which includes techniques for fast
solving reachability via nesting paths.
The encodings, which in total comprise 11 lines, are publicly
available in the online repository of \solver.
}
\FIX{
  Technically, our focus is on avoiding extremely large abstractions
  at the cost of larger nested bag formulas. Still, nesting allows to
  obtain refined abstractions again at higher depths. Thereby, we
  achieve a good trade off between runtime and quality.
}

\FIX{By the first encoding (``guess\_min\_degree.lp''), we compute a \emph{reasonably-sized subset of
    vertices} of smallest degree, more precisely, such that the number of neighboring vertices not in the set 
is minimized.
We take a subset of size at most~$95$, which turned out to be practically useful.
We run the ASP solver \clingo for \emph{up to~$10$ seconds}. The
solver might not return an optimum within $10$ seconds, but always
returns a subset of vertices that can be used subsequently.
}
%

\FIX{By the second encoding (``guess\_increase.lp''), we guess
among the thereby obtained subset of vertices of preferably smallest degree, 
a preferably maximal set~$A$ of at most~$64$ abstraction variables such 
 the resulting graph~$\nested{F}{A}$ is reasonably sparse, 
which is achieved by minimizing the number of edges of~$\nested{F}{A}$.
To this end, we also use built-in (cost) optimization, where we take the best results obtained by \clingo after running \emph{at most 35 seconds}.
For more details on ASP, we refer to introductory texts~\cite{GebserKaufmannSchaub09a,GebserEtAl19}.
}%

\FIX{%
We expect that this approach, which driven mostly by practical considerations, can be improved. Furthermore, it can also be extending by problem-specific as well as domain-specific information, which might help in choosing promising abstraction variables~$A$.
%
}

As rows of tables during (nested) DP can be independently computed
and parallelized~\cite{FichteHecherZisser19},
hybrid solver \solver potentially calls standard solvers for solving subproblems in parallel using a thread pool.
Thereby, the uniquely compatible node for relevant compatible sets~$U$, as denoted in this paper by means of~$\compat(\cdot)$,
 is decided during runtime among compatible nodes on a first-come-first-serve basis.

\subsection{Experimental Evaluation}\label{sec:experiments}

In order to evaluate the concept of hybrid dynamic programming,
we conducted a series of experiments considering a variety of solvers and benchmarks,
both for model counting (\sharpSAT) as well as projected model counting (\PMC).
During the evaluation we thereby compared the performance of 
algorithm $\hdpa_{\algHES}$ of Listing~\ref{fig:hdpontd}.
We benchmarked this algorithm both for the projected model counting problem, but also for the special case of model counting, where all variables are projection variables.


\begin{figure}[t]\centering
\includegraphics[scale=.56]{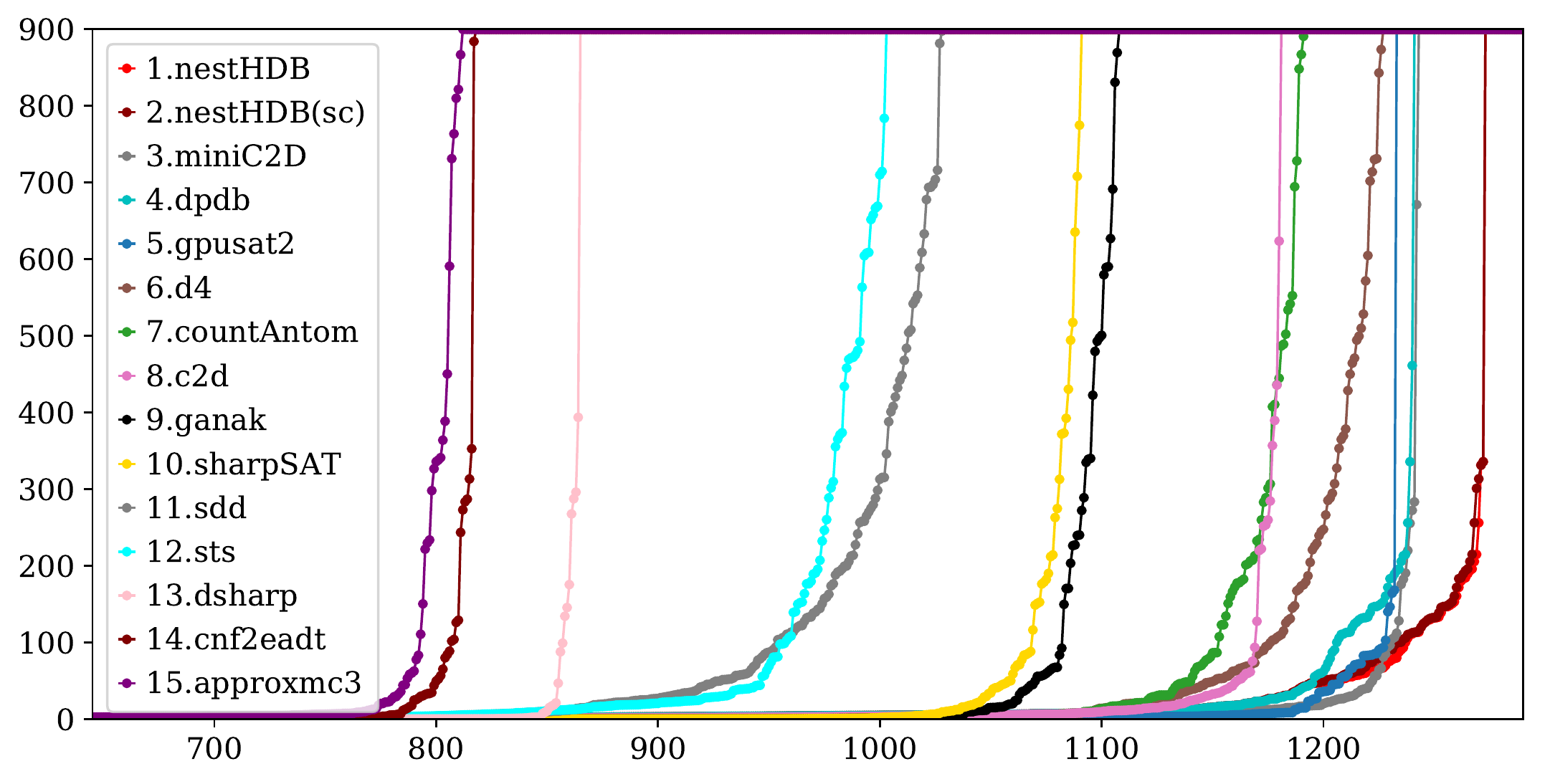}
\caption{Cactus plot of instances for~\cSAT, where instances (x-axis) are ordered for each solver individually by runtime[seconds] (y-axis).
$\text{threshold}_{\text{abstr}}=38$. 
}
\label{fig:sat}
\end{figure}

\paragraph{Benchmarked Solvers \& Instances} 
We benchmarked \nesthdb and 16 other publicly 
available~\cSAT solvers on 1,494 instances recently considered~\cite{FichteEtAl20}.
\FIX{Most of the existing solvers
of other approaches are single-core solvers, which adhere to
three different techniques, namely \emph{knowledge-compilation based},
\emph{caching based}, and \emph{approximate}.
Among the knowledge-compilation based solvers, which aim
to obtain compact representations of the formulas that are concise and
easier to solve,
are 
{\textsf{miniC2D}}~\cite{OztokDarwiche15a}, 
{\textsf{d4}}~\cite{LagniezMarquis17a},
{\textsf{c2d}}~\cite{Darwiche04a},
{\textsf{sdd}}~\cite{Darwiche11a},
{\textsf{dsharp}}~\cite{MuiseEtAl12a},
{\textsf{cnf2eadt}}~\cite{KoricheLagniezMarquisThomas13a}, and
{\textsf{bdd\_{}minisat}}~\cite{TodaSoh15a}.
These solvers use different variants and flavors of knowledge-compilation,
thereby finding decent trade-offs between the time needed to obtain those representations and their succinctness.
We also considered the caching-based solvers
{\textsf{cachet}}~\cite{SangEtAl04},
{\textsf{sharpSAT}}~\cite{Thurley06a}, and
{\ganak}~\cite{SharmaEtAl19}, which employ existing \SAT-based 
solvers by sophisticated caching techniques.
%
Finally, among the approximate counters, we focused on
{\textsf{sts}}~\cite{ErmonGomesSelman12a},
{\textsf{sharpCDCL}}~\cite{KlebanovEtAl13}, and
{\textsf{approxmc3}}~\cite{ChakrabortyEtAl14a},
which employ sampling-based techniques to approximately obtain the model counts.
Our comparison also included the multi-core solvers
\dpdb~\cite{FichteEtAl20},
{\textsf{gpusat2}}~\cite{FichteHecherZisser19}, which is also based on dynamic programming and 
uses massively parallel graphics processing units (GPUs), as well as
\textsf{countAntom}~\cite{BurchardSchubertBecker15a}, which relies on sophisticated techniques
for work-balancing.
For a more ample description of the used techniques, we refer to the 
model counting competition report~\cite{FichteHecherHamiti20}.
}
%
  Note that we excluded distributed solvers such as
  dCountAntom~\cite{BurchardSchubertBecker16a} and
  DMC~\cite{LagniezMarquisSzczepanski18a} from our experimental
  setup. Both solvers require a cluster with access to the OpenMPI
  framework~\cite{GabrielFaggBosilca04} and fast physical
  interconnections. Unfortunately, we do not have access to OpenMPI on
  our cluster. Nonetheless, our focus are shared-memory systems and
  since \dpdb might well be used in a distributed setting, it leaves an
  experimental comparison between distributed solvers that also
  include \dpdb as subsolver to future work.
%
%
%
While \nesthdb itself is a multi-core solver,
we additionally included in our comparison \solversc, which is \nesthdb,
but restricted to a single core only.
The instances~\cite{FichteEtAl20} we took are \emph{already preprocessed} by \pmcp~\cite{LagniezMarquis14} using recommended options \texttt{-vivification -eliminateLit -litImplied -iterate=10 -equiv -orGate -affine}, which guarantee that the model counts are preserved.
However, \nesthdb still uses \pmcp with these options in Line~\ref{line:bcp} of Listing~\ref{fig:hdpontd}, \FIX{which is used 
in the light of nested bag formulas that appear due to nesting}.

Further, we considered the problem~\cESAT, where we compare solvers \projmc~\cite{LagniezMarquis19}, \clingo~\cite{GebserEtAl19}, \ganak~\cite{SharmaEtAl19}, \nesthdb$^{\hspace{-.25em}\ref{foot:source}}$, and \solversc on 610 publicly available instances\footnote{Sources: \hspace{-0.2em}\href{https://tinyurl.com/projmc}{tinyurl.com/projmc};\hspace{0.2em}\href{https://tinyurl.com/pmc-fremont-01-2020}{tinyurl.com/pmc-fremont-01-2020}.} from~\projmc (consisting of 15 \emph{planning}, 60 \emph{circuit}, and 100 \emph{random} instances) and Fremont, with 170 \emph{symbolic-markov} applications, and 265 \emph{misc} instances.
\FIX{For simplifying nested bag formulas under assignments encountered due to nesting in Line~\ref{line:bcp} of Listing~\ref{fig:hdpontd}}, \nesthdb
uses \pmcp as before, but \emph{without options} \texttt{-equiv -orGate -affine} to ensure preservation of models (equivalence).

\begin{figure}[tb]
  \centering
    \smaller
    \begin{tabular}{{r|l|rrrrHH|r|rHr}}
      \toprule
      {bench-} & \multirow{2}{*}{solver} & \multicolumn{4}{c}{tw upper bound} & best & unique & \multirow{2}{*}{$\sum$} & time & rank \\
      {mark set} & & $\text{max}$ & 0-30 & 31-50 & $>$50 & best & unique &  & [h] & rank \\
\midrule
%
%
%
planning &\nesthdb  & \textbf{30} & \textbf{7} & {0} & 0 & 1 & 0 & \textbf{7} & \textbf{2.88} \\
&{\solversc}  & \textbf{30} & \textbf{7} & {0} & 0 & 1 & 0 & \textbf{7} & {3.31} \\
&\projmc & {26} & {6} & 0 & {0} & 5 & 0 & 6 & {3.01} \\
& \ganak & 19 & 5 & 0 & 0 & 1 & 0 & 5 & 3.36\\ 
& \clingo & 4 & 1 & 0 & 0 & 0 & 0 & 1 & 4.00 \\
%
\midrule
circ &\nesthdb & \textbf{99} & \textbf{34} & \textbf{10} & \textbf{16} & \textbf{20} & \textbf{0} & \textbf{60} & {2.10} \\
&{\solversc} & \textbf{99} & \textbf{34} & {4} & {14} & {20} & \textbf{0} & {52} & {4.60} \\
&\projmc & {91} & {28} & \textbf{10} & {11} & \textbf{43} & {0} & {49} & 6.23 \\
&\ganak & \textbf{99} & \textbf{34} & \textbf{10} & \textbf{16} & 58 & 0 & \textbf{60} &  \textbf{1.21} \\
&\clingo & \textbf{99} & 31 & \textbf{10} & \textbf{16} & 2 & 0 & 57 &  4.44 \\
%
\midrule
random & \nesthdb & {79} & \textbf{30} & \textbf{20} & \textbf{17} & \textbf{11} & \textbf{0} & \textbf{67} & \textbf{10.91} \\
& \solversc & {79} & \textbf{30} & \textbf{20} & {15} & \textbf{11} & \textbf{0} & {65} & {11.29} \\
& \projmc & \textbf {84} & \textbf{30} & \textbf{20} & {15} & \textbf{62} & 0 & {65} & 11.09 \\
& \ganak & 19 & 19 & 0 & 0 & 4 & 0 & 19 & 23.18 \\
& \clingo & 24 & 25 & 0 & 0 & 0 & 0 & 25 & 21.38 \\
\midrule
markov & \nesthdb & {23} & {62} & {0} & {0} & {30} & {0} & {62} & {31.98} \\
& \solversc & {23} & {61} & {0} & {0} & {30} & {0} & {61} & {32.54} \\
& \projmc& {8} & 54 & 0 & {0} & 54 & 14 & {54} & 33.65 \\
& \ganak & \textbf{59} & \textbf{64} & 0 & \textbf{4} & 31 & 6 & \textbf{68} & \textbf{30.32} \\
& \clingo & 3 & 38 & 0 & 0 & 0 & 0 & 38 &  37.54 \\
\midrule
misc & \nesthdb & {47} & \textbf{38} & \textbf{17} & {0} & \textbf{7} & 0 & \textbf{55} & 46.12 \\
& \solversc & {47} & \textbf{38} & {13} & {0} & \textbf{7} & 0 & {51} & 48.20 \\
& \projmc & {47} & \textbf{38} & {13} & {0} & 50 & 0 & {51} & {45.90} \\
& \ganak & 44 & \textbf{38} & 15 & 0 & 6 & 3 & 53 &  \textbf{45.72} \\
& \clingo & \textbf{63} & \textbf{38} & 15 & \textbf{1} & 27 & \textbf{1} & 54 &  44.79 \\
\midrule
\midrule
$\Sigma$ & \nesthdb & \textbf{99} & \textbf{171} & \textbf{47} & \textbf{33} & \textbf{69/74} & \textbf{0/256} & \textbf{251} & \textbf{93.99} \\
& \solversc & \textbf{99} & {170} & {37} & {29} & \textbf{69/74} & \textbf{0/256} & {236} & {99.95} \\
& \projmc & {91} & 156 & 43 & 26 & 215 & 15 & 225 & 99.88 \\
& \ganak & \textbf{99} & 160 & 25 & 20 & 100 & 9 & 205 & 103.78 \\
& \clingo & \textbf{99} & 133 & 25 & 17 & 215 & 15 & 175 & 112.15 \\
    \bottomrule
  \end{tabular}%
  \caption{%
    Number of solved~\cESAT instances, grouped by upper bound intervals of treewidth. time[h] is cumulated wall clock time, timeouts count as 900s. $\text{threshold}_{\text{abstr}}{=}8$.
  }%
\label{table:pmc}
\end{figure}

\begin{figure}\centering
\includegraphics[scale=.6]{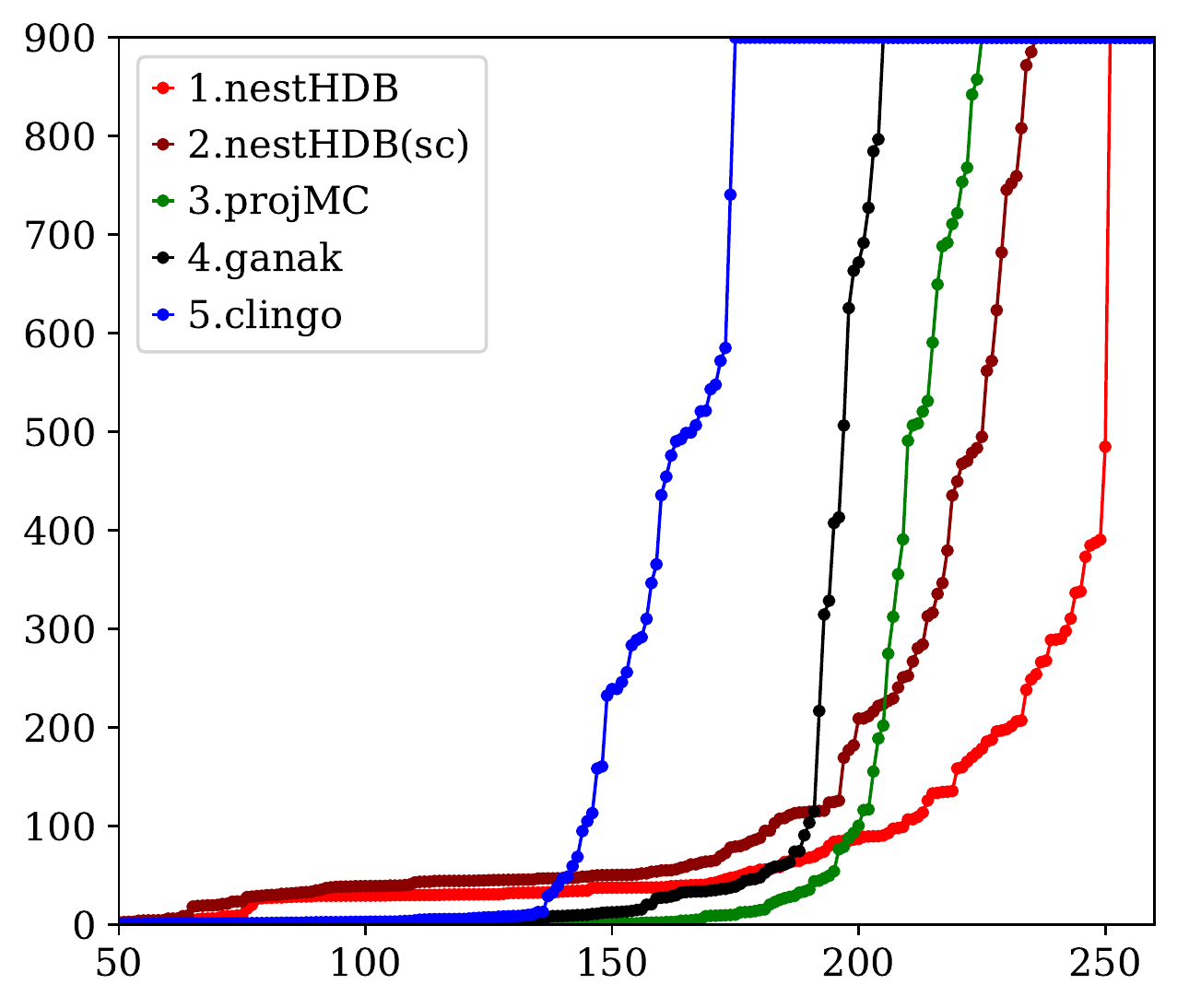}
\caption{%
    Cactus plot showing the number of solved~\cESAT instances, where the x-axis shows for each solver (configuration) individually, the number of instances ordered by increasing runtime. time[h] is cumulated wall clock time, timeouts count as 900s. $\text{threshold}_{\text{abstr}}{=}8$.
  }%
\label{table:pmc2}
\end{figure}

\paragraph{Benchmark Setup}
Solvers ran on a cluster of 12 nodes. 
Each node of the cluster is equipped
with two Intel Xeon E5-2650 CPUs consisting of 12 physical cores each
at 2.2 GHz clock speed, 256 GB RAM. %
For \dpdb and \nesthdb, 
we used PostgreSQL 12 on a tmpfs-ramdisk (/tmp) that could grow up to 
at most 1 GB per run.
Results were gathered on Ubuntu~16.04.1 LTS machines with disabled hyperthreading on kernel~4.4.0-139.
We mainly compare total wall clock time and number of timeouts. 
%
%
For parallel solvers~(\dpdb, \countAntom, \nesthdb) we 
allow 12 physical cores. 
%
%
%
%
%
Timeout is 900 seconds and RAM is limited to~16 GB per instance and solver.
Results for \gpusat are taken from~\cite{FichteEtAl20}, where a machine equipped with a consumer GPU is used:
Intel Core i3-3245 CPU operating at 3.4 GHz, 16 GB RAM, and one
Sapphire Pulse ITX Radeon RX 570 GPU running at 1.24 GHz with 32
compute units, 2048 shader units, and 4GB VRAM using driver
amdgpu-pro-18.30-641594 and OpenCL~1.2.
The system operated on Ubuntu~18.04.1 LTS with kernel 4.15.0-34.
%

\paragraph{Benchmark Results} 
%
%
The results for~\cSAT showing the best 14 solvers are summarized 
in the cactus plot of Figure~\ref{fig:sat}.
Overall it shows~\nesthdb among the best solvers, solving 1,273 instances.
The reason for this is, compared to~\dpdb, that \nesthdb
can solve instances using TDs of primal graphs
of widths larger than 44, up to width 266.
This limit is even slightly larger than the width of 264 that \sharpsat
on its own can handle.
We also tried using~\minic instead of \sharpsat as standard solver
for solvers \nesthdb and \solversc, 
but we could only solve one instance more.
Notably, \solversc has about the same performance as \nesthdb,
indicating that parallelism does not help much on the instances.
Further, we observed that the employed simple cache as used in Listing~\ref{fig:hdpontd}, 
provides only a marginal improvement.
%
%
%

Figure~\ref{table:pmc} depicts a table of results on \cESAT, where 
we observe that \nesthdb does a good job on instances with low
widths below $\text{threshold}_{\text{abstr}}=8$ (containing ideas of~\dpdb), 
but also on widths well above~$8$  (using nested DP).
Notably, \nesthdb is also competitive on widths well above~$50$.
Indeed, \nesthdb and \solversc perform well
on all benchmark sets, whereas on some sets the solvers \projmc, \clingo and \ganak
are faster.
Overall, parallelism provides a significant improvement here,
but still \solversc shows competitive performance,
which is also visualized in the cactus plot of Figure~\ref{table:pmc2}.
%
Figure~\ref{fig:scatter} shows scatter plots comparing \nesthdb to \projmc (left)
and to \ganak (right).
Overall, both plots show that \nesthdb solves more instances,
since in both cases the y-axis shows more black dots at 900 seconds than the x-axis.
Further, the bottom left of both plots shows that
there are plenty easy instances that can be solved by \projmc and \ganak
in well below 50 seconds, where \nesthdb needs up to 200 seconds.
Similarly, the cactus plot given in Figure~\ref{table:pmc2} shows
that \nesthdb can have some overhead compared to the three standard solvers, which is not surprising.
This indicates that there is still room for improvement
if, e.g., easy instances are easily detected, and if standard solvers
are used for those instances.
Alternatively, one could also just run a standard solver for at most 50 seconds
and if not solved within 50 seconds, the heavier machinery of nested dynamic programming
is invoked.
Apart from these instances, Figure~\ref{fig:scatter} shows that \nesthdb 
solves harder instances faster, where standard solvers struggle.

\begin{figure}[t]\centering
\includegraphics[scale=.34]{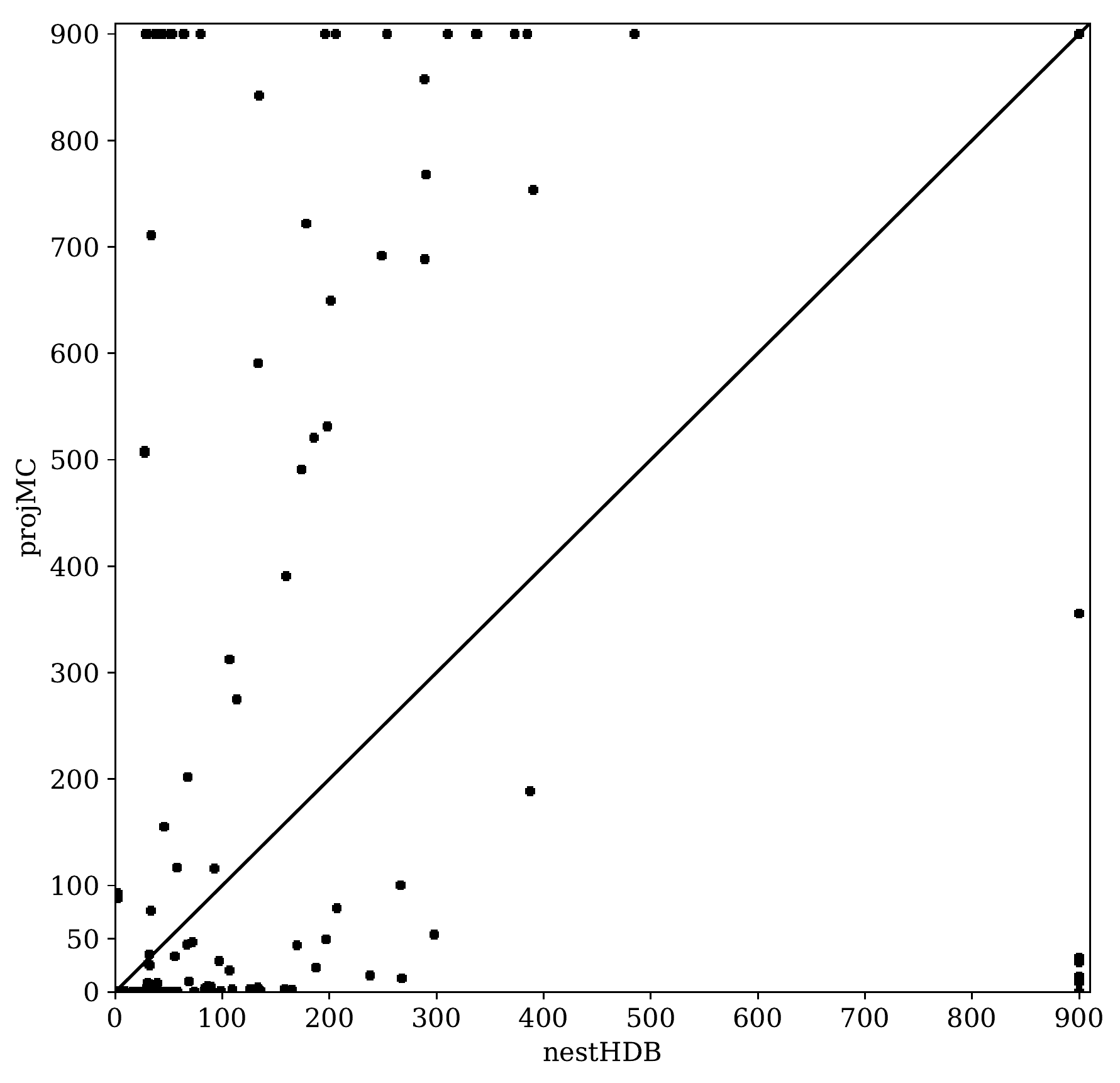}
\includegraphics[scale=.34]{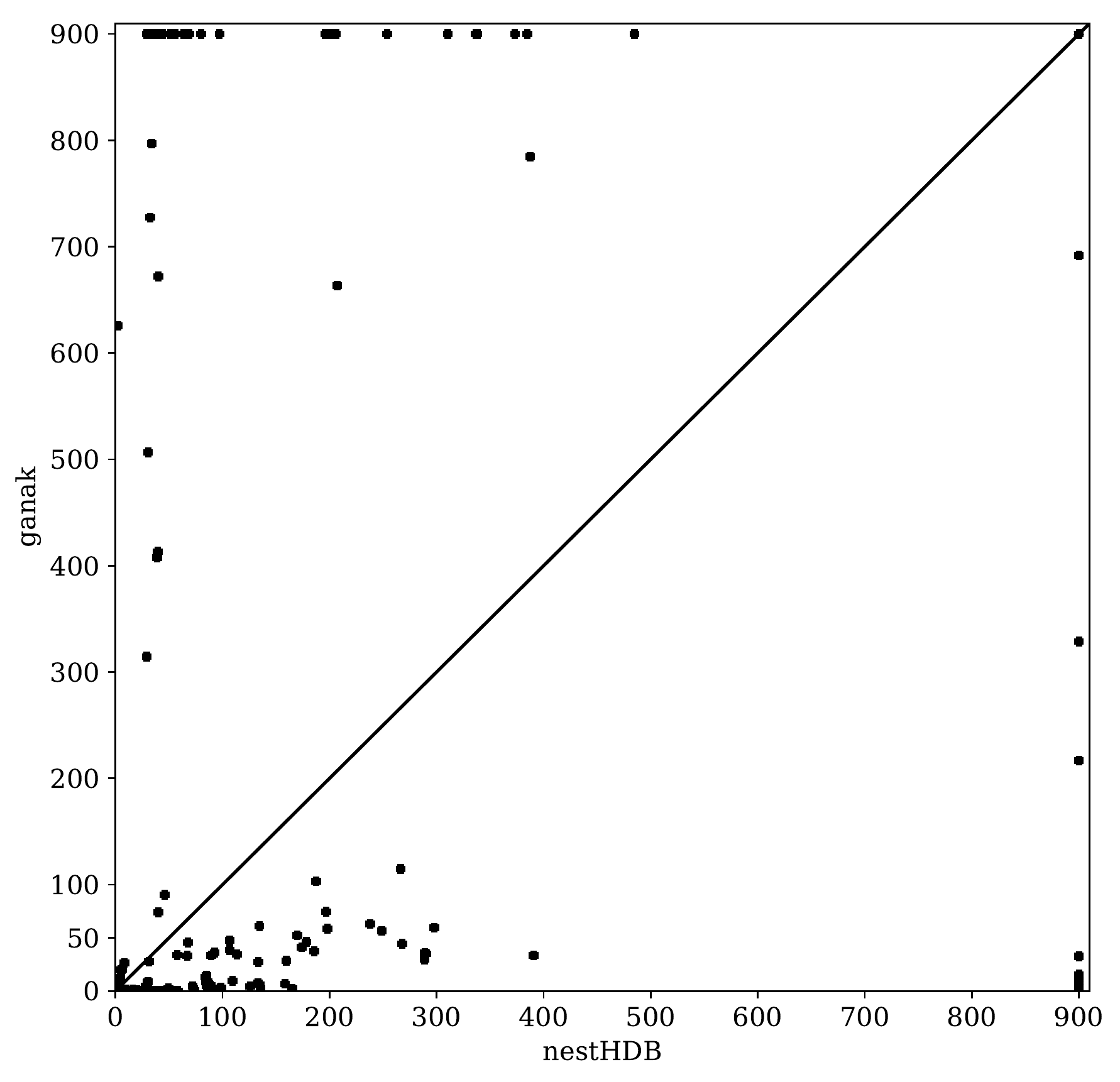}
\caption{Scatter plot of instances for~\cESAT, where the x-axis shows runtime in seconds of \nesthdb compared to the y-axis showing runtime of \projmc (left) and of \ganak (right).
$\text{threshold}_{\text{abstr}}=8$.
}
\label{fig:scatter}
\end{figure}

\section{Discussion and Conclusion}\label{sec:conclusions}
We introduced a dynamic programming algorithm to solve projected model
counting (\PMC) by exploiting the structural parameter treewidth. Our
algorithm is asymptotically optimal under the exponential time
hypothesis (ETH). Its runtime is double exponential in the treewidth
of the primal 
graph of the instance and polynomial in the size of the input
instance. We believe that our results can also be extended to another
graph representation, namely the incidence graph.
Our approach 
is very general and might be applicable to a wide range of other hard
combinatorial problems, such as projection for
ASP~\cite{FichteEtAl17a} and QBF~\cite{CharwatWoltran16a}.

Then, in order to still efficiently deal with projected model counting in practice, we presented 
nested dynamic programming (nested DP) using
%
different levels
of abstractions, 
which are subsequently refined and
solved recursively.
This approach is complemented with hybrid solving,
where 
(search-intense) subproblems 
are solved by standard solvers.
%
%
We provided nested DP algorithms for problems related to Boolean satisfiability, but the idea can be easily applied for other formalisms.
We implemented some of these algorithms and our benchmark results are promising.

\FIX{In the light of related works on properties for efficient counting algorithms, e.g.,~\cite{DurandMengel13,ChenMengel17,GrecoScarcello17}, we are curious to revisit some of those and potentially study precise runtime dependencies. 
We expect interesting insights when focusing on the search for properties of local instance parts that in combination with treewidth allow algorithms that are significantly better than double-exponential in the treewidth. 
As we demonstrated in the experimental results, we can solve the problem \PMC that theoretically requires double-exponential worst-case effort in the treewidth, on instances of decent treewidth upper bounds (up to 99). On plain model counting (\sharpSAT), which is only single-exponential in the treewidth, our solver even deals with instances of larger treewidth upper bounds (up to 260).
This also opens the question of whether similar empirical observations can be drawn in other areas and formalisms like constraint solving or database query languages.
}
Further, we plan on deeper studies of problem-specific abstractions, in particular for quantified Boolean formulas.
%
%
We want to further tune our solver parameters (e.g., thresholds, timeouts, sizes), deepen interleaving with \PMC solvers like \projmc,
and to use incremental solving for obtaining abstractions and evaluating nested bag formulas, 
where intermediate solver references are kept during dynamic programming and formulas are iteratively added and (re-)solved. 
%

\appendix
\newpage

\section*{Acknowledgements}\noindent
    This work has been supported by the Austrian Science Fund (FWF),
    Grants J4656, P32830 and Y698, 
as well as the Vienna Science and Technology Fund, Grant WWTF ICT19-065. Hecher is also affiliated with
  the 
  University of Potsdam, Germany.
  The main research was conducted while Fichte and Morak were affiliated with TU Vienna, Austria.
  Part of the research was carried out while Hecher and Fichte were visiting the Simons Institute for the Theory of Computing at UC Berkeley.
  We thank the anonymous reviewers for their detailed comments. 
\bibliography{references}

\begin{thebibliography}{10}
\expandafter\ifx\csname url\endcsname\relax
  \def\url#1{\texttt{#1}}\fi
\expandafter\ifx\csname urlprefix\endcsname\relax\def\urlprefix{URL }\fi
\expandafter\ifx\csname href\endcsname\relax
  \def\href#1#2{#2} \def\path#1{#1}\fi

\bibitem{AbramsonBrownEdwards96a}
B.~Abramson, J.~Brown, W.~Edwards, A.~Murphy, R.~L. Winkler, Hailfinder: A
  {B}ayesian system for forecasting severe weather, International Journal of
  Forecasting 12~(1) (1996) 57--71.
\newblock \href {https://doi.org/10.1016/0169-2070(95)00664-8}
  {\path{doi:10.1016/0169-2070(95)00664-8}}.

\bibitem{ChoiBroeckDarwiche15a}
A.~Choi, G.~Van~den Broeck, A.~Darwiche, Tractable learning for structured
  probability spaces: A case study in learning preference distributions, in:
  Q.~Yang (Ed.), Proceedings of 24th International Joint Conference on
  Artificial Intelligence (IJCAI'15), The AAAI Press, 2015.

\bibitem{DomshlakHoffmann07a}
C.~Domshlak, J.~Hoffmann, Probabilistic planning via heuristic forward search
  and weighted model counting, J. Artif. Intell. Res. 30 (2007) 565--620.
\newblock \href {https://doi.org/10.1613/jair.2289}
  {\path{doi:10.1613/jair.2289}}.

\bibitem{MeelEtAl17a}
L.~Due{\~{n}}as{-}Osorio, K.~S. Meel, R.~Paredes, M.~Y. Vardi, Counting-based
  reliability estimation for power-transmission grids, in: S.~P. Singh,
  S.~Markovitch (Eds.), Proceedings of the Thirty-First {AAAI} Conference on
  Artificial Intelligence (AAAI'17), The AAAI Press, San Francisco, CA, {USA},
  2017, pp. 4488--4494.

\bibitem{ManningRaghavanSchutze08a}
C.~D. Manning, P.~Raghavan, H.~Sch{\"u}tze, Introduction to Information
  Retrieval, Cambridge University Press, Cambridge, 2008.

\bibitem{PourretNaimBruce08a}
O.~Pourret, P.~Naim, M.~Bruce, Bayesian Networks - A Practical Guide to
  Applications, John Wiley \& Sons, 2008.

\bibitem{SahamiDumaisHeckerman98a}
M.~Sahami, S.~Dumais, D.~Heckerman, E.~Horvitz, A {B}ayesian approach to
  filtering junk e-mail, in: T.~Joachims (Ed.), Proceedings of the AAAI-98
  Workshop on Learning for Text Categorization, Vol.~62, 1998, pp. 98--105.

\bibitem{SangBeameKautz05a}
T.~Sang, P.~Beame, H.~Kautz, Performing {B}ayesian inference by weighted model
  counting, in: M.~M. Veloso, S.~Kambhampati (Eds.), Proceedings of the 29th
  National Conference on Artificial Intelligence (AAAI'05), The AAAI Press,
  2005.

\bibitem{XueChoiDarwiche12a}
Y.~Xue, A.~Choi, A.~Darwiche, Basing decisions on sentences in decision
  diagrams, in: J.~Hoffmann, B.~Selman (Eds.), Proceedings of the 26th AAAI
  Conference on Artificial Intelligence (AAAI'12), The AAAI Press, Toronto, ON,
  Canada, 2012.

\bibitem{GomesKautzSabharwalSelman08a}
C.~P. Gomes, A.~Sabharwal, B.~Selman, Chapter 20: Model counting, in: A.~Biere,
  M.~Heule, H.~van Maaren, T.~Walsh (Eds.), Handbook of Satisfiability, Vol.
  185 of Frontiers in Artificial Intelligence and Applications, IOS Press,
  Amsterdam, Netherlands, 2009, pp. 633--654.
\newblock \href {https://doi.org/10.3233/978-1-58603-929-5-633}
  {\path{doi:10.3233/978-1-58603-929-5-633}}.

\bibitem{Valiant79}
L.~Valiant, The complexity of enumeration and reliability problems, SIAM J.
  Comput. 8~(3) (1979) 410--421.

\bibitem{Roth96a}
D.~Roth, On the hardness of approximate reasoning, Artificial Intelligence
  82~(1--2) (1996).
\newblock \href {https://doi.org/10.1016/0004-3702(94)00092-1}
  {\path{doi:10.1016/0004-3702(94)00092-1}}.

\bibitem{ChakrabortyMeelVardi16a}
S.~Chakraborty, K.~S. Meel, M.~Y. Vardi, Improving approximate counting for
  probabilistic inference: From linear to logarithmic {SAT} solver calls, in:
  S.~Kambhampati (Ed.), Proceedings of 25th International Joint Conference on
  Artificial Intelligence (IJCAI'16), The AAAI Press, New York City, NY, USA,
  2016, pp. 3569--3576.

\bibitem{LagniezMarquis17a}
J.-M. Lagniez, P.~Marquis, An improved decision-{DNNF} compiler, in: C.~Sierra
  (Ed.), Proceedings of the Twenty-Sixth International Joint Conference on
  Artificial Intelligence (IJCAI'17), The AAAI Press, 2017.

\bibitem{SaetherTelleVatshelle15a}
S.~H. S{\ae}ther, J.~A. Telle, M.~Vatshelle, Solving \#{SAT} and {MAXSAT} by
  dynamic programming, J. Artif. Intell. Res. 54 (2015) 59--82.

\bibitem{AbiteboulHullVianu95}
S.~Abiteboul, R.~Hull, V.~Vianu, Foundations of Databases: The Logical Level,
  1st Edition, Addison-Wesley, Boston, MA, USA, 1995.
\newblock \href {https://doi.org/10.020.153/7710} {\path{doi:10.020.153/7710}}.

\bibitem{GebserSchaubThieleVeber11}
M.~Gebser, T.~Schaub, S.~Thiele, P.~Veber, Detecting inconsistencies in large
  biological networks with answer set programming, Theory Pract. Log. Program.
  11~(2-3) (2011) 323--360.

\bibitem{GinsbergParkesRoy98a}
M.~L. Ginsberg, A.~J. Parkes, A.~Roy, Supermodels and robustness, in: C.~Rich,
  J.~Mostow (Eds.), Proceedings of the 15th National Conference on Artificial
  Intelligence and 10th Innovative Applications of Artificial Intelligence
  Conference (AAAI/IAAI'98), The AAAI Press, Madison, Wisconsin, USA, 1998, pp.
  334--339.

\bibitem{FichteHecherHamiti20}
J.~K. Fichte, M.~Hecher, F.~Hamiti, The model counting competition 2020, {ACM}
  Journal of Experimental Algorithmics 26 (2021).
\newblock \href {https://doi.org/10.1145/3459080} {\path{doi:10.1145/3459080}}.

\bibitem{LagniezMarquis19}
J.~Lagniez, P.~Marquis, {A Recursive Algorithm for Projected Model Counting},
  in: 33rd Conference on Artificial Intelligence, The AAAI Press, 2019, pp.
  1536--1543.

\bibitem{AzizChuMuise15a}
R.~A. Aziz, G.~Chu, C.~Muise, P.~Stuckey, {\#($\exists$)SAT: Projected Model
  Counting}, in: M.~Heule, S.~Weaver (Eds.), Proceedings of the 18th
  International Conference on Theory and Applications of Satisfiability Testing
  (SAT'15), Springer Verlag, Austin, TX, USA, 2015, pp. 121--137.
\newblock \href {https://doi.org/10.1007/978-3-319-24318-4_10}
  {\path{doi:10.1007/978-3-319-24318-4_10}}.

\bibitem{CapelliMengel19}
F.~Capelli, S.~Mengel, Tractable {QBF} by knowledge compilation, in: 36th
  International Symposium on Theoretical Aspects of Computer Science
  ({STACS'19}), Vol. 126 of LIPIcs, Schloss Dagstuhl - Leibniz-Zentrum
  f{\"{u}}r Informatik, 2019, pp. 18:1--18:16.

\bibitem{FichteHecher19}
J.~K. Fichte, M.~Hecher, Treewidth and counting projected answer sets, in:
  M.~Balduccini, Y.~Lierler, S.~Woltran (Eds.), 15th International Conference
  on Logic Programming and Nonmonotonic Reasoning ({LPNMR} 2019), Vol. 11481 of
  Lecture Notes in Computer Science, Springer, 2019, pp. 105--119.
\newblock \href {https://doi.org/10.1007/978-3-030-20528-7\_9}
  {\path{doi:10.1007/978-3-030-20528-7\_9}}.

\bibitem{DudekPanVardi21}
J.~M. Dudek, V.~H.~N. Phan, M.~Y. Vardi, {ProCount: Weighted Projected Model
  Counting with Graded Project-Join Trees}, in: C.~Li, F.~Many{\`{a}} (Eds.),
  Theory and Applications of Satisfiability Testing ({SAT} 2021), Vol. 12831 of
  Lecture Notes in Computer Science, Springer, 2021, pp. 152--170.
\newblock \href {https://doi.org/10.1007/978-3-030-80223-3\_11}
  {\path{doi:10.1007/978-3-030-80223-3\_11}}.

\bibitem{DurandHermannKolaitis05}
A.~Durand, M.~Hermann, P.~G. Kolaitis, Subtractive reductions and complete
  problems for counting complexity classes, Theoretical Computer Science
  340~(3) (2005) 496--513.
\newblock \href {https://doi.org/10.1016/j.tcs.2005.03.012}
  {\path{doi:10.1016/j.tcs.2005.03.012}}.

\bibitem{GebserKaufmannSchaub09a}
M.~Gebser, B.~Kaufmann, T.~Schaub, Solution enumeration for projected boolean
  search problems, in: W.-J. van Hoeve, J.~N. Hooker (Eds.), Proceedings of the
  6th International Conference on Integration of AI and OR Techniques in
  Constraint Programming for Combinatorial Optimization Problems (CPAIOR'09),
  Vol. 5547 of Lecture Notes in Computer Science, Springer Verlag, Berlin,
  2009, pp. 71--86.
\newblock \href {https://doi.org/10.1007/978-3-642-01929-6_7}
  {\path{doi:10.1007/978-3-642-01929-6_7}}.

\bibitem{CyganEtAl15}
M.~Cygan, F.~V. Fomin, {\L}.~Kowalik, D.~Lokshtanov, M.~P. D{\'a}niel~Marx,
  M.~Pilipczuk, S.~Saurabh, Parameterized Algorithms, Springer Verlag, 2015.
\newblock \href {https://doi.org/10.1007/978-3-319-21275-3}
  {\path{doi:10.1007/978-3-319-21275-3}}.

\bibitem{DowneyFellows13}
R.~G. Downey, M.~R. Fellows, Fundamentals of Parameterized Complexity, Texts in
  Computer Science, Springer Verlag, London, UK, 2013.
\newblock \href {https://doi.org/10.1007/978-1-4471-5559-1}
  {\path{doi:10.1007/978-1-4471-5559-1}}.

\bibitem{FlumGrohe06}
J.~Flum, M.~Grohe, Parameterized Complexity Theory, Vol. XIV of Theoretical
  Computer Science, Springer Verlag, Berlin, 2006.
\newblock \href {https://doi.org/10.1007/3-540-29953-X}
  {\path{doi:10.1007/3-540-29953-X}}.

\bibitem{Niedermeier06}
R.~Niedermeier, Invitation to Fixed-Parameter Algorithms, Vol.~31 of Oxford
  Lecture Series in Mathematics and its Applications, Oxford University Press,
  New York, NY, USA, 2006.

\bibitem{SamerSzeider10b}
M.~Samer, S.~Szeider, Algorithms for propositional model counting, J. Discrete
  Algorithms 8~(1) (2010) 50---64.
\newblock \href {https://doi.org/10.1016/j.jda.2009.06.002}
  {\path{doi:10.1016/j.jda.2009.06.002}}.

\bibitem{DurandMengel13}
A.~Durand, S.~Mengel, {Structural tractability of counting of solutions to
  conjunctive queries}, in: W.~Tan, G.~Guerrini, B.~Catania, A.~Gounaris
  (Eds.), Joint 2013 {EDBT/ICDT} Conferences ({ICDT}'13), {ACM}, 2013, pp.
  81--92.
\newblock \href {https://doi.org/10.1145/2448496.2448508}
  {\path{doi:10.1145/2448496.2448508}}.

\bibitem{ChenMengel17}
H.~Chen, S.~Mengel, {A Trichotomy in the Complexity of Counting Answers to
  Conjunctive Queries}, in: M.~Arenas, M.~Ugarte (Eds.), 18th International
  Conference on Database Theory (ICDT'15), Vol.~31 of LIPIcs, Schloss Dagstuhl
  - Leibniz-Zentrum f{\"{u}}r Informatik, 2015, pp. 110--126.
\newblock \href {https://doi.org/10.4230/LIPIcs.ICDT.2015.110}
  {\path{doi:10.4230/LIPIcs.ICDT.2015.110}}.

\bibitem{GrecoScarcello17}
G.~Greco, F.~Scarcello, {The Power of Local Consistency in Conjunctive Queries
  and Constraint Satisfaction Problems}, {SIAM} J. Comput. 46~(3) (2017)
  1111--1145.
\newblock \href {https://doi.org/10.1137/16M1090272}
  {\path{doi:10.1137/16M1090272}}.

\bibitem{Dell17a}
H.~Dell, C.~Komusiewicz, N.~Talmon, M.~Weller, The pace 2017 parameterized
  algorithms and computational experiments challenge: The second iteration, in:
  IPEC'17, Leibniz International Proceedings in Informatics (LIPIcs), Dagstuhl
  Publishing, 2017, pp. 30:1---30:13.

\bibitem{AbseherMusliuWoltran17a}
M.~Abseher, N.~Musliu, S.~Woltran, htd -- a free, open-source framework for
  (customized) tree decompositions and beyond, in: CPAIOR'17, Vol. 10335 of
  Lecture Notes in Computer Science, Springer Verlag, 2017, pp. 376--386.

\bibitem{Tamaki19}
H.~Tamaki, Positive-instance driven dynamic programming for treewidth, J. Comb.
  Optim. 37~(4) (2019) 1283--1311.
\newblock \href {https://doi.org/10.1007/s10878-018-0353-z}
  {\path{doi:10.1007/s10878-018-0353-z}}.

\bibitem{ManiuSenellartJog2019}
S.~Maniu, P.~Senellart, S.~Jog, An experimental study of the treewidth of
  real-world graph data (extended version), CoRR abs/1901.06862 (2019).
\newblock \href {http://arxiv.org/abs/1901.06862} {\path{arXiv:1901.06862}}.

\bibitem{FichteEtAl20}
J.~K. Fichte, M.~Hecher, P.~Thier, S.~Woltran, Exploiting database management
  systems and treewidth for counting, in: PADL'20, Vol. 12007 of Lecture Notes
  in Computer Science, Springer Verlag, 2020, pp. 151--167.

\bibitem{FichteHecherZisser19}
J.~K. Fichte, M.~Hecher, M.~Zisser, An improved gpu-based {SAT} model counter,
  in: {CP'19}, Vol. 11802 of LNCS, Springer, 2019, pp. 491--509.

\bibitem{FichteEtAl17b}
J.~K. Fichte, M.~Hecher, M.~Morak, S.~Woltran, {DynASP2.5}: Dynamic programming
  on tree decompositions in action, in: D.~Lokshtanov, N.~Nishimura (Eds.),
  Proceedings of the 12th International Symposium on Parameterized and Exact
  Computation (IPEC'17), Dagstuhl Publishing, 2017.
\newblock \href {https://doi.org/10.4230/LIPIcs.IPEC.2017.17}
  {\path{doi:10.4230/LIPIcs.IPEC.2017.17}}.

\bibitem{ImpagliazzoPaturiZane01}
R.~Impagliazzo, R.~Paturi, F.~Zane, Which problems have strongly exponential
  complexity?, J. of Computer and System Sciences 63~(4) (2001) 512--530.
\newblock \href {https://doi.org/10.1006/jcss.2001.1774}
  {\path{doi:10.1006/jcss.2001.1774}}.

\bibitem{LampisMitsou17}
M.~Lampis, V.~Mitsou, Treewidth with a quantifier alternation revisited, in:
  D.~Lokshtanov, N.~Nishimura (Eds.), Proceedings of the 12th International
  Symposium on Parameterized and Exact Computation (IPEC'17), Dagstuhl
  Publishing, 2017.
\newblock \href {https://doi.org/10.4230/LIPIcs.IPEC.2017.17}
  {\path{doi:10.4230/LIPIcs.IPEC.2017.17}}.

\bibitem{FichteEtAl18}
J.~K. Fichte, M.~Hecher, M.~Morak, S.~Woltran, Exploiting treewidth for
  projected model counting and its limits, in: 21st International Conference on
  Theory and Applications of Satisfiability Testing (SAT), Vol. 10929 of
  Lecture Notes in Computer Science, Springer Verlag, 2018, pp. 165--184.
\newblock \href {https://doi.org/10.1007/978-3-319-94144-8\_11}
  {\path{doi:10.1007/978-3-319-94144-8\_11}}.

\bibitem{HecherThierWoltran20}
M.~Hecher, P.~Thier, S.~Woltran, {Taming High Treewidth with Abstraction,
  Nested Dynamic Programming, and Database Technology}, in: 23rd International
  Conference on Theory and Applications of Satisfiability Testing {SAT}, Vol.
  12178 of Lecture Notes in Computer Science, Springer Verlag, 2020, pp.
  343--360.

\bibitem{GrahamGrotschelLovasz95a}
R.~L. Graham, M.~Gr{\"o}tschel, L.~Lov{\'a}sz, Handbook of Combinatorics,
  Vol.~I, Elsevier Science Publishers, North-Holland, 1995.

\bibitem{KleineBuningLettman99}
H.~Kleine~B{\"u}ning, T.~Lettman, Propositional logic: deduction and
  algorithms, Cambridge University Press, Cambridge, New York, NY, USA, 1999.

\bibitem{Papadimitriou94}
C.~H. Papadimitriou, Computational Complexity, Addison-Wesley, 1994.

\bibitem{StockmeyerMeyer73}
L.~J. Stockmeyer, A.~R. Meyer, Word problems requiring exponential time, in:
  A.~V. Aho, A.~Borodin, R.~L. Constable, R.~W. Floyd, M.~A. Harrison, R.~M.
  Karp, H.~R. Strong (Eds.), Proceedings of the 5th Annual ACM Symposium on
  Theory of Computing (STOC'73), Assoc. Comput. Mach., New York, Austin, TX,
  {USA}, 1973, pp. 1--9.
\newblock \href {https://doi.org/10.1145/800125.804029}
  {\path{doi:10.1145/800125.804029}}.

\bibitem{BiereHeuleMaarenWalsh09}
A.~Biere, M.~Heule, H.~van Maaren, T.~Walsh (Eds.), Handbook of Satisfiability,
  Vol. 185 of Frontiers in Artificial Intelligence and Applications, IOS Press,
  Amsterdam, Netherlands, 2009.

\bibitem{HemaspaandraVollmer95a}
L.~A. Hemaspaandra, H.~Vollmer, The satanic notations: Counting classes beyond
  \#{P} and other definitional adventures, SIGACT News 26~(1) (1995) 2--13.
\newblock \href {https://doi.org/10.1145/203610.203611}
  {\path{doi:10.1145/203610.203611}}.

\bibitem{Diestel12}
R.~Diestel, Graph Theory, 4th Edition, Vol. 173 of Graduate Texts in
  Mathematics, Springer Verlag, 2012.

\bibitem{BondyMurty08}
J.~A. Bondy, U.~S.~R. Murty, Graph theory, Vol. 244 of Graduate Texts in
  Mathematics, Springer Verlag, New York, USA, 2008.

\bibitem{Bodlaender96}
H.~L. Bodlaender, A linear-time algorithm for finding tree-decompositions of
  small treewidth, SIAM J. Comput. 25~(6) (1996) 1305--1317.

\bibitem{BodlaenderKoster08}
H.~L. Bodlaender, A.~M. C.~A. Koster, Combinatorial optimization on graphs of
  bounded treewidth, The Computer Journal 51~(3) (2008) 255--269.
\newblock \href {https://doi.org/10.1093/comjnl/bxm037}
  {\path{doi:10.1093/comjnl/bxm037}}.

\bibitem{FichteEtAl17a}
J.~K. Fichte, M.~Hecher, M.~Morak, S.~Woltran, Answer set solving with bounded
  treewidth revisited, in: M.~Balduccini, T.~Janhunen (Eds.), Proceedings of
  the 14th International Conference on Logic Programming and Nonmonotonic
  Reasoning (LPNMR'17), Vol. 10377 of Lecture Notes in Computer Science,
  Springer Verlag, Espoo, Finland, 2017, pp. 132--145.
\newblock \href {https://doi.org/10.1007/978-3-319-61660-5_13}
  {\path{doi:10.1007/978-3-319-61660-5_13}}.

\bibitem{BodlaenderKloks96}
H.~L. Bodlaender, T.~Kloks, Efficient and constructive algorithms for the
  pathwidth and treewidth of graphs, J. Algorithms 21~(2) (1996) 358--402.

\bibitem{BannachBerndt22}
M.~Bannach, S.~Berndt, {Recent Advances in Positive-Instance Driven Graph
  Searching}, Algorithms 15~(2) (2022) 42.
\newblock \href {https://doi.org/10.3390/a15020042}
  {\path{doi:10.3390/a15020042}}.

\bibitem{Wilder12a}
R.~L. Wilder, Introduction to the Foundations of Mathematics, 2nd Edition, John
  Wiley \& Sons, 1965.

\bibitem{PichlerRuemmeleWoltran10}
R.~Pichler, S.~R{\"u}mmele, S.~Woltran, Counting and enumeration problems with
  bounded treewidth, in: E.~M. Clarke, A.~Voronkov (Eds.), Proceedings of the
  16th International Conference on Logic for Programming, Artificial
  Intelligence, and Reasoning (LPAR'10), Vol. 6355 of Lecture Notes in Computer
  Science, Springer Verlag, 2010, pp. 387--404.
\newblock \href {https://doi.org/10.1007/978-3-642-17511-4_22}
  {\path{doi:10.1007/978-3-642-17511-4_22}}.

\bibitem{Knuth1998}
D.~E. Knuth, How fast can we multiply?, in: The Art of Computer Programming,
  3rd Edition, Vol.~2 of Seminumerical Algorithms, Addison-Wesley, 1998, Ch.
  4.3.3, pp. 294--318.

\bibitem{Harvey2016}
D.~Harvey, J.~van~der Hoeven, G.~Lecerf, Even faster integer multiplication, J.
  Complexity 36 (2016) 1--30.
\newblock \href {https://doi.org/https://doi.org/10.1016/j.jco.2016.03.001}
  {\path{doi:https://doi.org/10.1016/j.jco.2016.03.001}}.

\bibitem{FichteHecherPfandler20}
J.~K. Fichte, M.~Hecher, A.~Pfandler, {Lower Bounds for QBFs of Bounded
  Treewidth}, in: 35th Annual ACM/IEEE Symposium on Logic in Computer Science
  (LICS'20), Assoc. Comput. Mach., New York, 2020, pp. 410--424.

\bibitem{BannachBerndt19}
M.~Bannach, S.~Berndt, Practical access to dynamic programming on tree
  decompositions, Algorithms 12~(8) (2019) 172.

\bibitem{DellRothWellnitz19}
H.~Dell, M.~Roth, P.~Wellnitz, Counting answers to existential questions, in:
  {ICALP'19}, Vol. 132 of LIPIcs, Schloss Dagstuhl - Leibniz-Zentrum f{\"{u}}r
  Informatik, 2019, pp. 113:1--113:15.

\bibitem{EibenEtAl19}
E.~Eiben, R.~Ganian, T.~Hamm, O.~Kwon, Measuring what matters: {A} hybrid
  approach to dynamic programming with treewidth, in: {MFCS'19}, Vol. 138 of
  LIPIcs, Dagstuhl Publishing, 2019, pp. 42:1--42:15.

\bibitem{stacs:GanianRS17}
R.~Ganian, M.~S. Ramanujan, S.~Szeider, Combining treewidth and backdoors for
  {CSP}, in: STACS'17, 2017, pp. 36:1--36:17.
\newblock \href {https://doi.org/10.4230/LIPIcs.STACS.2017.36}
  {\path{doi:10.4230/LIPIcs.STACS.2017.36}}.

\bibitem{HecherMorakWoltran20}
M.~Hecher, M.~Morak, S.~Woltran,
  \href{https://aaai.org/ojs/index.php/AAAI/article/view/5672}{{Structural
  Decompositions of Epistemic Logic Programs}}, in: AAAI'20, {AAAI} Press,
  2020, pp. 2830--2837.
\newline\urlprefix\url{https://aaai.org/ojs/index.php/AAAI/article/view/5672}

\bibitem{LagniezMarquis14}
J.~Lagniez, P.~Marquis, {Preprocessing for Propositional Model Counting}, in:
  28th AAAI Conference on Artificial Intelligence (AAAI), The AAAI Press, 2014,
  pp. 2688--2694.

\bibitem{Ullman89}
J.~D. Ullman, Principles of Database and Knowledge-Base Systems, Volume {II},
  Computer Science Press, New York, NY, USA, 1989.

\bibitem{Garcia-MolinaUllmanWidom09}
H.~Garcia-Molina, J.~D. Ullman, J.~Widom, Database systems: the complete book,
  2nd Edition, Pearson Prentice Hall, Upper Saddle River, New Jersey, 2009.

\bibitem{ElmasriNavathe16}
R.~Elmasri, S.~B. Navathe, Fundamentals of Database Systems, 7th Edition,
  Pearson, 2016.

\bibitem{Thurley06a}
M.~Thurley, {{sharpSAT} -- Counting Models with Advanced Component Caching and
  Implicit {BCP}}, in: 9th International Conference on Theory and Applications
  of Satisfiability Testing (SAT), Springer Verlag, 2006, pp. 424--429.

\bibitem{Biere08}
A.~Biere, {PicoSAT Essentials}, J. on Satisfiability, Boolean Modeling and
  Computation 4~(2-4) (2008) 75--97.

\bibitem{GebserEtAl19}
M.~Gebser, R.~Kaminski, B.~Kaufmann, T.~Schaub, {Multi-shot {ASP} solving with
  clingo}, Theory Pract. Log. Program. 19~(1) (2019) 27--82.
\newblock \href {https://doi.org/10.1017/S1471068418000054}
  {\path{doi:10.1017/S1471068418000054}}.

\bibitem{OztokDarwiche15a}
U.~Oztok, A.~Darwiche, {A Top-Down Compiler for Sentential Decision Diagrams},
  in: 24th International Joint Conference on Artificial Intelligence (IJCAI),
  The AAAI Press, 2015, pp. 3141--3148.

\bibitem{Darwiche04a}
A.~Darwiche, {New Advances in Compiling {CNF} to Decomposable Negation Normal
  Form}, in: 16th Eureopean Conference on Artificial Intelligence (ECAI), IOS
  Press, 2004, pp. 318--322.

\bibitem{Darwiche11a}
A.~Darwiche, {{SDD:} {A} New Canonical Representation of Propositional
  Knowledge Bases}, in: 22nd International Joint Conference on Artificial
  Intelligence (IJCAI), AAAI Press/IJCAI, 2011, pp. 819--826.

\bibitem{MuiseEtAl12a}
S.~A. Muise, Christian J .and~McIlraith, J.~C. Beck, E.~I. Hsu, {{Dsharp}: Fast
  d-{DNNF} Compilation with {sharpSAT}}, in: 25th Canadian Conference on
  Artificial Intelligence (AI), Vol. 7310 of Lecture Notes in Computer Science,
  Springer Verlag, 2012, pp. 356--361.

\bibitem{KoricheLagniezMarquisThomas13a}
F.~Koriche, J.-M. Lagniez, P.~Marquis, S.~Thomas, {Knowledge Compilation for
  Model Counting: Affine Decision Trees}, in: 23rd International Joint
  Conference on Artificial Intelligence (IJCAI), {IJCAI/AAAI}, 2013.

\bibitem{TodaSoh15a}
T.~Toda, T.~Soh, {Implementing Efficient All Solutions {SAT} Solvers}, {ACM}
  Journal of Experimental Algorithmics 21 (2015) 1.12, special Issue SEA 2014.

\bibitem{SangEtAl04}
T.~Sang, F.~Bacchus, P.~Beame, H.~Kautz, T.~Pitassi, {Combining Component
  Caching and Clause Learning for Effective Model Counting}, in: 7th
  International Conference on Theory and Applications of Satisfiability Testing
  (SAT), 2004.

\bibitem{SharmaEtAl19}
S.~Sharma, S.~Roy, M.~Soos, K.~S. Meel, {{GANAK:} {A} Scalable Probabilistic
  Exact Model Counter}, in: 28th International Joint Conference on Artificial
  Intelligence (IJCAI), ijcai.org, 2019, pp. 1169--1176.

\bibitem{ErmonGomesSelman12a}
S.~Ermon, C.~P. Gomes, B.~Selman, {Uniform Solution Sampling Using a Constraint
  Solver As an Oracle}, in: 28th Conference on Uncertainty in Artificial
  Intelligence (UAI), AUAI Press, 2012, pp. 255--264.

\bibitem{KlebanovEtAl13}
V.~Klebanov, N.~Manthey, C.~J. Muise, {SAT-Based Analysis and Quantification of
  Information Flow in Programs}, in: 10th International Conference on
  Quantitative Evaluation of Systems (QEST), Vol. 8054 of Lecture Notes in
  Computer Science, Springer Verlag, 2013, pp. 177--192.

\bibitem{ChakrabortyEtAl14a}
S.~Chakraborty, D.~J. Fremont, K.~S. Meel, S.~A. Seshia, M.~Y. Vardi,
  {Distribution-Aware Sampling and Weighted Model Counting for {SAT}}, in: 28th
  AAAI Conference on Artificial Intelligence (AAAI), The AAAI Press, 2014, pp.
  1722--1730.

\bibitem{BurchardSchubertBecker15a}
J.~Burchard, T.~Schubert, B.~Becker, {Laissez-Faire Caching for Parallel
  {\#}{SAT} Solving}, in: 18th International Conference on Theory and
  Applications of Satisfiability Testing (SAT), Vol. 9340 of Lecture Notes in
  Computer Science, Springer Verlag, 2015, pp. 46--61.

\bibitem{BurchardSchubertBecker16a}
J.~Burchard, T.~Schubert, B.~Becker, {Distributed Parallel {\#}SAT Solving},
  in: 18th IEEE International Conference on Cluster Computing (CLUSTER), {IEEE}
  Computer Society, 2016, pp. 326--335.

\bibitem{LagniezMarquisSzczepanski18a}
J.-M. Lagniez, P.~Marquis, N.~Szczepanski, {DMC: A Distributed Model Counter},
  in: 27th International Joint Conference on Artificial Intelligence (IJCAI),
  The AAAI Press, 2018, pp. 1331--1338.

\bibitem{GabrielFaggBosilca04}
E.~Gabriel, G.~E. Fagg, G.~Bosilca, T.~Angskun, J.~J. Dongarra, J.~M. Squyres,
  V.~Sahay, P.~Kambadur, B.~Barrett, A.~Lumsdaine, R.~H. Castain, D.~J. Daniel,
  R.~L. Graham, T.~S. Woodall, {Open {MPI}: Goals, Concept, and Design of a
  Next Generation {MPI} Implementation}, in: 11th European PVM/MPI Users' Group
  Meeting, Lecture Notes in Computer Science, 2004, pp. 97--104.

\bibitem{CharwatWoltran16a}
G.~Charwat, S.~Woltran, Dynamic programming-based {QBF} solving, in:
  F.~Lonsing, M.~Seidl (Eds.), Proceedings of the 4th International Workshop on
  Quantified Boolean Formulas (QBF'16), Vol. 1719, CEUR Workshop Proceedings
  (CEUR-WS.org), 2016, pp. 27--40, co-located with 19th International
  Conference on Theory and Applications of Satisfiability Testing (SAT'16).

\end{thebibliography}


\longversion{
\appendix
\newpage




\section{Omitted Proofs}

\begin{restateobservation}[obs:relation]%
\begin{observation}
The relation~$\bucket$ is an equivalence relation.
\end{observation}
\end{restateobservation}
\begin{proof}
One can easily see that~$=_P(A,B) \eqdef (A \cap P) = (B \cap P)$ is 
\begin{itemize}
	\item reflexive: $A\cap P = A \cap P$ for any two sets~$A, P$,
	\item symmetric: $A\cap P = B \cap P$ if and only if
          $B\cap P = A\cap P$ for given sets~$A,B,P$, and
	\item transitive: if $A \cap P = B \cap P$  and~$B \cap P = C \cap P$, then~$A\cap P = C\cap P$ holds as well for any sets~$A,B,C,P$.
\end{itemize}
As a result, $=_P$ is an equivalence relation.
\end{proof}

\shortversion{
\paragraph{\textbf{Global Assumptions}}
Here, we fix requirements for almost all statements in the following.
We assume that we have given an arbitrary instance~$(F,P)$ of \PMC and
a tree decomposition~$\TTT = (T,\chi)$ of
formula~$F$, where $T=(N, A)$, node $n=\rootOf(T)$ is the root and $\TTT$ is of width~$k$.
Moreover, for every~$t \in N$ of tree decomposition~$\TTT$, we let
$\ATabs{\AlgS}{t}$ be the tables that have been computed by running
algorithm~$\dpa_\AlgS$ for the dedicated input. Analogously, let
$\ATabs{\PROJ}{t}$ be the tables that have been computed by running
algorithm~$\dpa_\PROJ$ for the dedicated input.

\begin{lemma}\label{lem:runtime}
  For every node~$t\in N$, $\ATabs{\PROJ}{t}$ contains at
  most~$2^{2^{k+1}}$ rows.
\end{lemma}
\begin{proof}
  By definition of a tree decomposition, for every node~$t\in N$ the
  bag~$\chi(t)$ is of size at most~$k+1$. Therefore, we have at
  most~$2^{k+1}$ many~\cite{SamerSzeider10b} rows in the table
  obtained via~$\dpa_\AlgS$. In the end, we have at most $2^{2^{k+1}}$
  rows within table~$\ATabs{\PROJ}{t}$, since we have a row for each
  of the $2^{2^{k+1}}$ many subsets of a $\AlgS$-row.
\end{proof}

In the following, we state definitions required for the correctness
proofs of our algorithm \PROJ.  In the end, we only store rows that
are restricted to the bag content to maintain runtime bounds. In
related work~\cite{SamerSzeider10b}, it was shown that this suffices
for table algorithm~$\AlgS$, i.e., \PRIM and~\INC are both sound and
complete.  Similar to related work~\cite{FichteEtAl17a}, we define the
content of our tables in two steps. First, we define the properties of
so-called \emph{$\PROJ$-solutions up to~$t$}. Second, we restrict
these solutions to~\emph{$\PROJ$-row solutions} at~$t$.

%
%

\begin{definition}\label{def:globalsol}
  Let~$\emptyset \subsetneq \sigma \subseteq \ATab{\AlgS}[t]$ be a
  table with $\sigma \in \subbuckets_P(\ATab{\AlgS}[t])$.
  We define a \emph{${\PROJ}$-solution up to~$t$} to be the sequence
  $\langle \hat \sigma\rangle = \langle\PExt_{\leq t}(\sigma)\rangle$.
\end{definition}

%
%

Next, we recall that we can reconstruct all models from the tables.

\begin{proposition}[cf.,~\cite{SamerSzeider10b}]\label{prop:sat}
  \[I(\PExt_{\leq n}(\ATab{\AlgS}[n])) = I(\Exts) = \{J \mid J \in
    \ta{\var(F)}, \alpha_J\models F\}.\]
\end{proposition}
\begin{proof}[Idea]
  In fact, we can use the construction by Samer and
  Szeider~\cite{SamerSzeider10b} of the tables and extend them
  globally. Then, the extensions simply collect the corresponding,
  preceding rows. By taking the interpretation parts $I(\cdots)$ of
  these collected rows we obtain the set of all models of the formula.
  A similar construction is used by Pichler, R\"ummele, and
  Woltran~\cite[Fig.~1]{PichlerRuemmeleWoltran10}, however, hidden in
  an algorithm to enumerate solutions.
\end{proof}}

\shortversion{
Before we present equivalence results between~$\ipmc_{\leq t}(\ldots)$
and the recursive version~$\ipmc(t, \ldots)$
(Definition~\ref{def:ipmc}) used during the computation of
$\dpa_\PROJ$, recall that~$\ipmc_{\leq t}$ and~$\pmc_{\leq t}$
(Definition~\ref{def:pmc}) are key to compute the projected model
count. The following corollary states that computing $\ipmc_{\leq n}$
at the root~$n$ actually suffices to compute~$\pmc_{\leq n}$, which is
in fact the projected model count of the input formula.

\begin{corollary}\label{cor:psat}
  \begin{align*}
    \ipmc_{\leq n}(\ATab{\AlgS}[n]) =& \pmc_{\leq n}(\ATab{\AlgS}[n])\\
    =& \Card{I_P(\PExt_{\leq n}(\ATab{\AlgS}[n]))}\\
    =& \Card{I_P(\Exts)}\\
    =& \Card{\{J \cap P
       \mid J \in \ta{\var(F)}, \alpha_J\models F\}}
  \end{align*}
\end{corollary}
\begin{proof}
  The corollary immediately follows from Proposition~\ref{prop:sat}
  and the observation that the cardinality of $\ATab{\AlgS}[n]$ is at
  most one at root~$n$, by properties of the algorithm~$\AlgS$ and
  since $\chi(n) = \emptyset$.
\end{proof}

The following lemma establishes that the \PROJ-solutions up to
root~$n$ of a given tree decomposition solve the \PMC problem.

\begin{lemma}\label{lem:global}
  The
  value~$c = \sum_{\langle\hat\sigma\rangle\text{ is a \PROJ-solution
      up to } n}\Card{I_P(\hat \sigma)}$ if and only if $c$ is the
  projected model count of~$F$ with respect to the set~$P$ of
  projection variables.
\end{lemma}
\begin{proof}
  (``$\Longrightarrow$''): Assume
  that~$c = \sum_{\langle\hat\sigma\rangle\text{ is a \PROJ-solution
      up to } n}\Card{I_P(\hat \sigma)}$. Observe that there can be at
  most one projected solution up to~$n$,
  since~$\chi(n)=\emptyset$. %
  If~$c=0$, then $\ATab{\AlgS}[n]$ contains no rows. Hence, $F$ has no
  models,~cf., Proposition~\ref{prop:sat}, and obviously also no
  models projected to~$P$. Consequently, $c$ is the projected model
  count of~$F$.  
  If~$c>0$ we have by Corollary~\ref{cor:psat} that~$c$ is
  equivalent to the projected model count of~$F$ with respect to~$P$.

  (``$\Longleftarrow$''): The proof proceeds similar to the only-if
  direction.
\end{proof}

\medskip %

In the following, we provide for a given node~$t$ and a given \PROJ-solution up to~$t$,
the definition of a \PROJ-row solution at~$t$.

\begin{definition}\label{def:loctab}~
  \begin{enumerate}
    \item 
  Let $\hat\sigma$ be a~\PROJ-solution up to some node~$t'$. 
  %
  %
  Then, we define \emph{the local table for node}~$t$ as
  $\local(t,\hat\sigma)\eqdef \{ \langle \vec{\tabval}\rangle \mid
  \langle t, \vec{\tabval}\rangle \in \hat\sigma\}$.

  \item
  Let $t \in N$ be a node of the tree decomposition~$\TTT$ and
  $\langle \hat\sigma \rangle$ be a~$\PROJ$-solution up to~$t$. Then, we
  define the \emph{$\PROJ$-row solution at $t$} by
  $\langle \local(t,\hat\sigma), \Card{I_P(\hat\sigma)}\rangle$.
\end{enumerate}
\end{definition}






\begin{observation}\label{obs:unique}
  Let $\langle \hat \sigma\rangle$ be a \PROJ-solution up to a
  node~$t\in N$.  There is exactly one corresponding \PROJ-row
  solution
  $\langle \local(t,\hat\sigma), \Card{I_P(\hat\sigma)}\rangle$ at~$t$.

  Vice versa, let $\langle \sigma, c\rangle$ at~$t$ be a \PROJ-row
  solution at~$t$ for some integer~$c$. Then, there is exactly one
  corresponding \PROJ-solution~$\langle\PExt_{\leq t}(\sigma)\rangle$
  up to~$t$.
\end{observation}

We need to ensure that storing~$\PROJ$-row solutions at a
node~$t \in N$ suffices to solve the~\PMC problem, which is necessary
to obtain the runtime bounds as presented in
Corollary~\ref{cor:runtime}.

\begin{lemma}\label{lem:local}
  Let $t\in N$ be a node of the tree decomposition~$\TTT$.  There is a
  \PROJ-row solution at the root~$n$ if and only if the projected
  model count of~$F$ is larger than zero.
\end{lemma}
\begin{proof}%

  (``$\Longrightarrow$''): Let $\langle \sigma, c\rangle$ be a
  \PROJ-row solution at root~$n$ where $\sigma$ is a $\AlgS$-table and
  $c$ is a positive integer. Then, by Definition~\ref{def:loctab}
  there also exists a
  corresponding~$\PROJ$-solution~$\langle \hat\sigma \rangle$ up
  to~$n$ such that $\sigma = \local(t,\hat\sigma)$ and
  $c=\Card{I_P(\hat\sigma)}$.
  Moreover, since~$\chi(n)=\emptyset$, we
  have~$\Card{\ATab{\AlgS}[n]}=1$.  
  Then, by Definition~\ref{def:globalsol}
  $\hat \sigma = \ATab{\AlgS}[n]$. By Corollary~\ref{cor:psat}, we
  have $c=\Card{I_P(\ATab{\AlgS}[n])}$.
  Finally, the claim follows.
  

  (``$\Longleftarrow$''): The proof proceeds similar to the only-if
  direction.
\end{proof}}

\begin{restateobservation}[obs:main_incl_excl]
\begin{observation}
  Let $n$ be a positive integer, $X = \{1, \ldots, n\}$, and $X_1$,
  $X_2$, $\ldots$, $X_n$ subsets of $X$.
  The number of elements in the intersection over all sets~$A_i$ is
  \begin{align*}
    \Card{\bigcap_{i \in X} X_i} 
    =& %
       \Bigg|\Card{\bigcup^n_{j = 1} X_j} &&- %
                                       \sum_{\emptyset \subsetneq I \subsetneq X, \Card{I}=1}
                                       \Card{\bigcap_{i \in I} X_i} + %
                                       \sum_{\emptyset \subsetneq I \subsetneq X, \Card{I}=2}
                                       \Card{\bigcap_{i \in I} X_i} - \ldots \\
                                         & &&+ \sum_{\emptyset \subsetneq I \subsetneq X, \Card{I}=n-1} (-1)^{\Card{I}} 
                                              \Card{\bigcap_{i \in I} X_i}\Bigg|.
  \end{align*}
It trivially works to count arbitrary sets.
\end{observation}
\end{restateobservation}
\begin{proof}
  We take the well-known inclusion-exclusion
  principle~\cite{GrahamGrotschelLovasz95a} and rearrange the
  equation.
  \begin{align*}
    \Card{\bigcup^n_{j = 1} X_j} =& \sum_{\emptyset \subsetneq I \subseteq X} (-1)^{\Card{I}-1} &&\Card{\bigcap_{i \in I} X_i}\\
    \Card{\bigcup^n_{j = 1} X_j} =& \sum_{\emptyset \subsetneq I \subsetneq X} (-1)^{\Card{I}-1} &&\Card{\bigcap_{i \in I} X_i} + (-1)^{\Card{X}-1} \Card{\bigcap_{i \in X} X_i}\\
    (-1)^{\Card{X}-1} \Card{\bigcap_{i \in X} X_i}  =& &&\Card{\bigcup^n_{j = 1} X_j} - \sum_{\emptyset \subsetneq I \subsetneq X} (-1)^{\Card{I}-1} \Card{\bigcap_{i \in I} X_i} \\
    \Card{\bigcap_{i \in X} X_i}  =& \Bigg|&&\Card{\bigcup^n_{j = 1} X_j} - \sum_{\emptyset \subsetneq I \subsetneq X} (-1)^{\Card{I}-1} \Card{\bigcap_{i \in I} X_i}\Bigg| \\
    \Card{\bigcap_{i \in X} X_i} =& \Bigg|%
                                      &&\Card{\bigcup^n_{j = 1} X_j} - %
                                      \sum_{\emptyset \subsetneq I
                                        \subsetneq X, \Card{I}=1}
                                      \Card{\bigcap_{i \in I} X_i} \\ %
    &&&                                 + \sum_{\emptyset \subsetneq I
                                        \subsetneq X, \Card{I}=2}
                                      \Card{\bigcap_{i \in I} X_i} \\
    &&&                                  - \ldots \\
    &&& + \sum_{\emptyset \subsetneq I \subsetneq X,
                                        \Card{I}=n-1} (-1)^{\Card{I}
                                        } \Card{\bigcap_{i \in I}
                                        X_i}\Bigg|
  \end{align*}
\end{proof}

\shortversion{
\begin{lemma}\label{lem:main_incl_excl}
  Let $t\in N$ be a node of the tree decomposition~$\TTT$
  with~$\children(t) = \langle t_1, \ldots, t_\ell \rangle$ and let
  $\langle\sigma,\cdot\rangle$ be a~\PROJ-row solution at~$t$.
  Then,
  \begin{enumerate}
  \item %
    $\ipmc(t,\sigma,\langle\ATab{\PROJ}[t_1], \ldots,
    \ATab{\PROJ}[t_{\ell}]\rangle) = \ipmc_{\leq t}(\sigma)$
  \item \smallskip%
    for $\type(t) \neq \leaf$:\\
    $\pmc(t,\sigma,\langle\ATab{\PROJ}[t_1], \ldots,
    \ATab{\PROJ}[t_{\ell}]\rangle) = \pmc_{\leq t}(\sigma)$.
  \end{enumerate}
\end{lemma}
\begin{proof}[Proof (Sketch)]
  We prove the statement by simultaneous induction.
  
  (``Induction Hypothesis''): Lemma~\ref{lem:main_incl_excl} holds for the nodes in~$\children(t)$ and also for node~$t$, but on strict subsets~$\rho\subsetneq\sigma$.

  (``Base Cases''): Let $\type(t) = \leaf$.
  Then by definition,
  $\ipmc(t,\emptyset, \langle \rangle) = \ipmc_{\leq t}(\emptyset) =
  1$.  
  Recall that for $\pmc$ the equivalence does not hold for leaves, but we use a node
  that has a node~$t'\in N$ with~$\type(t') = \leaf$ as child for the
  base case. Observe that by definition of a nice tree decomposition
  such a node~$t$ can have exactly one child.
  Then, we have that
  $\pmc(t,\sigma,\langle\ATab{\PROJ}[t']\rangle) = \sum_{\emptyset
    \subsetneq O \subseteq {\origs(t,\sigma)}} (-1)^{(\Card{O} - 1)}
  \cdot \sipmc(\langle \ATab{\AlgS}[t']\rangle, O) =
  \Card{\bigcup_{\vec u\in\sigma} I_P(\PExt_{\leq t}(\{\vec u\}))} =
  \pmc_{\leq t}(\sigma) = 1$ where $\langle\sigma,\cdot\rangle$ is
  a~\PROJ-row solution at~$t$.

  (``Induction Step''): We proceed by case distinction.

  Assume that $\type(t) = \intr$.
  Let $a \in (\chi(t) \setminus \chi(t'))$ be an introduced
  variable. We have two cases. Assume Case (i): $a$ also belongs to
  $(\var(F) \setminus P)$,~i.e., $a$ is not a projection variable. 
  %
  %
  %
  Let~$\langle \sigma, c \rangle$ be a \PROJ-row solution at~$t$ for
  some integer~$c$. By construction of the table algorithm~$\AlgS$
  there are many rows in the table~$\ATab{\AlgS}[t]$ for one row in
  the table~$\ATab{\AlgS}[t']$, more precisely,
  $\Card{\buckets_P(\sigma)} = 1$.
  As a result,
  $\pmc_{\leq t}(\sigma) = \pmc_{\leq t'}(\orig(t,\sigma))$ by
  applying Observation~\ref{obs:unique}.
  We apply the inclusion-exclusion principle on every subset~$\rho$ of
  the origins of~$\sigma$ in the definition of~$\pmc$ and by induction
  hypothesis we know that
  $\ipmc(t',\rho,\langle\ATab{\PROJ}[t']\rangle) = \ipmc_{\leq
    t'}(\rho)$, therefore,
  $\sipmc(\ATab{\PROJ}[t'], \rho) = \ipmc_{\leq t'}(\rho)$.  This
  concludes Case~(i) for $\pmc$. The induction step for $\ipmc$ works
  similar, but swapped by applying
  Observation~\ref{obs:main_incl_excl} and comparing the corresponding
  \PROJ-solutions up to~$t$ or $t'$, respectively. Further, for showing the lemma for~$\ipmc$, one has to additionally apply the hypothesis for node~$t$, but on strict subsets~$\emptyset\subsetneq\rho\subsetneq\sigma$ of~$\sigma$.
  %
  Assume that we have Case~(ii): $a$ also belongs to $P$,~i.e., $a$ is a projection
  variable. We proceed similar as in Case~(i), and obtain that
  $\Card{\buckets_P(\sigma)} = 1$.

  Assume that $\type(t) = \rem$. Let
  $a \in (\chi(t') \setminus \chi(t))$ be a removed variable. We have
  two cases. Case (i) $a$ also belongs to
  $(\var(F) \setminus P)$,~i.e., $a$ is not a projection variable; and
  Case (ii) $a$ also belongs to $P$,~i.e., $a$ is a projection
  variable.
  Assume that we have Case~(i).  Let~$\langle \sigma, c \rangle$ be a
  \PROJ-row solution at~$t$ for some integer~$c$.
  By construction of the table algorithm~$\AlgS$ there are many rows
  in the table~$\ATab{\AlgS}[t]$ for one row in the
  table~$\ATab{\AlgS}[t']$ (and vice-versa). Nonetheless we still have
  $\pmc_{\leq t}(\sigma) = \pmc_{\leq t'}(\orig(t,\sigma))$, because
  $a \notin P$ by applying Observation~\ref{obs:unique}.
  We apply the inclusion-exclusion principle on every subset~$\rho$ of
  the origins of~$\sigma$ in the definition of~$\pmc$ and by induction
  hypothesis we know that
  $\ipmc(t',\rho,\langle\ATab{\PROJ}[t']\rangle) = \ipmc_{\leq
    t'}(\rho)$, therefore,
  $\sipmc(\ATab{\PROJ}[t'], \rho) = \ipmc_{\leq t'}(\rho)$.  This
  concludes Case~(i) for $\pmc$. Again, the induction step for $\ipmc$
  works similar, but swapped.
  Assume that we have Case~(ii).
  Let~$\langle \sigma, c \rangle$ be a \PROJ-row solution at~$t$ for
  some integer~$c$.
  Here we cannot ensure
  $\pmc_{\leq t}(\sigma) = \pmc_{\leq t'}(\orig(t,\sigma))$, since
  buckets fall together.  However, by applying
  Observation~\ref{obs:unique} we have
  $\pmc_{\leq t}(\sigma) = \sum_{\rho \in
    \buckets_P(\origs(t,\sigma)_{(1)})} \pmc(t', \rho, C) $ where the
  sequence~$C$ consists of the tables of the children of~$t'$.
  For every~$\rho \in \subbuckets_P(\origs(t,\sigma)_{(1)})$ by
  induction hypothesis we know that
  $\ipmc(t',\rho,\langle\ATab{\PROJ}[t']\rangle) = \ipmc_{\leq
    t'}(\rho)$.
  Hence, we apply the inclusion-exclusion principle over all
  subsets~$\zeta$ of~$\rho$ for all~$\rho$ independently.  By
  construction
  $\sipmc(\ATab{\PROJ}[t'], \zeta) = \ipmc_{\leq t'}(\zeta)$.  Then,
  by construction
  $\pcnt(t,\sigma, C') = \sum_{\emptyset \subsetneq O \subseteq
    {\origs(t,\sigma)}} (-1)^{(\Card{O} - 1)} \cdot \sipmc(C',\allowbreak O) =
  \pmc_{\leq t}(\sigma)$ where
  $C' = \langle \ATab{\PROJ}[t'] \rangle$, since for the remaining
  terms $\sipmc(C', O)$ is simply zero, including cases where
  different buckets are involved.
  This concludes Case~(ii) for $\pmc$. Again, the induction step for
  $\ipmc$ works similar, but swapped by again applying
  Observation~\ref{obs:main_incl_excl}.

  Assume that $\type(t) = \join$. We proceed similar to the introduce
  case. However, we have two \PROJ-tables for the children of~$t$.
  Hence, we have to both sides when computing $\sipmc$
  (Definition~\ref{def:childpcnt}). There we consider the
  cross-product of two \AlgS-tables and we can also correctly apply
  the inclusion-exclusion principle on subsets of this cross-product,
  which we can do by simply multiplying $\sipmc$-values
  accordingly. The multiplication is closely related to the join case
  in table algorithm~\AlgS. For $\ipmc$ this does not apply, since the
  inclusion-exclusion principle is carried out at the node~$t$ and not
  for its children.

  Since we outlined all cases that can occur for node~$t$, this
  concludes the proof sketch.
\end{proof}


\begin{lemma}[Soundness]\label{lem:correct}
  Let $t\in N$ be a node of the tree decomposition~$\TTT$
  with~$\children(t) = \langle t_1, \ldots, t_\ell \rangle$.
  Then, each row~$\langle \tab{}, c \rangle$ at node~$t$ constructed
  by table algorithm~$\PROJ$ is also a~\PROJ-row solution for
  node~$t$.
\end{lemma}
\begin{proof}[Idea]
  Observe that Listing~\ref{fig:dpontd3} computes a row for each
  sub-bucket $\sigma \in \subbuckets_P(\ATab{\AlgS}[t])$. The
  resulting row~$\langle\sigma, c \rangle$ obtained by~$\ipmc$ is
  indeed a \PROJ-row solution for~$t$ according to
  Lemma~\ref{lem:main_incl_excl}.
\end{proof}


\begin{lemma}[Completeness]\label{lem:complete}
  Let~$t\in N$ be node of the tree decomposition~$\TTT$ where
  $\type(t) \neq \leaf$.  Given a
  \PROJ-row solution~$\langle \sigma, c \rangle$ at node~$t$.
  There exists $\langle C_1, \ldots, C_\ell\rangle$ where $C_i$ is set
  of \PROJ-row solutions of the from~$\langle\sigma_i, c_i\rangle$
  such that
  $\sigma = \PROJ(t, \cdot, \cdot, P, \langle C_1, \ldots,
  C_\ell\rangle, \ATab{\AlgS})$.
\end{lemma}
\begin{proof}[Idea]
Since~$\langle\sigma,c \rangle$ is a~\PROJ-row solution for~$t$, there is by Definition~\ref{def:loctab} a corresponding ~\PROJ-solution~$\langle\hat\sigma\rangle$ up to~$t$ such that~$\local(t,\hat\sigma) = \sigma$. 

We proceed again by case distinction. Assume that~$\type(t)=\intr$. Then we define~$\hat{\sigma'}\eqdef \{(t',\hat\rho) \mid (t', \hat\rho)\in \sigma, t \neq t'\}$. Then, for each subset~$\emptyset\subsetneq\rho\subseteq\local(t',\hat{\sigma'})$, we define~$\langle \rho, \Card{I_P(\PExt_{\leq t}(\rho))}\rangle$ in accordance with Definition~\ref{def:loctab}. By Observation~\ref{obs:unique}, we have that~$\langle \rho, \Card{I_P(\PExt_{\leq t}(\rho))}\rangle$ is a \AlgS-row solution at node~$t'$. 
Since we defined the~\PROJ-row solutions for~$t'$ for all the respective \PROJ-solutions up to~$t'$, we encountered every~\PROJ-row solution for~$t'$ that is required for deriving~$\langle \sigma, c\rangle$ via~\PROJ (cf., Definitions~\ref{def:ipmc} and~\ref{def:pcnt}).

Assume that~$\type(t)=\rem$. The case is slightly easier as the one
above. We do not need to define a~\PROJ-row solution for~$t$' for all
subsets~$\rho$, since we only have to consider subsets~$\rho$ here,
with~$\Card{\buckets_P(\rho)}=1$. The remainder works similar.

Similarly, one can show the result for the remaining node with~$\type(t)=\join$, but define \PROJ-row solutions for two preceding child nodes of~$t$.
\end{proof}

We are now in the position to proof our theorem.

\begin{restatetheorem}[thm:correctness]%
\begin{theorem}
  The algorithm~$\dpa_\PROJ$ is correct. \\
  More precisely, 
  %
  the algorithm~$\dpa_\PROJ((F,P),\TTT,\ATab{\AlgS})$ returns
  tables~$\ATab{\PROJ}$ such that $c=\sipmc(\ATab{\AlgS}[n], \cdot)$
  is the projected model count of~$F$ with respect to the set~$P$ of
  projection variables.
\end{theorem}
\end{restatetheorem}
\begin{proof}
  %
  By Lemma~\ref{lem:correct} we have soundness for every
  node~$t \in N$ and hence only valid rows as output of table
  algorithm~$\PROJ$ when traversing the tree decomposition in
  post-order up to the root~$n$.
  By Lemma~\ref{lem:local} we know that the projected model count~$c$
  of~$F$ is larger than zero if and only if there exists a
  certain~\PROJ-row solution for~$n$.
  This~\PROJ-row solution at node~$n$ is of the
  form~$\langle \{\langle\emptyset, \ldots\rangle\} ,c\rangle$. If
  there is no \PROJ-row solution at node~$n$,
  then~$\ATab{\AlgS}[n]=\emptyset$ since the table algorithm~$\AlgS$
  is correct (cf., Proposition~\ref{prop:sat}). Consequently, we have
  $c=0$. Therefore, $c=\sipmc(\ATab{\AlgS}[n], \cdot)$ is the
  projected model count of~$F$ with respect to~$P$ in both cases.
  %
  %
  
  %
  %

  %
  %

  Next, we establish completeness by induction starting from the
  root~$n$. Let therefore, $\langle \hat\sigma \rangle$ be the~\PROJ-solution up to node~$n$, where for each row
  in~$\vec u\in \hat\sigma$, $I(\vec u)$ corresponds to a model of~$F$.  By
  Definition~\ref{def:loctab}, we know that for the root~$n$ we
  can construct a \PROJ-row solution at~$n$ of the
  form~$\langle \{\langle\emptyset, \ldots\rangle\} ,c\rangle$
  for~$\hat\sigma$.  We already established the induction step in
  Lemma~\ref{lem:complete}.
  Hence, we obtain some (corresponding) rows for every
  node~$t$. Finally, we stop at the leaves.

  In consequence, we have shown both soundness and completeness. As a
  result, Theorem~\ref{thm:correctness} is sustains.
\end{proof}

\begin{restatecorollary}[cor:correctness]%
\begin{corollary}
  The algorithm $\mdpa{\AlgS}$ is correct and outputs for any instance
  of \PMC its projected model count.
\end{corollary}
\end{restatecorollary}
\begin{proof}
  The result follows immediately, since~$\mdpa{\AlgS}$ consists of two
  dynamic programming passes~$\dpa_\AlgS$, a purging step and~$\dpa_\PROJ$. For the
  soundness and completeness of~$\dpa_\AlgS$ we refer to other
  sources~\cite{SamerSzeider10b,FichteEtAl17a}. By Proposition~\ref{prop:sat}, the
  ``purging'' step does neither destroy soundness nor completeness
  of~$\dpa_\PRIM$.
\end{proof}

\begin{corollary}
  Unless ETH fails, $\PMC$ cannot be solved in
  time~$2^{2^{o(k)}}\cdot \CCard{F}^{o(k)}$ for a given instance
  $(F,P)$ where~$k$ is the treewidth of the incidence graph of~$F$.
\end{corollary}
\begin{proof}
  Let $w_i$ and $w_p$ be the treewidth of the incidence graph and
  primal graph of~$F$, respectively. Then,
  $w_i \leq w_p +1$~\cite{SamerSzeider10b}, which establishes the
  claim.
\end{proof}}

}

\end{document}